\documentclass[5p,twocolumn]{elsarticle}
\usepackage{amsmath}
\usepackage{amssymb}
\usepackage{amsfonts}
\usepackage{mathtools}
\usepackage{forest}
\usepackage{stackengine}
\usepackage{relsize}
\usepackage{setspace}            
\usepackage{xy} 
\xyoption{all}

\newcommand*{\id}{{\mathrm{id}}}
\newcommand*{\R}{\mathbb R}

\newcommand*{\pa}{{\partial}}
\newcommand*{\rd}{{\mathrm{d}}}
\newcommand*{\lb}{[\![}
\newcommand*{\rb}{]\!]}
\newcommand*{\Vb}{\mathbb V}

\newcommand*{\tbf}{\boldsymbol{t}}

\newcommand*{\el}{\mathfrak l}
\newcommand*{\gl}{\mathfrak g}
\newcommand*{\rl}{\mathfrak r}
\newcommand*{\pf}{\mathfrak p}
\newcommand*{\qf}{\mathfrak q}
\newcommand*{\df}{\mathfrak d}
\newcommand*{\cendot}{\boldsymbol{\cdot}}
\newcommand*{\botimes}{\boldsymbol{\otimes}}
\newcommand*{\balpha}{\boldsymbol{\alpha}}
\newcommand*{\bbeta}{\boldsymbol{\beta}}
\newcommand*{\bgamma}{\boldsymbol{\gamma}}
\newcommand*{\bdelta}{\boldsymbol{\delta}}

\newtheorem{theorem}{Theorem}
\newtheorem{lemma}[theorem]{Lemma}
\newtheorem{corollary}[theorem]{Corollary}
\newtheorem{prop}[theorem]{Proposition}

\newdefinition{definition}{Definition}
\newdefinition{prescription}{Prescription}
\newdefinition{remark}{Remark}
\newdefinition{example}{Example}
\newdefinition{procedure}{Procedure}
\newproof{proof}{Proof}

\allowdisplaybreaks

\journal{...}

\begin{document}
\begin{frontmatter}
\title{The noncommutative KP hierarchy and its solution via descent algebra}
\author[1]{Gordon Blower}
\ead{G.Blower@lancaster.ac.uk}
\author[2]{Simon~J.A.~Malham\corref{cor1}}    
\ead{S.J.A.Malham@hw.ac.uk} 
\cortext[cor1]{Corresponding author} 
\address[1]{Department of Mathematics and Statistics, Lancaster University, Lancaster LA1~4YF, UK}
\address[2]{Maxwell Institute for Mathematical Sciences, and School of Mathematical and Computer Sciences,   
  Heriot-Watt University, Edinburgh EH14~4AS, UK} 

\begin{abstract}
  We give the solution to the complete noncommutative Kadomtsev--Petviashvili (KP) hierarchy.
  We achieve this via direct linearisation which involves the Gelfand--Levitan--Marchenko (GLM) equation.
  This is a linear integral equation in which the scattering data satisfies the linearised KP hierarchy. 
  The solution to the GLM equation is then shown to coincide with the solution to the noncommutative KP hierarchy.
  We achieve this using two approaches. In the first approach we use the standard Sato--Wilson dressing transformation.  
  In the second approach, which was pioneered by P\"oppe~\cite{PKP}, we assume the scattering data is semi-additive and by direct substitution,
  we show that the solution to the GLM equation satisfies the infinite set of field equations representing the noncommutative KP hierarchy.
  This approach relies on the augmented pre-P\"oppe algebra.
  This is a representative algebra that underlies the field equations representing the hierarchy.
  It is nonassociative and isomorphic to a descent algebra equipped with a grafting product.
  While we perform computations in the nonassociative descent algebra, the final result which establishes the solution to the complete hierarchy, resides in the natural associative subalgebra.
  The advantages of this second approach are that it is constructive, explicit, highlights the underlying combinatorial structures within the hierarchy,
  and reveals the mechanisms underlying the solution procedure. 
\end{abstract}
\end{frontmatter}

\section{Introduction}
Our objective is to solve the complete noncommutative Kadomtsev--Petviashvili (KP) hierarchy by direct linearisation. 
In this solution approach, we solve a linear integral equation known as the Gelfand--Levitan--Marchenko (GLM) equation.
The coefficients of the linear integral equation, or scattering data, are assumed to satisfy the linearised form of the KP hierarchy equations. 
The solution to this linear integral equation are then shown to satisfy the noncommutative KP hierarchy equations.
This method is called direct linearisation as the solution to the noncommutative KP hierarchy equations are obtained by successively solving two linear problems.
First, the linearised hierarchy equations are solved. Second, with that solution as coefficients, the linear integral equation is solved.
We thus demonstrate that the noncommutative KP hierarchy is integrable in the sense that the solution can be reconstructed from a pair of linear equations.

We achieve this via two approaches. The first approach uses the standard Sato--Wilson dressing transformation.
The second approach follows P\"oppe~\cite{PKP}, who pioneered a ``head-on'' direct linearisation approach for the KP equation.
We assume the scattering data is semi-additive, and gain a more explicit solution to the hierarchy.
This means that the first order linearised equation in the KP hierarchy is automatically satisfied; we naturally assume the scattering data satisfies the higher order linearised equations as well. 
As a result, the solution to the GLM equation and its derivatives satisfy the pre-P\"oppe product.
We also show that `higher' dependent variables in the hierarchy,
see equations \eqref{eq:gaugeexpansionPoppe}, \eqref{eq:SatoformPoppe} and \eqref{eq:BopsPoppe} just below,
are given by a simple formula in terms of the first dependent variable; thus we only need to focus on the latter.
Underlying the pre-P\"oppe product is the property that the kernels of certain linear combinations of products of the solution to the GLM equation and its derivatives,
correspond to the product of a pair of kernels.
These represent the natural noncommutative products of terms in the noncommutative KP hierarchy itself.
We construct the pre-P\"oppe algebra based on the pre-P\"oppe product.
The natural setting for the Sato formulation in this context, is the pre-P\"oppe algebra equipped with natural, closed, left and right actions, which generates a bimodule structure.
We augment the pre-P\"oppe product to mimic these actions and thus construct the `abstract' augmented pre-P\"oppe algebra that is isomorphic to the bimodule structure.
However, the augmented pre-P\"oppe algebra is nonassociative, reflecting that the left and right actions are distinct.
The original pre-P\"oppe algebra, which is associative, is a natural subalgebra of the augmented pre-P\"oppe algebra.
While we perform computations in the nonassociative augmented pre-P\"oppe algebra, at any point, we can map back to the bimodule structure and its associated actions,
that are closed in the associative pre-P\"oppe algebra.
In fact, we abstract the problem further, and show that the augmented pre-P\"oppe algebra is isomorphic to a descent algebra, equipped with a grafting product.
We perform our computations in the descent algebra context, and, consequently show that the solution to the GLM equation satisfies the noncommutative KP hierarchy.
From our perspective, the advantages of this second ``head-on'' P\"oppe approach to direct linearisation of the noncommutative KP hierarchy, are that it:
(i) Establishes explicit formulae at all stages for all the fields and coefficients involved;
(ii) Highlights the underlying, and hitherto hidden, algebraic combinatorial structures present in the hierarchy; and
(iii) Reveals the structures and mechanisms underlying the solution procedure that could, in principle, be adapted, extended and abstracted,
to establish the direct linearisation of more general integrable hierarchies. 

We consider herein, the Sato formulation of the noncommutative KP hierarchy. 
We seek a solution of the form,
\begin{equation}\label{eq:gaugeexpansionPoppe}
W=\id-\sum_{m\geqslant1}w_m\pa^{-m}, 
\end{equation}
to the system of equations,
\begin{equation}\label{eq:SatoformPoppe}
  \pa_{t_n}W=B_nW-W\pa^n,
\end{equation}
for all $n\in\mathbb N$, where the operator $B_n$ is given by,
\begin{equation}\label{eq:BopsPoppe}
  B_n\coloneqq\pa^n+\sum_{k=0}^{n-2}b_{k}^{n}(x,\tbf)\pa^k,
\end{equation}
where $\tbf\coloneqq(t_1,t_2,t_3,\ldots)$ and $\pa\coloneqq\pa_x$.
This generates polynomial relations between partial derivatives of the set of coefficients $w_k$ which constitute the hierarchy equations.
If we substitute the pseudo-differential operator expansion \eqref{eq:gaugeexpansionPoppe} into \eqref{eq:SatoformPoppe}, and equate the coefficients
of the positive and negative `powers' of $\pa_x$, the equations at the first three orders are as follows. For $n=1$, for all $k\in\mathbb N$, we find,
\begin{equation}\label{eq:x=t1}
\pa_{t_1}w_k=\pa_xw_k.
\end{equation}
We identify $t_1=x$. For $n=2$, for all $k\in\mathbb N$, we find, 
\begin{subequations}\label{eq:n=2relations}
\begin{align}
  b_{0}^{2}&=2\pa_xw_1,\label{eq:SatoVevol1}\\
  \pa_{t_2}w_k&=(\pa_x^2+b_{0}^{2})w_k+2\pa_xw_{k+1}.\label{eq:SatoVevol2}
\end{align}
\end{subequations}
It is standard to identify $t_2=y$. For $n=3$, we find,
\begin{subequations}
\begin{align}
  b_1^3=&\;3\pa_x w_1,\label{eq:n=3casebn2intro}\\
  b_0^3=&\;\bigl(3\pa_x^2+b_1^3\bigr)w_1+3\pa_xw_2,\label{eq:n=3casebn3intro}\\
  \pa_{t_3}w_k=&\;\bigl(\pa_x^3+b_1^3\,\pa_x+b_0^3\bigr)w_k\nonumber\\
              &\;+\bigl(3\pa_x^2+b_1^3\bigr)w_{k+1}+3\pa_xw_{k+2},\label{eq:n=3caseintro}
\end{align}
\end{subequations}
for all $k\in\mathbb N$. Equation~\eqref{eq:n=3caseintro} for $w_1$ is equivalent to the noncommutative KP equation.
Indeed, if we use the relations for $w_2$ and $w_3$ from \eqref{eq:n=2relations}, then for $w\coloneqq w_1$ we have,
\begin{equation}\label{eq:ncKP}
  \tfrac13(4\pa_tw-\pa_x^3w)=\pa_x^{-1}\pa_y^2w+2(\pa_xw)^2+2\pa_x^{-1}[\pa_yw,\pa_xw],
\end{equation}
where $[\cdot,\cdot]$ is the usual commutator and $t=t_3$.
This is the noncommutative KP equation; see for example Dimakis and M\"uller--Hoissen~\cite{DMH}, Gilson and Nimmo~\cite{GN}, Kupershmidt~\cite{K},
Hamanaka and Toda~\cite{HamanakaToda}, Wang and Wadati~\cite{WW} and Blower and Malham~\cite{BM-KP}.
It is common to rescale $t$ to render the coefficient of `$4$' to be unity.
Naturally we can proceed further and derive equivalent relations for $\pa_{t_n}w_k$ for all $n,k\in\mathbb N$.

The direct linearisation solution approach for the noncommutative KP hierarchy that we adopt herein has its origins in the work of P\"oppe and Nijhoff.
See: Nijhoff, Quispel, Van Der Linden and Capel~\cite{NQVC}; Nijhoff~\cite{Nijhoff}; Nijhoff and Capel~\cite{NC} and P\"oppe~\cite{PSG,PKdV,PKP}.
Our hypothesis that the scattering operator in the linear integral equation is semi-additive, is based on P\"oppe's approach for the commutative KP equation~\cite{PKP}.
In general, the direct linearisation of integrable nonlinear partial differential equations is classical, see for example:
Ablowitz, Ramani and Segur~\cite{ARS}; Dyson~\cite{Dyson}; Miura~\cite{Miura}; Santini, Ablowitz and Fokas~\cite{SAF} and Zakharov and Shabat~\cite{ZS,ZS2}.
For the Korteweg--de Vries equation, there are two versions of the linear integral equation in the literature.
However one is just the generalised Fourier transform of the other; see Blower and Malham~\cite[App.~D]{BM-KP}.
Direct linearisation for the KP equations has received much recent interest; see Fu~\cite{Fu} and Fu and Nijhoff~\cite{FN1,FN2,FN3}.
The development of a combinatorial algebra approach to solve noncommutative integrable systems based on P\"oppe's approach, and in the spirit of Manin~\cite{Manin}, 
is relatively recent.
See: Blower and Malham~\cite{BM,BM-KP}; Doikou \textit{et al.\/}~\cite{DMS}; Doikou \textit{et al.\/}~\cite{DMSWc} and Malham~\cite{MNLS,MKdV}.
Therein, this approach was applied to the noncommutative, KdV, and nonlinear Schr\"odinger (NLS) and modified KdV, hierarchies, as well as the noncommutative KP and mKP equations.
Importantly, and as a reference point for our work herein, one of the first papers to establish integrability for the complete KP hierarchy was Mulase~\cite{Mu84}.
Therein, integrability was established in the sense of Frobenius, by transforming the KP hierarchy to a linear total differential equation, involving infinitely many variables.

We now outline the direct linearisation procedure for the noncommutative KP hierarchy more precisely. 
The goal is to solve all the equations in the hierarchy.
To solve the infinite set of nonlinear partial differential equations (PDEs) constituting the noncommutative KP hierarchy,
we solve the following two linear problems in succession:

(1) \emph{Linearised PDE:} First, we solve the linearised form of the KP equations for $p=p(z,\zeta;t_1,t_2,t_3,\ldots)$, namely:
\begin{equation}\label{eq:linearform}
  \pa_{t_n}p=\pa_z^np-(-1)^n\pa_\zeta^np.
\end{equation}
for $n\in\mathbb N$. These are known as the \emph{base equations}. 

(2) \emph{Linear integral equation:} Second, we solve a linear integral equation of the form,
\begin{equation}\label{eq:GLMeqintro}
  P=G(\id-P),
\end{equation}
for the operator $G$, or equivalently, its kernel, $g$.
Here, $P$ is the scattering operator associated with the kernel function solution $p$ from \eqref{eq:linearform}.
This is also known as the \emph{Gelfand--Levitan--Marchenko (GLM) equation}.  

Our goal herein, is to show that $G$ generates the solution to the complete noncommutative KP hierarchy.
In P\"oppe's approach, we assume $P$ is a Hilbert--Schmidt and semi-additive operator family.
This means that its kernel $p$ is square-integrable and, in particular, has the form,
\begin{equation*}
p=p(z+x,\zeta+x;\tbf),
\end{equation*}
where $z,\zeta\in(-\infty,0]$ are the primary variables parametrising the operator kernel, and we consider $x\in\R$ and $\tbf$ as additional parameters.
Note, this means that the first equation in \eqref{eq:linearform}, namely $\pa_{t_1}p=\pa_z p+\pa_\zeta p$, is automatically satisfied, and $t_1=x$.
Thus, for this form we suppose $\tbf=(t_2,t_3,\ldots)$.
Herein we use the notion of left, $\pa_{\el}$, and right, $\pa_{\rl}$, partial derivatives which respectively correspond to $\pa_z$ and $\pa_\zeta$.
One of our first results establishes that the dependent variables $w_k$ in the hierarchy are signed $\pa_{\rl}^{k-1}$ derivatives of $g$
and thus we need only focus on solving the nonlinear PDE prescribing $g=w_1$; see Section~\ref{sec:SatoPoppeconnection}.
If we set $V\coloneqq(\id-P)^{-1}$, then the solution $G$ to the GLM equation~\eqref{eq:GLMeqintro} is given by, 
\begin{equation*}
G=PV.
\end{equation*}
Under mild assumptions, we can establish, $G$ is Hilbert--Schmidt valued, with kernel $g=g(z,\zeta;x,\tbf)$. 
Herein, following P\"oppe~\cite{PKP}, we use the notation $\lb G\rb$ to denote the kernel $g$ of $G$, i.e.\/ we set,
\begin{equation}\label{eq:bracketnotation}
\lb G\rb(z,\zeta;x,\tbf)\coloneqq g(z,\zeta;x,\tbf).
\end{equation}
Our goal now, under the assumption $p$ satisfies the linear PDE \eqref{eq:linearform},
is to show $\lb G\rb$ satisfies the noncommutative KP hierarchy equations for $w_1$.
We substitute $w_1=-\lb G\rb$ with $G=PV$ into the righthand side of the equation for $\pa_{t_n}w_1$  
and show the terms collapse to,
\begin{equation*}
  \pa_{t_n}\lb V\rb=\lb V(\pa_{\el}^nP-(-1)^n\pa_{\rl}^nP)V\rb.
\end{equation*}
Here we used the partial fraction formulae $V=\id-PV=\id-VP$ and that $\pa_{t_n}G=\pa_{t_n}V=V(\pa_{t_n}P)V$.
Substituting $w_1=-\lb G\rb$ into the equation for $\pa_{t_n}w_1$ means that we need to compute $\pa_x$ derivatives
of $\lb G\rb$ and products of such derivatives. The former generate linear combinations of monomials of the form,
\begin{equation}\label{eq:exmonomial}
\lb VP_{a_1,b_1}VP_{a_2,b_2}V\cdots VP_{a_k,b_k}V\rb,
\end{equation}
where the subindicies $a_1$ and $b_1$ respectively denote the number of $\pa_{\el}$ and $\pa_{\rl}$ derivatives of $P$, and so forth.
The product of such monomials can be computed via the pre-P\"oppe product; see Lemma~\ref{lemma:Poppeprodmonomials}.
This product crucially depends on our semi-additive assumption for $P$.
The product of a monomial, say~\eqref{eq:exmonomial}, with another, say $\lb VP_{a_1^\prime,b_1^\prime}V\cdots VP_{a_\kappa^\prime,b_\kappa^\prime}V\rb$, 
generates the terms,
\begin{align}
   &\lb VP_{a_1,b_1}V\cdot\!\cdot\!\cdot V(\pa_{\rl}P_{a_k,b_k})VP_{a_1^\prime,b_1^\prime}V\cdot\!\cdot\!\cdot VP_{a_\kappa^\prime,b_\kappa^\prime}V\rb\nonumber\\
  +&\lb VP_{a_1,b_1}V\cdot\!\cdot\!\cdot VP_{a_k,b_k}V(\pa_{\el}P_{a_1^\prime,b_1^\prime})V\cdot\!\cdot\!\cdot VP_{a_\kappa^\prime,b_\kappa^\prime}V\rb\nonumber\\
  +&\lb VP_{a_1,b_1}V\cdot\!\cdot\!\cdot VP_{a_k,b_k}VP_{\hat{1}}VP_{a_1^\prime,b_1^\prime}V\cdot\!\cdot\!\cdot VP_{a_\kappa^\prime,b_\kappa^\prime}V\rb,\label{eq:Poppeprodintro}
\end{align}
where $P_{\hat{1}}\coloneqq \pa_{\el}P+\pa_{\rl}P$.
In the above, we preclude the cases when either of the factors are simply $\lb V\rb$---note that $V$ is not a Hilbert--Schmidt operator. 
With this product we can define the pre-P\"oppe algebra which is associative; see Blower and Malham~\cite{BM-KP}.
However, the term $b_0^nw_1$ present in \eqref{eq:SatoVevol2} and then \eqref{eq:n=3casebn3intro} and \eqref{eq:n=3caseintro}, for example, requires us to include the right factor, $\lb G\rb$, with $G=VP$.
This is not one of the monomial forms \eqref{eq:exmonomial} which represent a basis for the pre-P\"oppe algebra.
Further, the terms $w_k$ for $k\geqslant2$ involve $\pa_{\rl}$ derivatives $\lb G\rb$, also do not generate such monomials. 
We can compute them, but in terms that are outwith the monomial basis terms for the pre-P\"oppe algebra. 
We could extend our algebra to include such terms.
However, fortunately, within the prescription of the Sato formulation for the noncommutative KP hierarchy, these two aspects can be naturally combined and encoded via a right action operation.
Indeed we can show that, $\pa_{\rl}$ differentiation plus right multiplication by $\lb G\rb$, transports elements within the pre-P\"oppe algebra.
This suggests we naturally incorporate such a right action within a pre-P\"oppe algebra framework, which indeed we do, though we also incorporate a corresponding left action.
Indeed we construct a bimodule structure on the pre-P\"oppe algebra base on these actions in Section~\ref{sec:steptodescents}.
Recall we precluded the factor $\lb V\rb$ in the product above---from a functional analytic perspective the product does not apply with this factor.
However, we can define an abstract \emph{augmented pre-P\"oppe algebra}, in which we augment the product above to include the cases of left and right $\lb V\rb$ factors.
The two additional cases involving left and right factors, $\lb V\rb$, correspond precisely to the natural specification we would anticipate from the pre-P\"oppe product above,
and, correspond precisely, to the left and right actions, respectively, of the bimodule structure.
We show that the bimodule structure and the augmented pre-P\"oppe product are isomorphic.
However, the augmented pre-P\"oppe algebra is \emph{nonassociative} which reflects the fact that the left and right actions are distinct.
The pre-P\"oppe algebra is a subalgebra of the augmented pre-P\"oppe algebra.
Indeed, it's the subalgebra in which we exclude the element $\lb V\rb$.
Also recall that the left and right actions are closed in the pre-P\"oppe algebra.
Hence any computations we perform in the augmented pre-P\"oppe algebra, can be mapped back via the bimodule structure to the pre-P\"oppe subalgebra, which is associative.  

The augmented pre-P\"oppe algebra is isomorphic to the descent algebra equipped with a grafting product.
This is not immediately obvious, and we establish this result in Sections~\ref{sec:trees}--\ref{sec:descent-alg}, wherein we provide many examples.
We first introduce the algebra of planar binary rooted trees equipped with grafting at the root as the product.
This is the archetypal nonassociative algebra.
Indeed it is then natural to define the descent algebra as the algebra of planar binary rooted trees in which we equivalence by trees whose word codings have the same descent set.
The descent grafting product is the natural product induced from the binary tree grafting product and represents/mimics the augmented pre-P\"oppe product in this context.
To show that the augmented pre-P\"oppe algebra is isomorphic to the descent algebra we utilise and representation for the descent algebra, which we denote as the left-glue-right or $\el\gl\rl$-algebra.
This algebra records all the $\pa_{\el}$ and $\pa_{\rl}$ operations generated by augmented pre-P\"oppe products, such as in the first and second terms in the pre-P\"oppe product \eqref{eq:Poppeprodintro} above,
as well as the `glue' terms, such as the final term in the pre-P\"oppe product \eqref{eq:Poppeprodintro} above.
This algebra provides the natural connection between the augmented pre-P\"oppe and descent algebras.
Subsequently, in Sections~\ref{sec:degrafting}--\ref{sec:solutions}, we introduce a degrafting operation, that helps to encode the necessary descent expansions in the Sato formulation of
the noncommutative hierarchy, and then directly show that $\lb G\rb$ solves the complete hierarchy.
Strictly speaking we show that $\lb G\rb=\lb G\rb(0,0;x,\tbf)$ satisfies the complete noncommutative KP hierarchy; see Remark~\ref{rmk:multifactor} and Theorem~\ref{thm:gaugecoeffsid}.
Within those sections we prove formulae for the so-called weight of a descent that naturally arises in the Sato formulation in the descent setting,
as well as explicit formulae for the Sato coefficients $b_k$.
These require further combinatorial algebraic structures and non-trivial binomial coefficient identities.
We finally finish the proof of the overall result, using the $\el\gl\rl$-algebra, in Section~\ref{sec:solutions}. 

Some further qualifications of our statements are required.
In principle, we have only identified herein, solutions to the complete noncommutative hierarchy that are square-integrable---with the assumption that
the solution $p$ to the linearised equations~\eqref{eq:linearform} is sufficiently smooth.  
There may be weaker classes of solutions that are not included here, for example, the classes of solutions for the Korteweg--de Vries equation considered by Grudsky and Rybkin~\cite{GR}. 
However, the solution form we generate via either of the direct linearisation approaches considered herein, is both constructive and practical.
Indeed, for example, we can numerically generate solutions to the noncommutative KP hierarchy straightforwardly as follows.
We can evaluate the solution $p$ to the linearised equations~\eqref{eq:linearform} at anytime $t>0$ directly in Fourier space,
and substitute the inverse Fourier transform of the solution into the GLM equation~\eqref{eq:GLMeqintro}, which we can solve as a numerical linear algebra problem.
This generates the solution to noncommutative KP hierarchy at any time $t>0$.
See Blower and Malham~\cite{BM-KP} where we explicitly generate noncommutative solutions to the KP equation in this way. 

For convenience for the reader, we summarise the main combinatorial algebras and structures that we require herein. These are the:
(a) \emph{Sato algebra:} This the algebra of pseudo-differential operators, with composition as product.
This algebra provides the context for generating the noncommutative KP hierarchy;
(b) \emph{Pre-P\"oppe algebra:} This is the algebra of smooth kernels of Hilbert--Schmidt operators on the half line. 
The product in this algebra is the pre-P\"oppe product based on semi-additive operators, as described above; 
(c) \emph{Action bimodule:} This is a bimodule constructed by equipping the pre-P\"oppe algebra with some specific actions that arise in the Sato formulation of the noncommutative KP hierarchy.
It represents the natural context for the Sato formulation in the P\"oppe approach;
(d) \emph{Augmented pre-P\"oppe algebra:}
This is an abstract version of the pre-P\"oppe algebra where we augment the pre-P\"oppe product to include some basic products that naturally mimic the actions in the bimodule.
It is naturally nonassociative and isomorphic to the action bimodule. The pre-P\"oppe algebra is a natural subalgebra;
(e) \emph{Descent algebra:} This is the algebra of descents, equipped with the so-called grafting product.
There are two main representations of the descent algebra that we utilise herein, these are the:
(i) \emph{Standard representation:} We record the descents associated with the standard encoding of planar binary trees, equivalencing by the sets of trees with the same descent sets; and
(ii) \emph{Left-glue-right representation:} This is an isomorphic representation that reveals that the descent algebra and augmented pre-P\"oppe algebra are isomorphic.
Hereafter we use `$\el\gl\rl$' as shorthand for `left-glue-right'.    
The descent algebra is usually associated with the Solomon Descent Algebra, which is associated with permutations and/or shuffles as well as the algebra of quasi-symmetric functions~\cite{MR,FMP}.

We now summarise what is new in this paper. In the noncommutative case, we establish:
\begin{enumerate}
\item[(1)] The connection between the Sato formulation for the KP hierarchy via a gauge transformation, and, the GLM equation with semi-additive scattering data; 
\item[(2)] An augmented pre-P\"oppe algebra. This is the underlying algebra for the Sato formulation in the context of direct linearisation. It is nonassociative;
\item[(3)] The isomorphism between the augmented pre-P\"oppe algebra and a descent algebra equipped with a grafting product;
\item[(4)] That the Sato coefficients have a simple explicit expansion formula in descents; 
\item[(5)] That the complete KP hierarchy is integrable via direct linearisation, by using (1)--(4).
\end{enumerate}

The KP equation has an illustrious history and many connections to parallel branches of mathematics and theoretical physics.
It was originally derived as a natural generalisation of the KdV equation to two dimensions by Kadomtsev and Petviashvili~\cite{KadomtsevPetviashvili}, with applications to long waves in shallow water.
However it has many further applications, including, to nonlinear optics (Pelinovsky, Stepanyants and Kivshar~\cite{PSK}) and to ferromagnetism and Bose--Einstein condensates.
Following the characterisation by Dyson~\cite{Dyson} of KdV solutions in terms of Fredholm determinants, in the nineteen eighties,
Sato~\cite{SatoI,SatoII}, Miwa, Jimbo and Date~\cite{MJD}, Hirota~\cite{Hirota} and Segal and Wilson~\cite{SW} demonstrated how solutions to the KdV and KP hierarchies 
could be formulated in terms of Fredholm Grassmannians and their associated determinantal bundle.
Simultaneously connections between the KP equation, Jacobians of algebraic curves and theta functions were also made, see Mulase~\cite{Mu94} and Mumford~\cite{Mumford}.
Also see: Ball and Vinnikov~\cite{BV}; Ercolani and McKean~\cite{EM} and McKean~\cite{Mc87}.
In the nineteen nineties Bourgain~\cite{Bourgain} established global well-posedness for periodic square-integrable solutions to the KP equation, given such data, as well as for smooth solutions for smoother data.
However, there are also rational solutions, see for example, Kasman~\cite{Kasman}, Pelinovsky~\cite{Pelinovsky} and P\"oppe~\cite{PKP}.
Kodama~\cite{Kodama} also completed the classification of all possible soliton interactions, paramaterised via a finite-dimensional Grassmannians. 
Recent results have established many more connections between the KP equation and other fields.
See, for example, Bertola, Grava and Orsatti~\cite{BGO}, who establish a connection between the $\tau$-functions of the KP and nonlinear Schr\"odinger hierarchies.
Or for another example, see Quastel and Remenik~\cite{QR}, who establish a connection between the KP equation and the KPZ equation for the random growth off a one-dimensional substrate.
In this direction, also see: McKean~\cite{McKean11}; Tracy and Widom \cite{TW94,TW96,TW03} and Zhang~\cite{Zhang}. 
In nineteen ninety, Witten~\cite{Witten} conjectured a connection between partition functions in string theory and the $\tau$-function of the KdV hierarchy, which was proved by Kontsevich~\cite{Kontsevich}.
Also see Cafasso and Wu~\cite{CafassoWu}. The connection between the noncommutative KP equation and D-branes was studied in Paniak~\cite{Paniak}. 
Also see Hamanaka~\cite{Hamanaka06,Hamanaka} and Harvey~\cite{Harvey}.
Herein, due to quantisation of the phase space or a noncommutative geometry interpretation, noncommutativity refers to when the independent coordinates, such as $x$, are noncommutative.
However, often, via the `Moyal product', the derivatives with respect to the independent coordinates can be interpreted classically, but the product
between the dependent field and its derivatives becomes noncommutative.
This is the interpretation adopted by many authors, which we adopt herein. 
See, for example: Gilson and Nimmo~\cite{GN}; Hamanaka~\cite{Hamanaka}; Hamanaka and Toda~\cite{HamanakaToda}; Koikawa~\cite{Koikawa}; Toda~\cite{Toda} and Wang and Wadati~\cite{WW}.
Noncommutativity of the dependent variables is often referred to as the `nonabelian' context, see Nijhoff~\cite{Nijhoff-Lagrangian3form}.
There are also many applications of noncommutative integrable systems, for example, boomeron, trappon and simulton solutions in noncommutative nonlinear optics;
see Calogero and Degasperis~\cite{CD}, Degasperis and Lombardo~\cite{DL}. 
Lastly, multi-soliton and $\tau$-function solutions to the noncommutative KP hierarchy can be characterised in terms of quasi-determinants,
see: Etingof, Gelfand and Retakh~\cite{EGR97,EGR98}; Gilson and Nimmo~\cite{GN}; Hamanaka~\cite{Hamanaka} and Sooman~\cite{So}.

Our paper is structured as follows. 
In Section~\ref{sec:Poppealg} we introduce the pre-P\"oppe algebra that underlies P\"oppe's approach to direct linearisation of the KP hierarchy via semi-additive scattering data.
We give some background material demonstrating the connection between the Lax, Zakharov--Shabat and Sato formulations for the noncommutative KP hierarchy in Section~\ref{sec:Satoindirect},
and introduce the Sato--Wilson dressing transformation approach to direct linearisation.
We then establish the connection between the solution to the GLM equation in P\"oppe's approach and the Sato formulation for the noncommutative KP hierarchy in Section~\ref{sec:SatoPoppeconnection}.
We also establish direct linearisation of the noncommutative KP hierarchy via the Sato--Wilson dressing transformation approach, that generates the solution fields to the whole hierarchy.
Section~\ref{sec:hierarchyequations} outlines the explicit formulae for the Sato coefficients from which the noncommutative KP hierarchy equations are constructed in the context
of P\"oppe's direct linearisation approach.
The goal of Section~\ref{sec:steptodescents} is to explain how we make the step, from the Sato formulation of the noncommutative KP hierarchy,
to the nonassociative noncommutative augmented pre-P\"oppe algebra, which we use to solve the hierarchy by direct linearisation. 
In Section~\ref{sec:trees}, we review the algebra of planar binary rooted trees. This is a more general nonassociative algebra than we need.
However, the concepts we introduce in that context, such as the grafting of trees at their roots, are the ones we naturally induce on the  
$\el\gl\rl$-algebra, which we discuss in Section~\ref{sec:lgr-alg}, and the descent algebra, which we discuss in Section~\ref{sec:descent-alg}.
The latter two algebras, and the augmented pre-P\"oppe algebra, are all isomorphic.
The $\el\gl\rl$-algebra is a natural representation for the pre-P\"oppe algebra. The descent algebra is a useful abstract representation.
We introduce the degrafting operator in Section~\ref{sec:degrafting} and establish and explicit formula for the weight character of a descent in Section~\ref{sec:weightchar}.
These are both crucial for establishing the simple explicit expansion formula for the Sato coefficients that we provide in Section~\ref{sec:Satocoeffs}.  
Finally in Section~\ref{sec:solutions} we show that the solution to the GLM equation with semi-additive scattering data that satisfies the linearised equations,
solves the corresponding noncommutative KP hierarchy equations.
Lastly, in Section~\ref{sec:discussion} we discuss possible future research directions.

\section{Pre-P\"oppe algebra}\label{sec:Poppealg}
Herein we introduce the pre-P\"oppe algebra. 
Let $\mathbb V\coloneqq L^2((-\infty,0];\R^m)$ denote the space of square-integrable $\R^m$-valued functions on $(-\infty,0]$,
and $\mathfrak J_2=\mathfrak J_2(\Vb)$ denote the space of Hilbert--Schmidt operators on $\mathbb V$.
Any Hilbert--Schmidt valued operator $G$ on $\Vb$ generates a unique kernel function $g\in L^2((-\infty,0]^{\times2};\R^{m\times m})$
such that for all $\phi\in\Vb$ and $z\in(-\infty,0]$,
\begin{equation*}
\bigl(G\phi\bigr)(z)=\int_{-\infty}^0 g(z,\zeta)\,\phi(\zeta)\,\rd\zeta.
\end{equation*}
Conversely, any $g\in L^2((-\infty,0]^{\times2};\R^{m\times m})$ generates a Hilbert--Schmidt valued operator $G$ on $\Vb$,
with $\|G\|_{\mathfrak J_2(\Vb)}=\|g\|_{L^2((-\infty,0]^{\times2};\R^{m\times m})}$. See, for example, Simon~\cite{Simon}.
\begin{definition}[Left and right derivatives]
  Consider a Hilbert--Schmidt valued operator $G$ on $\Vb$, with a continuously differentiable kernel $g=g(z,\zeta)$.
  Further, assume the partial derivatives, $\pa_z g$ and $\pa_\zeta g$ of $g$, are square-integrable on $(-\infty,0]^2$.
  We define the left $\pa_\el$ and right $\pa_\rl$ partial derivatives of $G$, respectively,
  to be the Hilbert--Schmidt valued operators corresponding to the respective kernels $\pa_z g$ and $\pa_\zeta g$.
\end{definition}
\begin{remark}
  Note that for any pair of Hilbert--Schmidt valued operators, $G$ and $\hat{G}$, with continuously differentiable kernels, we have (naturally $\pa_\el$ and $\pa_\rl$ commute),
  \begin{equation}
   \pa_\el(G\hat{G})=(\pa_\el G)\hat{G}\quad\text{and}\quad \pa_\rl(G\hat{G})=G(\pa_\rl\hat{G}).\label{eq:leftandrightderivs}
  \end{equation}
\end{remark}
\begin{definition}[Semi-additive operator]\label{def:semiadd}
  A Hilbert--Schmidt operator $P$ on $\Vb$ with corresponding square-integrable kernel $p$ is \emph{semi-additive with parameters} $x$ and $\tbf$
  if its operation for any square-integrable function $\phi\in\Vb$, with $z\in(-\infty,0]$, is,
  \begin{equation*}
     \bigl(P\phi\bigr)(z;x,\tbf)=\int_{-\infty}^0 p(z+x,\zeta+x;\tbf)\,\phi(\zeta)\,\rd\zeta.
  \end{equation*}
\end{definition}
Note that for any such semi-additive operator, we have
\begin{equation*}
  \pa_xP=(\pa_\el+\pa_{\rl})P.  
\end{equation*}
\begin{remark}
  The parameter $\tbf$ is an infinite collection of parameters $\tbf=(t_2,t_3,\ldots)$.
  As will become apparent, we usually denote $t_1=x$.
  For the Hilbert--Schmidt valued semi-additive operator $P$ with kernel $p=p(z+x,\zeta+x;\tbf)$,
  we consider $z$ and $\zeta$ to be the primary kernel variables, while $x$ and $\tbf$ are additional parameters.
\end{remark}
Recall the bracket `$\lb\cdot\rb$' notation in \eqref{eq:bracketnotation} for Hilbert--Schmidt valued kernels from the Introduction.
For any Hilbert--Schmidt operator $G$, we denote the kernel $g=g(z,\zeta)$ of $G$ by $\lb G\rb=\lb G\rb(z,\zeta)$, i.e.\/ $\lb G\rb\coloneqq g$.
Crucial to our results herein is the pre-P\"oppe product. 
\begin{lemma}[Pre-P\"oppe product]\label{lemma:Poppeprodsemiadd}
  Assume that $P$ and $\hat P$ are semi-additive Hilbert--Schmidt operators with parameters $x$ and $\tbf$, and continuously differentiable kernels.
  Further assume that $G$ and $\hat G$ are Hilbert--Schmidt operators with continuous kernels. Then we have the following pre-P\"oppe product rule:
  \begin{multline*}
    \lb G(P_{0,1}\hat P\!+\!P\hat{P}_{1,0})\hat{G}\rb(z,\zeta;x,\tbf)\\=\lb GP\rb(z,0;x,\tbf)\,\lb\hat{P}\hat{G}\rb(0,\zeta;x,\tbf).
  \end{multline*}
\end{lemma}
This result is straightforwardly established using the Fundamental Theorem of Calculus.

We now consider the GLM equation \eqref{eq:GLMeqintro}, i.e.\/ the linear integral equation $P=G(\id-P)$.
Suppose we set, 
\begin{equation}\label{eq:defV}
V\coloneqq (\id-P)^{-1}.
\end{equation}
Then the solution $G$ to the GLM equation is given by, 
\begin{equation}\label{eq:defG}
G=PV.
\end{equation}
Existence and uniqueness of such a solution is given, provided $P$ is a Hilbert--Schmidt operator and $\mathrm{det}_2(\id-P)\neq0$,
where $\mathrm{det}_2$ is Carleman's determinant.
For more details, see Blower and Malham~\cite[Lemma~7]{BM-KP}.
Herein, we assume the operator $P$ is semi-additive with parameters $x$ and $\tbf$.
We are concerned the kernels $\lb G\rb$ and $\lb V\rb$ and their left $\pa_\el$, right $\pa_\rl$, and $\pa_x$ partial derivatives.
Note that by partial fractions, we have,
\begin{equation}\label{eq:partialfractions}
V\equiv\id+PV\equiv\id+VP,
\end{equation}
and thus $\pa_x$ derivatives of $\lb G\rb$ and $\lb V\rb$ are identical.
Let $\mathcal C(n)$ denote the set of all compositions of $n$. 
Hereafter we use the notation, $P_{\hat a}\coloneqq\pa_x^aP$.
We assume $P$ is continuously differentiable to any order we require, and its partial derivatives at that order are Hilbert--Schmidt valued.
By successively differentiating $V$ with respect to $x$, and using the partial fraction formulae~\eqref{eq:partialfractions}, and then applying the
kernel bracket operator `$\lb\,\cdot\,\rb$', we observe,
\begin{equation}\label{eq:expansion}
\pa_x^n\lb V\rb=\sum \chi(a_1\cdots a_k)\cdot\lb VP_{\hat{a}_1}V\cdots VP_{\hat{a}_k}V\rb.
\end{equation}
The sum is over all compositions $a_1a_2\cdots a_k\in\mathcal C(n)$, and, 
\begin{equation*}
  \chi(a_1a_2\cdots a_k)\coloneqq\prod_{\ell=1}^kC_{a_\ell}^{a_\ell+\cdots+a_k}.
\end{equation*}
Here $C_\ell^n$ is the binomial coefficient $n$ choose $\ell$. We now lean on the fact that $P$ is semi-additive.
We observe that, by the chain rule, we have that, $P_{\hat a}=(\pa_\el+\pa_{\rl})^aP$ for any $a\in\mathbb N$.
Hence $P_{\hat a}$ is represented by, 
\begin{equation}\label{eq:Leibnizexpansionorig}
P_{\hat a }=\sum_{k=0}^a C_{k}^a\,\pa_\el^k\pa_{\rl}^{a-k}P.
\end{equation}
The terms in the sum are indexed by pairs $(k,a-k)$ of left and right derivatives, respectively.
We can replace each of the terms $P_{\hat{a}_1}$, \ldots, $P_{\hat{a}_k}$ in~\eqref{eq:expansion} by a Leibniz sum of the form \eqref{eq:Leibnizexpansionorig},
which involves pairs of left and right derivatives of the form $P_{a,b}$.
Hence we can regard the base monomials in~\eqref{eq:expansion} to be of the form \eqref{eq:exmonomial}.
The following formulation of the pre-P\"oppe product formula in Lemma~\ref{lemma:Poppeprodsemiadd} is crucial to our construction of the pre-P\"oppe algebra. 
\begin{lemma}\label{lemma:Poppeprodmonomials}
  For arbitrary Hilbert--Schmidt operators $F$ and $\hat{F}$ and a semi-additive Hilbert--Schmidt operator $P$
  with parameters $x,y$ and a smooth kernel, and for any $a,b,c,d\in\mathbb N\cup\{0\}$, we have,
  \begin{align*}
  \lb FP_{a,b} V\rb\,\lb VP_{c,d}\hat{F}\rb=&\;\lb F(\pa_{\rl}P_{a,b})VP_{c,d}\hat{F}\rb\\
                                              &\;+\lb FP_{a,b} V(\pa_{\el}P_{c,d})\hat{F}\rb\\
                                              &\;+\lb FP_{a,b} VP_{\hat{1}} VP_{c,d}\hat{F}\rb.
  \end{align*}
\end{lemma}
\begin{remark}
  Note that the cases $FP_{a,b}=\id$ and $P_{c,d}\hat F=\id$ are naturally precluded as $F$, $P$ and $\hat{F}$ are all assumed to be Hilbert--Schmidt valued.
  In particular, the possibility of a pre- or postfactor being `$\lb V\rb$', is also precluded.
\end{remark}
The result of Lemma~\ref{lemma:Poppeprodmonomials} is straightforwardly established by systematically combining the partial fraction formulae~\eqref{eq:partialfractions}
with the pre-P\"oppe product rule in Lemma~\ref{lemma:Poppeprodsemiadd}.
\begin{definition}[Pre-P\"oppe algebra]\label{def:prePoppealg}
  We call the real matrix algebra of monomials of the form~\eqref{eq:exmonomial}, equipped with the pre-P\"oppe product in Lemma~\ref{lemma:Poppeprodmonomials}, the \emph{pre-P\"oppe algebra}.
  The monomial $\lb V\rb$ is \emph{excluded} from the algebra. 
\end{definition}
\begin{remark}
   We do not require the pre-P\"oppe algebra to be unital. And though `$\lb V\rb$' is excluded from the algbera, the monomials $\lb VP_{1,0}V\rb$ and $\lb VP_{0,1}V\rb$ are included.
\end{remark}
\begin{remark}\label{rmk:multifactor}
Some linear combinations of monomials of the form \eqref{eq:exmonomial} generate multi-factor products. Such multi-factor products are necessarily of the form,
\begin{equation*}
\lb\cdot\rb(z,0;x,\tbf)\lb\cdot\rb(0,0;x,\tbf)\cdots\lb\cdot\rb(0,0;x,\tbf)\lb\cdot\rb(0,\zeta;x,\tbf),
\end{equation*}
which hereafter, we represent by $\lb\cdot\rb\lb\cdot\rb\cdots\lb\cdot\rb\lb\cdot\rb$.
Further, though we have some leeway to leave $z$ and $\zeta$ `free', strictly speaking, we suppose the monomials \eqref{eq:exmonomial} in the pre-P\"oppe algebra are all evaluated at $z=\zeta=0$.
The pre-P\"oppe algebra remains consistent.
\end{remark}

\section{Lax, Zakharov--Shabat and Sato formulations}\label{sec:Satoindirect}
We outline the Lax, Zakharov--Shabat and Sato formulations for the noncommutative KP hierarchy. Modulo some constraints outlined, these are all essentially equivalent.
We also introduce the Sato--Wilson dressing transformation and demonstrate direct linearisation of the noncommutative KP hierarchy using this approach.
Much of the material herein is based on, or derived from, Kodama~\cite[Ch.~2]{Kodama} and P\"oppe and Sattinger~\cite{PoppeSattinger}.

Let us begin with the Lax formulation. Suppose $L$ is the pseudo-differential operator defined by,  
\begin{equation}\label{eq:L}
L\coloneqq \pa+u_1\pa^{-1}+u_2\pa^{-2}+\cdots,
\end{equation}
where $u=u_i(x,\tbf)$, with $x\in\R$, and where $\tbf$ represents the set of variables $\tbf=(t_1,t_2,t_3,\cdots)$, and $\pa\coloneqq\pa_x$.
The operator $\pa^{-1}$ represents the formal inverse of the operator $\pa$.
For $n\in\mathbb Z$, the usual Leibniz rule applies:
\begin{equation*}
\pa^n u=\sum_{k\geqslant0}C^n_k(\pa_x^k u)\pa^{n-k}.
\end{equation*}
For example, we have $\pa u=\pa_xu+u\pa$ and also that $\pa^{-1}u=u\pa^{-1}-u_x\pa^{-2}+u_{xx}\pa^{-3}-\cdots$, and so forth.
\begin{definition}[Lax formulation]\label{def:Laxform}
With $B_n\coloneqq (L^n)_+$, the \emph{Lax formulation} of the noncommutative KP hierarchy is given by the infinite set of equations for all $n\in\mathbb N$:
\begin{equation}\label{eq:KPLaxform}
\pa_{t_n}L=[B_n,L].
\end{equation}
Here, $(L^n)_+$ represents the polynomial part of $L^n$ so that $(L^n)_+$ is a pure differential operator, and $[B_n,L]\coloneqq B_n L-L B_n$ is the standard commutator for operators.
\end{definition}
The case $n=1$ in the definition above generates the equations, $\pa_{t_1}u_k=\pa_x u_k$, for all $k\in\mathbb N$.
Thus, as mentioned in the Introduction, we identify $t_1$ and $x$.
The cases $n=2,3,\ldots$ and so forth generate a coupled infinite set of evolution equations for the $u_k$ in $t_n$.
\begin{definition}[Zakharov--Shabat formulation]\label{def:ZSform}
  For any $n\in\mathbb N$, consider the pure differential operators in \eqref{eq:BopsPoppe}.
  The Zakharov--Shabat formulation of the noncommutative KP hierarchy is then given by the following system of equations for all $n,m\in\mathbb N$: $\pa_{t_n}B_m-\pa_{t_m}B_n=[B_n,B_m]$.
\end{definition}

\begin{remark}
  We note:
  (i) An equivalent Zakharov--Shabat formulation is, $[\pa_{t_n}-B_n,\pa_{t_m}-B_m]=0$;
  (ii) If $B_n\coloneqq (L^n)_+$, then a consequence of the Zakharov--Shabat formulation is that the flows defined by the Lax formulation~\eqref{eq:KPLaxform},
  commute for any $n,m\in\mathbb N$: $\pa_{t_n}\pa_{t_m}L=\pa_{t_m}\pa_{t_n}L$;
  (iii) Again, if $B_n\coloneqq (L^n)_+$, then the Zakharov--Shabat formulation directly follows from the Lax formulation, see for example Kodama~\cite[Th.~2.1]{Kodama}; and
  (iv) That the Lax formulation follows from the Zakharov--Shabat formulation ``is true up to a triangular change of time variables''; quoting from Krichever and Zabrodin~\cite{KZ}.
\end{remark}
\begin{example}
In the Zakharov--Shabat formulation in Definition~\ref{def:ZSform}, if we take $n=2$ and $m=3$ in which $B_2\coloneqq\pa_x^2+b_{2,0}$ and $B_3\coloneqq\pa_x^3+b_{3,1}\pa_x+b_{3,0}$,
then with $t_2$ identified as $y$, and $t_3$ identified as $t$, we eventually obtain the noncommutative KP equation~\eqref{eq:ncKP}.
\end{example}

Let us now outline the Sato--Wilson \emph{dressing transformation}.
It is possible to gauge transform $L$ into $\pa_x$.
Indeed, we can transform, $L\mapsto \pa_x=W^{-1} LW$, where $W$ is the Sato--Wilson operator~\eqref{eq:gaugeexpansionPoppe}. 
The coefficients $u_i$ of $L$ are related to the coefficients $w_k$ of $W$ via the relation $LW=W\pa_x$.
Thus for example, $u_1=\pa_xw_1$, $u_2=\pa_xw_2+u_1w_1$, and so forth.
In particular, we can consider the $w_k$ variables to be the \emph{primary variables} that determine the noncommutative KP hierarchy.
With this realisation in hand, we give the Sato formulation of the noncommutative KP hierarchy as follows.
\begin{definition}[Sato formulation]\label{def:Satoform}
  The Sato formulation of the noncommutative KP hierarchy is given as follows.
  For $n\in\mathbb N$, consider the system of equations, 
  \begin{equation}\label{eq:Satoform}
  \pa_{t_n}W=B_nW-W\pa_x^n,
  \end{equation}
  where the operator $B_n$ is given by $B_n=(W\pa_x^n W^{-1})_+$. 
\end{definition}
\begin{remark}
We can also take the operators $B_n$ to have the form \eqref{eq:BopsPoppe}, as we did in the Introduction.
\end{remark}
We quote the following result from Kodama~\cite[Th.~2.2]{Kodama}.
\begin{theorem}
If the Sato--Wilson operator $W$ satisfies the Sato equation~\eqref{eq:Satoform}, then the operator $L=W\pa W^{-1}$ satisfies the Lax equation~\eqref{eq:KPLaxform}
for the noncommutative KP hierarchy, and the operators $B_n=(W\pa ^nW^{-1})_+=(L^n)_+$ satisfy the Zakharov--Shabat equations in Definition~\ref{def:ZSform}.
\end{theorem}
For convenience, we set:
\begin{equation*}
A_n^\ast\coloneqq\pa_{t_n}-\pa_x^n\qquad\text{and}\qquad A_n\coloneqq\pa_{t_n}-B_n.
\end{equation*}
Then we also have the following.
\begin{lemma}\label{eq:Satodressing}
The Sato equations~\eqref{eq:Satoform}, are equivalent to the statement, 
\begin{equation*}
WA_n^\ast W^{-1}=A_n. 
\end{equation*}
\end{lemma}
\begin{proof}
  We observe that, $WA_n^\ast=A_nW$, is equivalent to, $W\pa_{t_n}-W\pa_x^n=\pa_{t_n}W-B_nW$, which, via the Sato equations~\eqref{eq:Satoform}, is equivalent to $W\pa_{t_n}-W\pa_x^n=-W\pa_x^n$.
  Since $W\pa_{t_n}=0$, the result follows.
\qed
\end{proof}

Let us now consider the more classical formulation of the dressing transform, see P\"oppe and Sattinger~\cite{PoppeSattinger}.
We define the Volterra integral operators $K_{\pm}$ as follows:
\begin{align*}
\bigl(K_+\psi\bigr)(x)\coloneqq\int_x^\infty \hat K_+(x,\xi)\psi(\xi)\,\rd \xi,\\
\bigl(K_-\psi\bigr)(x)\coloneqq\int_{-\infty}^x \hat K_-(x,\xi)\psi(\xi)\,\rd \xi,
\end{align*}
where the $\hat{K}_{\pm}$ represent the corresponding kernels, with $\hat{K}_+(x,\xi)=0$ if $\xi<x$ and $\hat{K}_-(x,\xi)=0$ if $\xi>x$. 
\emph{Suppose} that `$\id+K_-$' and `$\id+K_+$' both \emph{dress} the operator $A_n^\ast$ to $A_n$, i.e.\/ we have, 
\begin{equation}\label{eq:classicaldressing}
A_n(\id+K_\pm)=(\id+K_\pm)A_n^\ast. 
\end{equation}
In other words, both $\id+K_+$ and $\id+K_-$ are representatives of $W$.
\begin{definition}[Scattering operator]\label{def:scatteringoperator}
We define the \emph{scattering operator} $P$ as the operator given by,
\begin{equation}\label{eq:scatteringoperator}
\id-P\coloneqq(\id+K_-)^{-1}(\id+K_+). 
\end{equation}
\end{definition}
\begin{theorem}[Base equations]\label{thm:baseeqns}
The dressing equations \eqref{eq:classicaldressing} hold if and only if,
\begin{equation}\label{eq:baseeqns}
[\id-P,A_n^\ast]=0. 
\end{equation}
\end{theorem}
\begin{proof}
  We follow P\"oppe and Sattinger~\cite{PoppeSattinger}.
  First we demonstrate the base equations~\eqref{eq:baseeqns} imply the dressing equations~\eqref{eq:classicaldressing}. 
  We observe that, 
  \begin{align*}
    (\id+K_-)A_n^\ast&(\id+K_-)^{-1}\\
                   =&\;(\id+K_-)A_n^\ast(\id-P)(\id+K_+)^{-1}\\
                   =&\;(\id+K_-)(\id-P)A_n^\ast(\id+K_+)^{-1}\\
                   =&\;(\id+K_+)A_n^\ast(\id+K_+)^{-1},
  \end{align*}
  giving the result. Now we demonstrate that the dressing equations~\eqref{eq:classicaldressing} imply the base equations~\eqref{eq:baseeqns}. 
  We observe that $(\id-P)A_n^\ast-A_n^\ast(\id-P)$ equals,
  \begin{align*}
     &\;(\id+K_-)^{-1}(\id+K_+)A_n^\ast-A_n^\ast(\id+K_-)^{-1}(\id+K_+)\\
    =&\;(\id+K_-)^{-1}A_n(\id+K_+)-(\id+K_-)^{-1}A_n(\id+K_+),
  \end{align*}
  giving the zero operator and thus the second result.
\qed
\end{proof}
\begin{corollary}[Base equations equivalent]\label{cor:samebaseeqns}
  The base equations~\eqref{eq:baseeqns}, i.e.\/ $[\id-P,A_n^\ast]=0$, correspond precisely to the base equations~\eqref{eq:linearform}, i.e.\/ $ \pa_{t_n}p=\pa_{\el}^np-(-1)^n\pa_{\rl}^np$.
\end{corollary}
\begin{proof}
  We observe from Theorem~\ref{thm:baseeqns}, that the dressing equations \eqref{eq:classicaldressing} for $G^\prime\coloneqq K_-$ hold if and only if,
  \begin{align*}
    [\id-P,A_n^\ast]=0
    \quad\Leftrightarrow\quad & A_n^\ast-PA_n^\ast=A_n^\ast-A_n^\ast P\\
    \quad\Leftrightarrow\quad & P(\pa_{t_n}-\pa_x^n)=(\pa_{t_n}-\pa_x^n)P\\
    \quad\Leftrightarrow\quad & \pa_{t_n}P=\pa_x^nP-P\pa_x^n\\
    \quad\Leftrightarrow\quad & \pa_{t_n}P=\bigl(\pa_{\el}^nP-(-1)^n\pa_{\rl}^nP\bigr).
  \end{align*}
  Here for an operator $P$ with kernel $p=p(x,\xi)$, we have used the notation $\pa_{\el}p=\pa_x p$ and $\pa_{\rl}p=\pa_\xi p$.
\qed
\end{proof}
\begin{remark}
  The latter form of the base equations, namely~\eqref{eq:linearform}, are the base equations considered by P\"oppe and Sattinger~\cite[eq.~2.4]{PoppeSattinger} and P\"oppe~\cite[eq.~3.6]{PKP}.
\end{remark}
\begin{lemma}[GLM equation]\label{lemma:GLM}
The operator $G^\prime\coloneqq K_-$ satisfies the Gelfand--Levitan--Marchenko (GLM) equation,
\begin{equation}\label{eq:GLM}
P=G^\prime(\id-P). 
\end{equation}
\end{lemma}
\begin{proof}
  By definition, the scattering operator satisfies,
  \begin{align*}
                    &&  (\id+K_-)(\id-P)&=\id+K_+\\
    \Leftrightarrow && K_--K_-P-P&=K_+\\
    \Leftrightarrow && K_--K_-P&=P.
  \end{align*}
  In the last step we used that $\hat{K}_+(x,\xi)=0$ when $\xi<x$. 
\qed
\end{proof}
Recall our choice of `$\id+K_-$' in~\eqref{eq:classicaldressing} as one representative for the gauge transformation $W$.
In other words, an example operator $W$ in the gauge transformation, $L\mapsto \pa_x=W^{-1} LW$, is $W=\id+K_-=\id+G^\prime$. 
If we combine this fact with Lemma~\ref{eq:Satodressing}, Corollary~\ref{cor:samebaseeqns} and Lemma~\ref{lemma:GLM}, we establish the following. 
\begin{theorem}[Direct linearisation]\label{thm:DLviadressing}
  Suppose $G^\prime$ is the solution to the GLM equation~\eqref{eq:GLM}, i.e.\/ $P=G^\prime(\id-P)$, where the kernel $p$, of the operator $P$, satisfies the linearised equations~\eqref{eq:linearform}.
  Then the coefficients $w_k=w_k(x,\tbf)$ of $W=\id+K_-=\id+G^\prime$, satisfy the Sato system~\eqref{eq:Satoform}, and thus the noncommutative KP hierarchy.
\end{theorem}
\begin{remark}\label{rmk:extractingwk}
If we assume that $P$ is semi-additive, then there is a straightforward procedure for extracting the coefficients $w_k$ from the solution to the GLM equation, $G^\prime$.
See Theorem~\ref{thm:first} in Section~\ref{sec:SatoPoppeconnection}.
\end{remark}
\begin{remark}\label{rmk:baseeqns}
  For the same scattering data $P$ satisfying the linearised equations~\eqref{eq:linearform}, the GLM equation~\eqref{eq:GLM} for $G^\prime$ and the GLM equation~\eqref{eq:GLMeqintro} for $G$ are equivalent.
  We explore this connection in detail in Section~\ref{sec:SatoPoppeconnection}.
  The direct linearisation formulation in~\eqref{eq:linearform} and \eqref{eq:GLMeqintro}, matches the formulation for the KP equation by P\"oppe~\cite{PKP},
  and for the noncommutative KP equation by Blower and Malham~\cite{BM-KP}.
\end{remark}

\section{The Sato and pre-P\"oppe algebra connection}\label{sec:SatoPoppeconnection}
We establish the connection between the classical GLM equation for the noncommutative KP hierarchy
and the corresponding GLM equation based on a semi-additive form for the scattering operator $P$ in Section~\ref{sec:Poppealg}.
Let us denote the scattering operator in~\eqref{eq:scatteringoperator} by $P^\prime$, i.e.\/ $P^\prime$ is defined via $\id-P^\prime\coloneqq(\id+K_-)^{-1}(\id+K_+)$,
where $K_\pm$ both dress the base operator $\pa_{t_n}-\pa_x^n$ to $\pa_{t_n}-B_n$.
The classical GLM equation~\eqref{eq:GLM} then has the form,
\begin{equation*}
P^\prime=G^\prime(\id-P^\prime),
\end{equation*}
or more explicitly, in terms of the kernels $p^\prime$ and $g^\prime$,
\begin{equation}\label{eq:GLMclassical}
p^\prime(x,\xi;\tbf)=g^\prime(x,\xi;\tbf)-\int_{-\infty}^x g^\prime(x,\eta;\tbf)p^\prime(\eta,\xi;\tbf)\,\rd\eta.
\end{equation}
Herein, we define the GLM equation to be the linear integral equation in \eqref{eq:GLMeqintro}, i.e.\/ $P=G(\id-P)$, for the operator $G$ with kernel $g=g(z,\zeta;x,\tbf)$.
In this formulation, we assume the kernel of the scattering operator $P$ has the semi-additive form $p=p(z+x,\zeta+x;\tbf)$. 
In terms of the underlying kernels, setting $z=0$, we have that $g$ satisfies (suppressing the $\tbf$-dependence in $p$ and $g$),
\begin{equation}
  p(x,\zeta+x)=g(0,\zeta;x)-\int_{-\infty}^0\!\!g(0,\nu;x)p(\nu+x,\zeta+x)\,\rd\nu.\label{eq:GLMPoppe}
\end{equation}
If, in~\eqref{eq:GLMPoppe}, we set $\xi\coloneqq\zeta+x$ and $\eta\coloneqq\nu+x$, and then set,
\begin{subequations}\label{eq:primekernels}
\begin{align}
  g^\prime(x,\xi;\tbf)&\coloneqq g(0,\xi-x;x,\tbf),\\
  p^\prime(x,\xi;\tbf)&\coloneqq p(x,\xi;\tbf),\label{eq:pkernel}
\end{align}
\end{subequations}
then we observe that $g^\prime$ and $p^\prime$ satisfy~\eqref{eq:GLMclassical}.
It is straightforward to show we can also reverse this statement.
Indeed, we have established the following.
\begin{lemma}[Equivalent GLM formulations]
  Suppose $p=p(z+x,\zeta+x;\tbf)$ is semi-additive, and the fields $g^\prime$ and $p^\prime$, and, $g=g(0,\zeta;x,\tbf)$ and $p$, are related via~\eqref{eq:primekernels}.  
  Then, if $g$ satisfies \eqref{eq:GLMPoppe}, the field $g^\prime$ satisfies the classical GLM equation~\eqref{eq:GLMclassical}, and vice versa. 
\end{lemma}
Hereafter, we suppose the scattering operator $P$ is semi-additive with kernel $p=p(z+x,\zeta+x;\tbf)$.

We observe, by successively integrating by parts, we can generate the following expansion for the operator $G^\prime$:
\begin{align*}
  \bigl(G^\prime\psi\bigr)(x;\tbf)=&\;\Bigl(\bigl(g^\prime(x,x;\tbf)\pa_x^{-1}-\bigl(\pa_{\rl}g^\prime(x,x;\tbf)\bigr)\pa_x^{-2}\\
  &\;~~~~~+\bigl(\pa_{\rl}^2g^\prime(x,x;\tbf)\bigr)\pa_x^{-3}-\cdots\bigr)\psi\Bigr)(x).
\end{align*}
Since by definition, $g^\prime(x,x;\tbf)=g(0,0;x,\tbf)$, we observe that coefficients of the terms involving $\pa_x^{-k}$ in the expansion above,  
are given by $(-1)^{k-1}\pa^{k-1}_{\rl}\lb G\rb(0,0;x,\tbf)$.

Recall from Theorem~\ref{thm:first} that $W=\id+K_-\equiv\id+G^\prime$.
Hence for $k\in\mathbb N$, the coefficients $w_k=w_k(x;\tbf)$, in the expansion for the gauge operator $W$ in~\eqref{eq:gaugeexpansionPoppe},
are $w_1(x;\tbf)=-g^\prime(x,x;\tbf)$, $w_2(x;\tbf)=\pa_{\rl}g^\prime(x,x;\tbf)$, and so forth.
We have thus established the following result.
\begin{theorem}[Gauge coefficient identities]\label{thm:gaugecoeffsid}
  The coefficients $w_k$ from the operator $W$ of the gauge transformation, are related to the solution $g$ of the P\"oppe form of the GLM equation~\eqref{eq:GLMPoppe}, for all $k\in\mathbb N$ via,
  \begin{equation}\label{eq:solutionansatz}
    w_k(x,\tbf)\equiv(-1)^{k}\pa^{k-1}_{\rl}\lb G\rb(0,0;x,\tbf).
  \end{equation}
\end{theorem}
The coefficients $w_k$ of the operator $W$ generate the fields in the noncommutative KP hierarchy.
Theorem~\ref{thm:gaugecoeffsid} establishes that these coefficients are given by the signed partial derivatives $\pa_{\rl}^{k-1}$ of the kernels $g=g(z,\zeta;x,\tbf)$ evaluated at $z=\zeta=0$.
Recall from Section~\ref{sec:Poppealg} that we set $V\coloneqq(\id-P)^{-1}$.
\begin{corollary}\label{cor:WequalsV}
  The Gauge operator $W$ and the operator $V$ are one and the same, i.e.\/ $W=V$.
\end{corollary}
\begin{proof}
From \eqref{eq:pkernel} we know that $p^\prime\equiv p$ and thus $P^\prime=P$. Hence if $V^\prime\coloneqq(\id-P^\prime)^{-1}$, then $G=PV=P^\prime V^\prime=G^\prime$.
Since $W=\id+G^\prime$, we have $W=\id+G$, and thus via the partial fractions formulae \eqref{eq:partialfractions}, the result follows.
\qed
\end{proof}
Recall Theorem~\ref{thm:DLviadressing} and Remark~\ref{rmk:extractingwk} from Section~\ref{sec:Satoindirect}. We can now state the following.
\begin{theorem}[Direct linearisation via dressing]\label{thm:first}
  Suppose $G$ is the solution to the GLM equation~\eqref{eq:GLMeqintro}, i.e.\/ $P=G(\id-P)$, where $P$ is a Hilbert--Schmidt operator of semi-additive form
  and and its kernel $p$ satisfies the linearised equations~\eqref{eq:linearform}.
  Then the coefficients $w_k=w_k(x,\tbf)$ given in \eqref{eq:solutionansatz}, i.e. $w_k(x,\tbf)\equiv(-1)^{k}\pa^{k-1}_{\rl}\lb G\rb(0,0;x,\tbf)$,
  satisfy the Sato system~\eqref{eq:Satoform}, and thus the noncommutative KP hierarchy.
\end{theorem}
Having establish direct linearisation of the noncommutative KP hierarchy via the Sato--Wilson dressing transformation,
in the remaining sections we consider a `head-on' approach to direct linearisation via the P\"oppe approach.

\section{Noncommutative hierarchy equations}\label{sec:hierarchyequations}
Herein and henceforth, we consider the P\"oppe approach to the direct linearisation of the noncommutative KP hierarchy.
We present the Sato formulation of the noncommutative KP hierarchy equations, given by the system of equations \eqref{eq:gaugeexpansionPoppe}--\eqref{eq:BopsPoppe}, in the context of P\"oppe's approach.
From Corollary~\ref{cor:WequalsV}, note that we can equivalently replace $W$ in \eqref{eq:SatoformPoppe} by $V$.
Recall the base equations~\eqref{eq:linearform}, i.e.\/ $\pa_{t_n}P=\pa_{\el}^nP-(-1)^n\pa_{\rl}^nP$.
Since $V=(\id-P)^{-1}$ we observe that for all $n\in\mathbb N$,
\begin{equation}\label{eq:PoppeV}
  \pa_{t_n}V=V\bigl(\pa_{\el}^nP-(-1)^n\pa_{\rl}^nP\bigr)V.
\end{equation}
For example, when $n=1$, this relation reduces to the relation, $\pa_{t_1}V=V P_{\hat 1}V=\pa_x V$, consistent with our previous designation of $t_1=x$.
When $n=2$, we observe that~\eqref{eq:PoppeV} is equivalent to the relation, $\pa_{t_2}V=V(P_{2,0}-P_{0,2})V$.
In the case $n=3$, we have, $\pa_{t_2}V=V(P_{3,0}+P_{0,3})V$, and so forth. See, for example, Blower and Malham~\cite{BM-KP}.

We now explore the Sato formulation for the noncommutative KP hierarchy in the context of P\"oppe's approach. 
\begin{example}[Case $n=2$]\label{ex:n=2}
  The equations in this case are given by~\eqref{eq:n=2relations} in the Introduction.
  In particular, equation~\eqref{eq:SatoVevol2} holds for $w_1$. 
  Also recall that we identify $t_2=y$.
  If we use the solution ansatz~\eqref{eq:solutionansatz}, then the relation
  for $w_1$ is equivalent to the relation, 
  \begin{equation}\label{eq:n2id}
    \pa_y\lb G\rb=\pa_x^2\lb G\rb-2\bigl(\pa_x\lb G\rb\bigr)\lb G\rb-2\pa_x\pa_{\rl}\lb G\rb.
  \end{equation}
  That this relation does indeed hold for $G=PV$ is straightforwardly established by using the P\"oppe product formula from Lemma~\ref{lemma:Poppeprodmonomials}
  on the nonlinear term `$2\bigl(\pa_x\lb G\rb\bigr)\lb G\rb$' therein and that we know that $\pa_{t_2}V=V(P_{2,0}-P_{0,2})V$ from~\eqref{eq:PoppeV}.
  The corresponding relation for $w_2$ has the form,
  \begin{equation*}
    \pa_{t_2}w_2=\pa_x^2w_2+2(\pa_xw_1)w_2+2\pa_xw_3.
  \end{equation*}
  That this identity holds follows by taking the `$\pa_{\rl}$' derivative of the relation~\eqref{eq:n2id} and using the properties~\eqref{eq:leftandrightderivs}.
  The corresponding relations for $w_3$, $w_4$, and so forth follow suit. Thus, relation~\eqref{eq:n2id} establishes that~\eqref{eq:SatoVevol2} holds.
\end{example}
In general, the Sato formulation is equivalent to the following prescription.
\begin{theorem}[KP hierarchy equations]\label{thm:hierarchyequations}
  The noncommutative KP hierarchy equations are given by the following formulae, for $n\in\mathbb N$ and $k\in\{0,1,\ldots,n-2\}$ and all $\kappa\in\mathbb N$:
  \begin{subequations}\label{eq:ncKPhierarchy}
  \begin{align}
    \mathcal L_k\coloneqq&\;\sum_{\ell=0}^{k-2}C_{\ell}^{n-k+\ell}b_{n-k+\ell}^n\pa_x^\ell+C_{k}^n\pa_x^k,\label{eq:ncKPhierarchyL}\\
    b_k^n\coloneqq&\;\sum_{\ell=1}^{n-k-1}\mathcal L_{n-k-\ell}w_\ell,\label{eq:ncKPhierarchybnk}\\
    \pa_{t_n}w_\kappa=&\;\sum_{\ell=1}^n\mathcal L_{n-\ell+1}w_{\kappa+\ell-1}.\label{eq:ncKPhierarchyw1}
  \end{align}
  \end{subequations}
\end{theorem}
\begin{proof}
We insert the expansion \eqref{eq:gaugeexpansionPoppe} for $W$ into \eqref{eq:SatoformPoppe}, and use \eqref{eq:BopsPoppe}. 
This gives,
\begin{align}
  \sum_{m\geqslant1}\bigl(\pa_{t_n}w_m\bigr)&\pa_x^{-m}+\sum_{k=0}^{n-2}b_k^n\,\pa_x^k\nonumber\\
  =&\;\sum_{m\geqslant1}\sum_{\ell=0}^{n-1}C_\ell^n\bigl(\pa_x^{n-\ell}w_m\bigr)\pa_x^{-m+\ell}\nonumber\\
  &\;+\sum_{m\geqslant1}\sum_{\ell=0}^{n-2}\sum_{k=\ell}^{n-2}C_\ell^nb_k^n\bigl(\pa_x^{k-\ell}w_m\bigr)\pa_x^{-m+\ell}.\label{eq:firstinsert}
\end{align}
If we equate the coefficients of $\pa_x^k$ for $k\in\{0,\ldots,n-2\}$ in \eqref{eq:firstinsert}, we observe that the first term on the left is not involved,
and the resulting relations give a sequential set of equations prescribing the coefficients $b_k^n$ from $b_{n-2}^n$ down through to $b_0^n$.
The first relation is $b_{n-2}^n=\mathcal L_1w_1$, where $\mathcal L_1\coloneqq C_{1}^n\pa_x$, using that $C_k^n\equiv C_{n-k}^n$.
The second relation is $b_{n-3}^n=\mathcal L_2w_1+\mathcal L_1w_2$, where $\mathcal L_2\coloneqq C_{2}^n\pa_x^2+C_{0}^{n-2}b_{n-2}^n$.
It is straightforward to show that for $k\in\{2,\cdots,n\}$, 
\begin{equation*}
  b_{n-k}^n=\mathcal L_{k-1}w_1+\mathcal L_{k-2}w_2+\cdots+\mathcal L_1w_{k-1},
\end{equation*}
where $\mathcal L_{k}=C_{k}^n\pa_x^k+C_{k-2}^{n-2}b_{n-2}^n\,\pa_x^{k-2}+\cdots+C_{0}^{n-k}b_{n-k}^n$.
Note that, $b_0^n=\mathcal L_{n-1}w_1+\mathcal L_{n-2}w_2+\cdots+\mathcal L_{1}w_{n-1}$.
We have thus established the form for the operators $\mathcal L_k$ and coefficients $b_k^n$ in the statement of the theorem.

If we now equate the coefficients of $\pa_x^{-m}$ for $m\geqslant1$ in \eqref{eq:firstinsert}, we observe that the second term on the right is not involved,
and via completely analogous arguments to those just above, we arrive at the expression for $\pa_{t_n}w_\kappa$ given in the statement of the theorem. 
\qed
\end{proof}
\begin{remark}\label{rmk:b_1}
  In the relation for $w_1$ given in Theorem~\ref{thm:hierarchyequations}, we observe that terms on the right are precisely the terms
  we would expect if we were to extend the parameter $k$ for the set of coefficients $b_k^n$ to the case $k=-1$.
  In other words, dropping the superindex $n$, we have,
  \begin{equation*}
    \pa_{t_n}w_1=b_{-1}.
  \end{equation*}
\end{remark}
\begin{remark}
  For a given order $n\in\mathbb N$, it is usual to close-off the system at the level of the evolution equation for $w_1$ involving $\pa_{t_n}w_1$.
  For example, see \eqref{eq:n=3caseintro}. In general, the equation for $\pa_{t_n}w_1$ involves $w_2$ through to $w_n$.
  However we can substitute for the latter dependent variables by using the corresponding relations for $\pa_{t_{n-1}}w_k$ for $k=n-1$ down to $k=1$, and so forth.
  We thus obtain a closed system of equations for $w_1$. See Example~\ref{ex:n=3} just below.
  From a different perspective, using the solution ansatz \eqref{eq:solutionansatz}, we only need to verify the equation for $w_1$ in Theorem~\ref{thm:hierarchyequations}.
  This is because we can use the properties of the $\pa_\rl$ operator outlined in~\eqref{eq:leftandrightderivs} and the pre-P\"oppe product.
  Using these, we observe that the equations for $\pa_{t_n}w_2$, $\pa_{t_n}w_3$, etc, automatically follow from the equation for $\pa_{t_n}w_1$.
  See Example~\ref{ex:n=2}. 
\end{remark}
\begin{example}[KP equation]\label{ex:n=3}
  In the case when $n=3$, the equation for $w_1$ in Theorem~\ref{thm:hierarchyequations} becomes---hereafter we drop the superindex for $b_k^n$,
  \begin{subequations}
  \begin{align}
     b_{n-2}=&\;C_{1}^n\pa_x w_1,\label{eq:n=3casebn2}\\
     b_{n-3}=&\;\bigl(C_{2}^n\pa_x^2+C_{0}^{n-2}b_{n-2}\bigr)w_1+C_{1}^n\pa_xw_2,\label{eq:n=3casebn3}\\
     \pa_{t_n}w_1=&\;\mathcal L_3w_1+\mathcal L_2w_2+\mathcal L_1w_3\nonumber\\
                =&\;\bigl(C_{3}^n\pa_x^3+C_{1}^{n-2}b_{n-2}\,\pa_x+C_{0}^{n-3}b_{n-3}\bigr)w_1\nonumber\\
                 &\;+\bigl(C_{2}^n\pa_x^2+C_{0}^{n-2}b_{n-2}\bigr)w_2+C_{1}^n\pa_xw_3.\label{eq:n=3case}
  \end{align}
  \end{subequations}
  We have retained the general `$n$' label for illustrative purposes later on.
  Equation~\eqref{eq:n=3case} is equivalent to the noncommutative KP equation.
  To see this, as briefly outlined in the Introduction, from \eqref{eq:n=2relations}, with $t_2=y$, we observe,
  \begin{subequations}\label{eq:n=2strut}
  \begin{align}
    \pa_{y}w_1&=\pa_x^2w_1+2(\pa_xw_1)w_1+2\pa_xw_{2},\\
    \pa_{y}w_2&=\pa_x^2w_2+2(\pa_xw_1)w_2+2\pa_xw_{3}.
  \end{align}
  \end{subequations}
  If we explicitly substitute that $n=3$ in the relations just above for $b_{n-2}$, $b_{n-3}$ and $\pa_{t_n}w_1$ in \eqref{eq:n=3case},
  set $t=t_3$, and use the relations \eqref{eq:n=2strut} to substitute for $w_3$ and $w_2$, after standard reductions,
  with $w\coloneqq w_1$, we arrive at the noncommutative KP equation~\eqref{eq:ncKP} given in the Introduction.
\end{example}
The noncommutative KP equation in \eqref{eq:ncKP} is in potential form and $w$ only appears in derivative form. 
Hence, to show that $\lb G\rb$ satisfies \eqref{eq:ncKP}, it is sufficient to show that
$\pa_{t_3}w=-\pa_{t_3}\lb G\rb=-\pa_{t_3}\lb V\rb=-\lb V(\pa_{t_3})V\rb=-\lb V(P_{3,0}+P_{0,3})V\rb$
matches the linear and nonlinear terms involving $\pa_xw=\lb VP_{\hat{1}}V\rb$ on the right-hand side.
This is demonstrated in Blower and Malham~\cite[Sec.~3]{BM-KP}.
However, in the Sato formulation of the noncommutative KP hierarchy, we need to deal with products involving $\lb G\rb$, where $G=VP=PV$.
We explain how we deal with this next.

\section{The augmented pre-P\"oppe algebra}\label{sec:steptodescents}
Herein we examine the consequences of substituting the ansatz \eqref{eq:solutionansatz} into the complete noncommutative KP hierarchy \eqref{eq:ncKPhierarchy} given in Theorem~\ref{thm:hierarchyequations}.
This leads to the augmented pre-P\"oppe algebra with the property that it is nonassociative. 
For convenience in this section, we adopt the following notation. We use $\mathbb P$ to denote the pre-P\"oppe algebra given in Definition~\ref{def:prePoppealg}.
Recall that a basis for $\mathbb P$ is the set of monomials of the form~\eqref{eq:exmonomial}, though the monomial `$\lb V\rb$' is excluded.
We could introduce an extended pre-P\"oppe algebra, in which we also include monomials of the form
$\lb VP_{a_1,b_1}VP_{a_2,b_2}V\cdots VP_{a_k,b_k}VP_{a_{k+1},b_{k+1}}\rb$ and $\lb P_{a_0,b_0}VP_{a_1,b_1}VP_{a_2,b_2}V\cdots VP_{a_k,b_k}V\rb$.
Products between such monomials are straightforwardly constructed from the pre-P\"oppe products in Lemmas~\ref{lemma:Poppeprodsemiadd} and \ref{lemma:Poppeprodmonomials}, 
combined with the partial fraction formulae~\eqref{eq:partialfractions}. 
As we see below, after substituting the ansatz \eqref{eq:solutionansatz} into the noncommutative KP hierarchy equations \eqref{eq:ncKPhierarchy}, 
most of the terms present lie in $\mathbb P$, but we also need to multiply on the right by $\lb G\rb\equiv\lb VP\rb$. 
Further, right derivative terms of the form $\pa_{\rl}^c\lb VP_{a_1,b_1}VP_{a_2,b_2}V\cdots VP_{a_k,b_k}V\rb$ also appear due to the ansatz \eqref{eq:solutionansatz}.
By using the partial fraction formulae~\eqref{eq:partialfractions}, such terms lie in the extended pre-P\"oppe algebra as opposed to $\mathbb P$.
However, in the Sato formulation of the noncommutative KP hierarchy, these two aberrations into the extended pre-P\"oppe algebra and away from $\mathbb P$, combine so that we remain in $\mathbb P$.
Indeed we can define a bimodule on $\mathbb P$ that succinctly incorporates this issue.
We now explain this in detail.
Let us define two rings, $\mathbb{A}_{\el}$ and $\mathbb{A}_{\rl}$. We set,
\begin{equation*}
A_{\el}\coloneqq\pa_{\el}+\lb G\rb\qquad\text{and}\qquad A_{\rl}\coloneqq\pa_{\rl}+\lb G\rb.
\end{equation*}
Suppose that $\mathbb{A}_{\el}$ is the simple ring with the single generator $A_{\el}$ and concatenation as product.
Thus $\mathbb{A}_{\el}$ consists of all the polynomials of $A_{\el}$ with a unit defined with the usual properties. 
The ring $\mathbb{A}_{\rl}$ is defined similarly, via the single generator $A_{\rl}$ and concatenation as product.
It consists of all the polynomials of $A_{\rl}$, and so forth.
We now construct the $\mathbb{A}_{\el}$-$\mathbb{A}_{\rl}$-bimodule on $\mathbb{P}$ as follows. We denote this bimodule by $\mathbb{P}_A$.
We define the actions of $A_{\el}$ and $A_{\rl}$ on $\mathbb{P}$ as follows.
For convenience, hereafter in this section, we denote a generic monomial in $\mathbb{P}$ by $\lb V_{a,b}\rb$, where for $a=(a_1,\ldots,a_k)$ and $b=(b_1,\ldots,b_k)$ we set,
\begin{equation*}
V_{a,b}\coloneqq VP_{a_1,b_1}V\cdots VP_{a_k,b_k}V.
\end{equation*}
\begin{definition}[Actions]\label{def:actions}
  We define the actions $A_{\el}$ and $A_{\rl}$ on $\mathbb{P}$, i.e.\/ $A_{\el}\colon\mathbb{A}_{\el}\times\mathbb P\to\mathbb P$ and $A_{\rl}\colon\mathbb P\times\mathbb{A}_{\rl}\to\mathbb P$, by,
  \begin{align*}
    A_{\el}\colon &\lb V_{a,b}\rb\mapsto \pa_{\el}\lb V_{a,b}\rb+\lb G\rb\,\lb V_{a,b}\rb,\\ 
    A_{\rl}\colon &\lb V_{a,b}\rb\mapsto \pa_{\rl}\lb V_{a,b}\rb+\lb V_{a,b}\rb\,\lb G\rb,
  \end{align*}
  where $\lb G\rb=\lb VP\rb=\lb PV\rb$. The products with $\lb G\rb$ on the right are standard pre-P\"oppe products. These actions naturally extend linearly.
\end{definition}
We confirm these actions are well-defined on $\mathbb{P}$, as follows.
\begin{lemma}[Actions closed in $\mathbb{P}$]\label{lemma:crucial}
  The standard pre-P\"oppe product implies:
  \begin{align*}
    A_{\el}\circ\lb V_{a,b}\rb=&\;\lb V(\pa_\el P_{a_1,b_1})V\cdots VP_{a_k,b_k}V\rb+\lb VP_{\hat{1}}V_{a_,b}\rb,\\
    A_{\rl}\circ\lb V_{a,b}\rb=&\;\lb VP_{a_1,b_1}V\cdots V(\pa_\rl P_{a_k,b_k})V\rb+\lb V_{a,b}P_{\hat{1}}V\rb.
  \end{align*}
\end{lemma}
\begin{proof}
    We establish the result for $A_{\el}$. The corresponding result for $A_{\rl}$ follows completely analogously.
    We use the standard pre-P\"oppe product from Lemma~\ref{lemma:Poppeprodmonomials} and the partial fraction formulae \eqref{eq:partialfractions}.
    Indeed, using that $V=\id+PV$, so that for example, $V_{a,b}=P_{a_1,b_1}V\cdots VP_{a_k,b_k}V+PV_{a,b}$,
    and the properties of the $\pa_{\el}$ operator from~\eqref{eq:leftandrightderivs}, we observe, 
    \begin{align*}
      \pa_{\el}\lb V_{a,b}\rb+\lb PV&\rb\,\lb V_{a,b}\rb\\
                        =&\;\pa_{\el}\bigl(\lb P_{a_1,b_1}V\cdots\rb+\lb PV_{a,b}\rb\bigr)+\lb P_{0,1}V_{a,b}\rb\\
                         &\;+\lb PV(\pa_{\el}P_{a_1,b_1})V\cdots\rb+\lb PVP_{\hat{1}}V_{a,b}\rb\\
                        =&\;\lb (\pa_{\el}P_{a_1,b_1})V\cdots\rb+\lb P_{1,0}V_{a,b}\rb+\lb P_{0,1}V_{a,b}\rb\\
                         &\;+\lb PV(\pa_{\el}P_{a_1,b_1})V\cdots\rb+\lb PVP_{\hat{1}}V_{a,b}\rb,
    \end{align*}
    which gives the result once we use that $P_{\hat{1}}=P_{1,0}+P_{0,1}$, and $V=\id+PV$ again.
\qed
\end{proof}
Lemma~\ref{lemma:crucial} confirms the actions $A_{\rl}$ and $A_{\el}$ are closed on $\mathbb P$.
To demonstrate that $\mathbb{P}_A$ is a well-defined $\mathbb{A}_{\el}$-$\mathbb{A}_{\rl}$-bimodule on $\mathbb{P}$, we need to confirm a few more properties.
The distributive laws for addition for both the left $\mathbb{A}_{\el}$-module and $\mathbb{A}_{\rl}$-module components are straightforward. 
It is straightforward to demonstrate that the actions $A_{\el}\circ\bigl(A_{\el}\circ\lb V_{a,b}\rb\bigr)$ and $A_{\el}^2\circ\lb V_{a,b}\rb$ match,
and that this establishes the associativity of the left $\mathbb{A}_{\el}$-module component.
Exactly analogous arguments establish the associativity of the right $\mathbb{A}_{\rl}$-module component.
It remains to demonstrate the bimodule compatability condition,
\begin{equation}\label{eq:bimodulecompatability}
  A_{\el}\circ\bigl(A_{\rl}\circ\lb V_{a,b}\rb\bigr)=A_{\rl}\circ\bigl(A_{\el}\circ\lb V_{a,b}\rb\bigr).
\end{equation}
That this statement holds is straightforwardly established using Lemma~\ref{lemma:crucial}.
We have thus established the following.
\begin{lemma}[Action bimodule $\mathbb{P}_A$]
  The set $\mathbb{P}_A$ is a well-defined $\mathbb{A}_{\el}$-$\mathbb{A}_{\rl}$-bimodule on $\mathbb{P}$.
\end{lemma}
We illustrate how the action $A_{\rl}$, in particular, appears in the Sato formulation of the noncommutative KP hierarchy via the following example.
\begin{example}[KP equation revisited]\label{ex:n=3revisited}
  We substitute the anstaz~\eqref{eq:solutionansatz} into the noncommutative KP equation \eqref{eq:n=3case}.
  If we focus on the linear terms on the right in \eqref{eq:n=3case}, then the terms associated with $b_{n-2}$ and $b_{n-3}$, we observe,
  \begin{align}
    \pa_{t_n}w_1\!=&-C_{3}^n\pa_x^3\lb G\rb+C_{2}^n\pa_x^2\pa_\rl\lb G\rb-C_{1}^n\pa_x\pa_\rl^2\lb G\rb\nonumber\\
                &+C_{1}^nC_{1}^{n-2}\bigl(\pa_x\lb G\rb\bigr)^2\nonumber\\
                &-C_{1}^nC_{0}^{n-2}\bigl(\pa_x\lb G\rb\bigr)\bigl(\pa_\rl\lb G\rb\bigr)\nonumber\\
                &+C_{2}^{n}C_{0}^{n-3}\bigl(\pa_x^2\lb G\rb\bigr)\lb G\rb\nonumber\\
                &-C_{1}^nC_{0}^{n-2}C_{n-3}^{n-3}\bigl(\pa_x\lb G\rb\bigr)\lb G\rb^2\nonumber\\
                &-C_{1}^{n}C_{0}^{n-3}\bigl(\pa_x\pa_\rl\lb G\rb\bigr)\lb G\rb\nonumber\\
               =&-C_{3}^n\pa_x^3\lb G\rb\!+\!C_{2}^n\,A_{\rl}\circ\pa_x^2\lb G\rb\!+\!C_{1}^nC_{1}^{n-2}\bigl(\pa_x\lb G\rb\bigr)^2\nonumber\\
                &-C_{1}^n\,A_{\rl}\circ\pa_x\pa_\rl\lb G\rb\!-\!C_{1}^n\,A_{\rl}\circ\bigl((\pa_x\lb G\rb)\lb G\rb\bigr)\nonumber\\
               =&-C_{3}^n\pa_x^3\lb G\rb\!+\!C_{2}^n\,A_{\rl}\circ\pa_x^2\lb G\rb\!+\!C_{1}^nC_{1}^{n-2}\bigl(\pa_x\lb G\rb\bigr)^2\nonumber\\
                &-C_{1}^n\,A_{\rl}\circ\bigl(A_{\rl}\circ\pa_x\lb G\rb\bigr).\label{eq:begright}
  \end{align}
  For completeness, since $G=PV\equiv VP$ and $\pa_x\lb G\rb\equiv\pa_x\lb V\rb=\lb VP_{\hat{1}}V\rb$, we can use Lemma~\ref{lemma:crucial} to compute that,
  \begin{align}
    A_{\rl}\circ\bigl(A_{\rl}\circ\pa_x\lb G\rb\bigr)=&\;A_{\rl}\circ\bigl(\lb V(\pa_\rl P_{\hat{1}})V\rb+\lb VP_{\hat{1}}VP_{\hat{1}}V\rb\bigr)\nonumber\\
    =&\;\lb V(\pa_\rl^2 P_{\hat{1}})V\rb+\lb V(\pa_{\rl}P_{\hat{1}})VP_{\hat{1}}V\rb\nonumber\\
    &\;+\lb VP_{\hat{1}}V(\pa_{\rl}P_{\hat{1}})V\rb\nonumber\\
    &\;+\lb VP_{\hat{1}}VP_{\hat{1}}VP_{\hat{1}}V\rb.\label{eq:firstright}
  \end{align}
  And similarly we have,
  \begin{align}
    A_{\rl}\circ\pa_x^2\lb G\rb=&\;\lb V(\pa_\rl P_{\hat{2}})V\rb+\lb VP_{\hat{2}}VP_{\hat{1}}V\rb\nonumber\\
    &\;+2\,\lb VP_{\hat{1}}V(\pa_\rl P_{\hat{1}})V\rb\nonumber\\
    &\;+2\,\lb VP_{\hat{1}}VP_{\hat{1}}VP_{\hat{1}}V\rb,\label{eq:secondright}
  \end{align}
  If we substitute \eqref{eq:firstright} and \eqref{eq:secondright} into \eqref{eq:begright},
  it is straightforward to show that the righthand side in \eqref{eq:begright} collapses to $-\lb V(P_{3,0}+P_{0,3})V\rb$. Thus $\lb G\rb$ solves \eqref{eq:n=3case}.
\end{example}

We observe from Example~\ref{ex:n=3revisited} that the natural context for the complete noncommutative KP hierarchy appears to be the $\mathbb{A}_{\el}$-$\mathbb{A}_{\rl}$-bimodule set $\mathbb{P}_A$.
We demonstrate this is indeed the case in Corollary~\ref{cor:hierarchyequations} at the end of this section.
Recall that in Definition~\ref{def:prePoppealg} for the pre-P\"oppe algebra $\mathbb{P}$, we excluded the monomial `$\lb V\rb$'.
This was because the standard pre-P\"oppe product in Lemma~\ref{lemma:Poppeprodmonomials}, precludes the cases $FP_{a,b}=\id$ and $P_{c,d}\hat{F}=\id$, for which it does not apply.
Hence we excluded a pre- or postfactor being `$\lb V\rb$'.
However, suppose we could include the case $P_{c,d}\hat F=\id$, for example, then the natural result we would expect from the pre-P\"oppe product in Lemma~\ref{lemma:Poppeprodmonomials}, would be,
\begin{equation}\label{eq:nogo}
  \lb V_{a,b}\rb \lb V\rb=\lb VP_{a_1,b_1}V\cdot\!\cdot\!\cdot V(\pa_\rl P_{a_k,b_k})V\rb+\lb V_{a,b}P_{\hat{1}}V\rb.
\end{equation}
Assume for the moment that \eqref{eq:nogo} holds.
Then we observe, 
\begin{equation}\label{eq:nogoconnect}
  A_{\rl}\circ\lb V_{a,b}\rb=\lb V_{a,b}\rb\,\lb V\rb.
\end{equation}
Propagating this perspective further, let us re-examine Example~\ref{ex:n=3revisited} with \eqref{eq:nogoconnect} in mind.
\begin{example}[KP equation reprised]\label{ex:n=3reprised}
  Consider $b_{n-2}$ and $b_{n-3}$ in \eqref{eq:n=3casebn2} and \eqref{eq:n=3casebn3}, respectively.
  Using the ansatz \eqref{eq:solutionansatz}, that $G=VP$ and $\pa_x\lb G\rb=\pa_x\lb V\rb$, we find,
  \begin{align*}
       b_{n-2}&=-C_1^n\pa_x\lb V\rb,\\
       b_{n-3}&=-C_2^n\pa_x^2\lb V\rb-b_{n-2}\,\lb VP\rb+C_1^n\pa_{\rl}\pa_x\lb V\rb\\
             &=-C_2^n\pa_x^2\lb V\rb-A_{\rl}\circ b_{n-2}\\
             &=-C_2^n\pa_x^2\lb V\rb-b_{n-2}\,\lb V\rb. 
  \end{align*}
  Then if we consider $\pa_{t_n}w_1$ in \eqref{eq:n=3case}, we observe,
  \begin{align*}
    \pa_{t_n}w_1=&\;C_3^n\pa_x^3\lb V\rb+C_1^{n-2}b_{n-2}\,\pa_x\lb V\rb+b_{n-3}\,\lb VP\rb\\
                &\;-C_2^n\pa_x^2\lb V\rb-C_1^{n-2}b_{n-2}\,\pa_{\rl}\lb V\rb+C_1^n\pa_{\rl}^2\pa_x\lb V\rb\\
               =&\;C_3^n\pa_x^3\lb V\rb+C_1^{n-2}b_{n-2}\,\pa_x\lb V\rb+A_{\rl}\circ b_{n-3}\\
               =&\;C_3^n\pa_x^3\lb V\rb+C_1^{n-2}b_{n-2}\,\pa_x\lb V\rb+b_{n-3}\lb V\rb. 
  \end{align*}
  We thus see that, if \eqref{eq:nogo} does hold, then the expressions above for $b_{n-3}$ and $\pa_{t_n}w_1$
  are more succinct and a simple pattern for the coefficients emerges.
\end{example}
We exclude `$\lb V\rb$' from the pre-P\"oppe algebra for functional analytic reasons.
However, we can still construct an \emph{abstract algebra} based on the pre-P\"oppe product, but which also includes the product \eqref{eq:nogo} and the equivalent for $\lb V\rb\,\lb V_{a,b}\rb$, as follows.
\begin{definition}[Augmented pre-P\"oppe product]\label{def:augmentedpPp}
 We define the following products on left by the kernel monomial linear combinations on the right:  
 \begin{align*}
   \lb V\rb\lb V_{a,b}\rb&\coloneqq\lb V(\pa_{\el}P_{a_1,b_1})V\cdot\!\cdot\!\cdot VP_{a_k,b_k}V\rb+\lb VP_{\hat{1}}V_{a,b}\rb,\\
   \lb V_{a,b}\rb\lb V\rb&\coloneqq\lb VP_{a_1,b_1}V\cdot\!\cdot\!\cdot V(\pa_{\rl}P_{a_k,b_k})V\rb+\lb V_{a,b}P_{\hat{1}}V\rb. 
 \end{align*}
The \emph{augmented pre-P\"oppe product} is the pre-P\"oppe product in Lemma~\ref{lemma:Poppeprodmonomials}, augmented to include these two cases.
\end{definition}
\begin{definition}[Augmented pre-P\"oppe algebra]\label{def:ApPA}
  The \emph{augmented pre-P\"oppe algebra}, denoted by $\mathbb{AP}$, is the real algebra of the forms~\eqref{eq:exmonomial} together with the monomial `$\lb V\rb$',
  equipped with the augmented pre-P\"oppe product. However, it is \emph{nonassociative}, see Example~\ref{ex:nonassoc} below.
\end{definition}  
\begin{remark}\label{rmk:nogoconnectall}
  Naturally, we observe that,
  \begin{equation*}
    A_{\rl}\circ\lb V_{a,b}\rb=\lb V_{a,b}\rb\,\lb V\rb\quad\text{and}\quad A_{\el}\circ\lb V_{a,b}\rb=\lb V\rb\,\lb V_{a,b}\rb. 
  \end{equation*}  
\end{remark}
We observe from the formulae in Remark~\ref{rmk:nogoconnectall} that the left and right multiplications in Definition~\ref{def:augmentedpPp} by $\lb V\rb$ in $\mathbb{AP}$,
precisely emulate the actions of $A_{\rl}$ and $A_{\el}$ in $\mathbb{P}_A$.
\begin{example}[$\mathbb{AP}$ is nonassociative]\label{ex:nonassoc}
  In $\mathbb{P}_A$ we observe that we have,
  \begin{align*}
    A_{\rl}\circ\lb VP_{\hat{1}}V\rb&=\lb V(\pa_{\rl}P_{\hat{1}})V\rb+\lb VP_{\hat{1}}VP_{\hat{1}}V\rb,\\
    A_{\el}\circ\lb VP_{\hat{1}}V\rb&=\lb V(\pa_{\el}P_{\hat{1}})V\rb+\lb VP_{\hat{1}}VP_{\hat{1}}V\rb.
  \end{align*}
  In $\mathbb{AP}$, since $\lb V\rb^2=\lb VP_{\hat 1}V\rb$, we observe that we have,
  \begin{align*}
    \bigl(\lb V\rb^2\bigr)\lb V\rb&=\lb V(\pa_{\rl}P_{\hat 1})V\rb+\lb VP_{\hat 1}VP_{\hat 1}V\rb,\\
    \lb V\rb\bigl(\lb V\rb^2\bigr)&=\lb V(\pa_{\el}P_{\hat 1})V\rb+\lb VP_{\hat 1}VP_{\hat 1}V\rb.   
  \end{align*}
  In particular, $A_{\rl}\circ\lb VP_{\hat{1}}V\rb\neq A_{\el}\circ\lb VP_{\hat{1}}V\rb$, and similarly, $\lb V\rb\bigl(\lb V\rb^2\bigr)\neq\bigl(\lb V\rb^2\bigr)\lb V\rb$.
  This means that the augmented pre-P\"oppe algebra $\mathbb{AP}$ is \emph{not associative}.
\end{example}
\begin{remark}[$\mathbb{P}$ is associative]
  The pre-P\"oppe algebra $\mathbb{P}$, which excludes the monomial `$\lb V\rb$', \emph{is associative}.
  This is because the pre-P\"oppe product given in Lemma~\ref{lemma:Poppeprodmonomials} is associative. 
  This is straightforwardly checked by computing, $\lb V_{a,b}\rb\,\bigl(\lb V_{a^\prime,b^\prime}\rb\,\lb V_{a^{\prime\prime},b^{\prime\prime}}\rb\bigr)$
  and $\bigl(\lb V_{a,b}\rb\,\lb V_{a^\prime,b^\prime}\rb\bigr)\,\lb V_{a^{\prime\prime},b^{\prime\prime}}\rb$, and comparing the two results. 
\end{remark}
As already hinted, we have the following important result.
\begin{lemma}\label{lemma:APandPAisomorphic}
  The the action bimodule $\mathbb{P}_A$ and the augmented pre-P\"oppe algebra $\mathbb{AP}$ are isomorphic, i.e.\/ we have,
\begin{equation*}
     \mathbb{P}_A\cong\mathbb{AP}.
\end{equation*}
\end{lemma}
\begin{proof}
  In view of the results established above, this is straightforward. We associate the actions $A_{\el}$ and $A_{\rl}$ of the action bimodule $\mathbb{P}_A$
  respectively with pre- and postmultiplication by $\lb V\rb$ in $\mathbb{AP}$. All other monomials of the form $\lb V_{a,b}\rb$ with $V_{a,b}\neq V$ are mapped to each other.
  This mapping is one-to-one. 
\qed
\end{proof}
Thus \emph{to summarise}, the pre-P\"oppe algebra, $\mathbb{P}$, in Definition~\ref{def:prePoppealg}, is noncommutative and associative.
The augmented pre-P\"oppe algebra, $\mathbb{AP}$, is noncommutative and nonassociative.
Further, the pre-P\"oppe algebra $\mathbb{P}$ is a natural subalgebra of the augmented pre-P\"oppe algebra $\mathbb{AP}$.
It is the subalgebra of $\mathbb{AP}$ in which we restrict ourselves to monomials of nonzero grade (see the next few sections), i.e.\/ we exclude the monomial `$\lb V\rb$'.
If we regard the incorporation of the actions $A_{\rl}$ and $A_{\el}$ on $\mathbb{P}$,
as `lifting' the pre-P\"oppe algebra $\mathbb{P}$ to the action bimodule, $\mathbb{P}_A$, then we have the following:
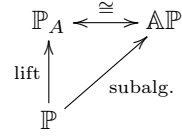
\begin{equation*}
\xymatrix{
  \ar@{<->}[r]^\cong \mathbb{P}_A& \mathbb{AP} \\
  \ar@{->}[u]^{\text{lift}}\mathbb{P}\ar@{->}[ur]_{\text{subalg.}} & }
\end{equation*}
The augmented pre-P\"oppe algebra is extremely rich in structure as we see in the coming sections.
It is also the context we use to solve the complete noncommutative KP hierarchy.
In particular, we establish that the augmented pre-P\"oppe algebra is isomorphic to the \emph{descent algebra} equipped with a so-called \emph{grafting product}.
To see that this is the case, the most edifying route requires us to consider several equivalent representations of the descent algebra;
as outlined briefly in the introduction and in detail in the following sections.
Indeed the first representation is that involving planar binary rooted trees in which at each grade, we equivalence specific subsets of elements.
In fact, we consider the algebra of planar rooted trees equipped with the grafting product, which is a natural algebra representation for a nonassociative noncommutative algebraic context. 
The equivalencing operation on the algebra of planar binary rooted trees is meant to model the fact that, in our context, there is some redundancy in the tree representation.
The second representation provides a natural encoding incorporating this redundancy. We call it the `left-glue-right' representation or `$\el\gl\rl$'-encoding for short.
This representation reveals the connection between the redundancy-free nonassociative algebra and the augmented pre-P\"oppe algebra. 
The third representation we use is that of the descent algebra itself, equipped, as mentioned, with the grafting product. 
This is the natural product induced from the algebra of planar binary rooted trees after equivalencing.
We also introduce a degrafting operator which we use to succinctly encode the expressions involved.
In Sections~\ref{sec:trees}--\ref{sec:degrafting} we introduce the algebraic structures we require, and then in Sections~\ref{sec:weightchar}--\ref{sec:solutions} we get down to solving the noncommutative KP hierarchy.
To round off this section, we give the following corollary to Theorem~\ref{thm:hierarchyequations}, demonstrating how, in the augmented pre-P\"oppe algebra context, assuming the ansatz~\eqref{eq:solutionansatz},  
the Sato coefficients $b_{n-k}$ are simply prescribed.
\begin{corollary}\label{cor:hierarchyequations}
  In the augmented pre-P\"oppe algebra, and assuming the ansatz~\eqref{eq:solutionansatz}, the Sato coefficients $b_{n-k}$ in Theorem~\ref{thm:hierarchyequations}, are given by,
  \begin{equation*}
    b_{n-k}=\biggl(C_{k-1}^n\pa_x^{k-1}+\sum_{\ell=3}^{k}C_{k-\ell}^{n-\ell+1}b_{n-\ell+1}\pa_x^{k-\ell}\biggr)\,\lb V\rb.
  \end{equation*}
  We emphasise that the final term is $C_0^{n-k+1}b_{n-k+1}\lb V\rb$. 
  The product in this term and in all other instances just above, is to be interpreted as the augmented pre-P\"oppe product.
  This formula also applies in the case $k=n+1$ corresponding to $b_{-1}=\pa_{t_n}w_1$.
\end{corollary}
\begin{proof}
  Explicitly, the coefficients $b_{n-k}$ from Theorem~\ref{thm:hierarchyequations} have the form,
  \begin{align*}
    b_{n-k}=&\;\bigl(C_{k-1}^n\pa_x^{k-1}+C_{k-3}^{n-2}b_{n-2}\pa_x^{k-3}+\cdots\\
           &\qquad\cdots+C_0^{n-k+1}b_{n-k+1}\bigr)w_1\\
           &\;+\bigl(C_{k-2}^n\pa_x^{k-2}+C_{k-4}^{n-2}b_{n-2}\pa_x^{k-4}+\cdots\\
           &\;\qquad\cdots+C_0^{n-k+2}b_{n-k+2}\bigr)w_2+\cdots\\
           &\;\cdots+\bigl(C_2^n\pa_x^2+C_0^{n-2}b_{n-2}\bigr)w_{k-2}+C_1^n\pa_x w_{k-1}\\
          =&\;-\bigl(C_{k-1}^n\pa_x^{k-1}+C_{k-3}^{n-2}b_{n-2}\pa_x^{k-3}+\cdots\\
           &\;\qquad\cdots+C_0^{n-k+1}b_{n-k+1}\bigr)\lb VP\rb-\pa_{\rl}b_{n-k+1}.
  \end{align*}
  Using the augmented pre-P\"oppe product and that $\pa_x\lb G\rb\equiv\pa_{x}\lb V\rb$ with $G=VP$, gives the result.
  The case of $b_{-1}$ is straightforwardly checked.
\qed
\end{proof}
\begin{remark}
  The perceptive reader will have noticed that in the Sato formulation of the noncommutative KP hierarchy above, we actually only needed the right action $A_{\rl}$.
  However, as we see in the next sections, to encode derivative terms such as $\pa_x^n\lb V\rb$,
  corresponding to the polynomials $\pf_n$ in Lemma~\ref{lemma:weightedsumofdescents} below, in the context of the augmented pre-P\"oppe algebra
  or equivalently the descent algebra as we are about to see, we also require the left action $A_{\el}$.
\end{remark}
\begin{remark}
  We also observe such left and right actions in the context of linear systems theory applied to the KP equation, see Blower and Malham~\cite{BM-KP-LS}.
\end{remark}

\section{Planar binary rooted trees}\label{sec:trees}
The archetypal nonassociative noncommutative algebra is the algebra of planar binary rooted trees. For further details on these see for example, Loday and Ronco~\cite{LR} and Foissy~\cite{Foissy}.
For some further background on the notation and concepts for planar binary rooted trees that we use herein, see Malham~\cite{M-coag}.
We observe that, using the augmented pre-P\"oppe product, we can generate operator kernel monomials of the form `$\lb VP_{a_1,b_1}VP_{a_2,b_2}V\cdots VP_{a_n,b_n}V\rb$' by computing products of $\lb V\rb$.
Thus for example, we have, $\lb V\rb^2=\lb VP_{\hat 1}V\rb$ and $\lb V\rb\bigl(\lb V\rb^2\bigr)$ and $\bigl(\lb V\rb^2\bigr)\lb V\rb$ generate the forms shown on the right-hand side in Example~\ref{ex:nonassoc}.
We can parametrise such products by planar binary rooted trees. Indeed the binary parenthesisation of strings is isomorphically mapped to planar binary trees, see Stanley~\cite[Ch.~6]{Stanley}.
We thus parametrise the first three basic products as follows,
\begin{equation*}
  \lb V\rb^2=\!{\tiny \begin{forest} for tree={grow'=90, parent anchor=center,child anchor=center, l=0cm,inner ysep=0cm,edge+=thick} [[][]] \end{forest}}\!,\quad
  \lb V\rb\bigl(\lb V\rb^2\bigr)=\!{\tiny \begin{forest} for tree={grow'=90, parent anchor=center,child anchor=center, l=0cm,inner ysep=0cm,edge+=thick} [[][[][]]] \end{forest}}\!\!,
  \quad\text{and}\quad
  \bigl(\lb V\rb^2\bigr)\lb V\rb=\!{\tiny \begin{forest} for tree={grow'=90, parent anchor=center,child anchor=center, l=0cm,inner ysep=0cm,edge+=thick} [[[][]][]] \end{forest}}\!.
\end{equation*}
The parametrisation of higher degree paranthesised products such as $\lb V\rb\bigl(\lb V\rb\bigl(\lb V\rb^2\bigr)\bigr)$ by planar binary rooted trees is straightforward.
This last example is given by the first tree of grade $3$ in Table~\ref{table:treesupto3}. Further examples of such trees are given in Tables~\ref{table:treesupto3} and \ref{table:treesof4}.
\begin{definition}[Grade of a tree]
  The \emph{grade}, $\|\tau\|$, of a planar binary rooted tree $\tau$, is the number of vertices
  `${\tiny \begin{forest} for tree={grow'=90, parent anchor=center,child anchor=center, l=0cm,inner ysep=0cm,edge+=thick} [[][]] \end{forest}}$'
  it contains.
\end{definition}
There are two natural procedures for generating planar binary rooted trees: grafting and branching. The induced grafting procedure underlies the product on the descent algebra we eventually use.
For the moment though, we focus on the branching procedure. This is outlined as follows.
All the trees of grade $n+1$ can be generated from those of grade $n$ by attaching a single branch
`${\tiny \begin{forest} for tree={grow'=90, parent anchor=center,child anchor=center, l=0cm,inner ysep=0cm,edge+=thick} [[][]] \end{forest}}$'
successively to each free end of each tree $\tau$ of grade $n$.
This procedure exhaustively generates all the trees of grade $n+1$, but some trees are multiply generated.
The multiplicity of any tree generated in this way is known as the \emph{weight character} of the tree.
Formally, see Malham~\cite{M-coag}, we define this as follows.
\begin{definition}[Weight character]\label{def:weightchar}
  We define the \emph{weight character} $\alpha$ of any planar binary rooted tree $\tau$ recursively as follows:
  $\alpha({\tiny \begin{forest} for tree={grow'=90, parent anchor=center,child anchor=center} [] \path[fill=black] (.anchor) circle[radius=1.5pt]; \end{forest}})\coloneqq1$,
  $\alpha(\raisebox{-1pt}{{\tiny \begin{forest} for tree={grow'=90, parent anchor=center,child anchor=center, l=0cm,inner ysep=0cm,edge+=thick} [[][]] \end{forest}}})=C^0_0$ and,
  \begin{equation*}
    \alpha\Bigl(\raisebox{-5pt}{{\scriptsize \begin{forest} for tree={grow'=90, parent anchor=center,child anchor=south, l=0cm,inner ysep=2pt,edge+=thick, s sep=0mm} [[$\tau_1$][$\tau_2$]] \end{forest}}}\Bigr)
    =C^{\|\tau_1\|+\|\tau_2\|}_{\|\tau_1\|} \alpha(\tau_1)\,\alpha(\tau_2).
  \end{equation*}
\end{definition}
The \emph{branching operator} actions the branching procedure as follows. Again, see Malham~\cite{M-coag}.
\begin{definition}[Branching operator]\label{def:branchingoperator}
  We define the linear \emph{branching operator} $B_{\vee}$ as the operator that acts on any tree $\tau$ by successively, additively attaching a branch
  `\raisebox{-1pt}{{\tiny \begin{forest} for tree={grow'=90, parent anchor=center,child anchor=center, l=0cm,inner ysep=0cm,edge+=thick} [[][]] \end{forest}}}'
  to each free end of the tree $\tau$, thus generating a sum of the corresponding trees at the next grade.
\end{definition}

\begin{example}\label{ex:branching}
  The action of the branching operator on some basic trees is illustrated as follows,
  \begin{align*}
    B_{\vee}({\tiny \begin{forest} for tree={grow'=90, parent anchor=center,child anchor=center} [] \path[fill=black] (.anchor) circle[radius=1.5pt]; \end{forest}})
    =&\;\raisebox{-1pt}{{\tiny \begin{forest} for tree={grow'=90, parent anchor=center,child anchor=center, l=0cm,inner ysep=0cm,edge+=thick} [[][]] \end{forest}}},\\    
    B_{\vee}(\raisebox{-1pt}{{\tiny \begin{forest} for tree={grow'=90, parent anchor=center,child anchor=center, l=0cm,inner ysep=0cm,edge+=thick} [[][]] \end{forest}}})
    =&\;\raisebox{-4pt}{{\tiny \begin{forest} for tree={grow'=90, parent anchor=center,child anchor=center, l=0cm,inner ysep=0cm,edge+=thick} [[][[][]]] \end{forest}}}
    +\raisebox{-4pt}{{\tiny \begin{forest} for tree={grow'=90, parent anchor=center,child anchor=center, l=0cm,inner ysep=0cm,edge+=thick} [[[][]][]] \end{forest}}},\\
    B_{\vee}^2(\raisebox{-1pt}{{\tiny \begin{forest} for tree={grow'=90, parent anchor=center,child anchor=center, l=0cm,inner ysep=0cm,edge+=thick} [[][]] \end{forest}}})
    =&\;B_\vee\Bigl(\raisebox{-4pt}{{\tiny \begin{forest} for tree={grow'=90, parent anchor=center,child anchor=center, l=0cm,inner ysep=0cm,edge+=thick} [[][[][]]] \end{forest}}}
    +\raisebox{-4pt}{{\tiny \begin{forest} for tree={grow'=90, parent anchor=center,child anchor=center, l=0cm,inner ysep=0cm,edge+=thick} [[[][]][]] \end{forest}}}\Bigr)\\
    =&\;\raisebox{-8pt}{{\tiny \begin{forest} for tree={grow'=90, parent anchor=center,child anchor=center, l=0cm,inner ysep=0cm,edge+=thick} [[][[][[][]]]] \end{forest}}}
    +\raisebox{-8pt}{{\tiny \begin{forest} for tree={grow'=90, parent anchor=center,child anchor=center, l=0cm,inner ysep=0cm,edge+=thick} [[][[[][]][]]] \end{forest}}}
    +2\cdot\raisebox{-8pt}{{\tiny \begin{forest} for tree={grow'=90, parent anchor=center,child anchor=center, l=0cm,inner ysep=0cm,edge+=thick} [[[][]][[][]]] \end{forest}}}
    +\raisebox{-8pt}{{\tiny \begin{forest} for tree={grow'=90, parent anchor=center,child anchor=center, l=0cm,inner ysep=0cm,edge+=thick} [[[][[][]]][]] \end{forest}}}
    +\raisebox{-8pt}{{\tiny \begin{forest} for tree={grow'=90, parent anchor=center,child anchor=center, l=0cm,inner ysep=0cm,edge+=thick} [[[[][]][]][]] \end{forest}}}.
  \end{align*}
\end{example}

\begin{table}
  \caption{We list all the rooted planar binary trees up to grade $3$. For the grades in the left column, we list each tree $\tau\in\mathbb T$ in the second column,
           its word-coding in terms of levels in the third column, its `$\el\gl\rl$'-code in the fourth column, and its descent set (see Section~\ref{sec:descent-alg}) in the final column.}
\label{table:treesupto3}
\begin{center}
   \setlength{\baselineskip}{2\baselineskip}
\begin{tabular}{c|c|c|c|c}
\hline
$\phantom{\biggl|}$ grade $\phantom{\biggl|}$ & tree & levels & `$\el\gl\rl$'-code & descent\\ \hline
\hline
0 & {\tiny \begin{forest} for tree={grow'=90, parent anchor=center,child anchor=center} [] \path[fill=black] (.anchor) circle[radius=1.5pt]; \end{forest}} & 0 & $-$ & $\emptyset_0$  \\[5pt]
1 & {\tiny \begin{forest} for tree={grow'=90, parent anchor=center,child anchor=center, l=0cm,inner ysep=0cm,edge+=thick} [[][]] \end{forest}}             & 1 & $\gl$ & $\emptyset_1$ \\\hline
&&&&\\[-5pt]
2 & {\tiny \begin{forest} for tree={grow'=90, parent anchor=center,child anchor=center, l=0cm,inner ysep=0cm,edge+=thick} [[][[][]]] \end{forest}}         & 12 & $(\el+\gl)\gl$ & $\emptyset_2$ \\
&&&&\\[-5pt]
2 & {\tiny \begin{forest} for tree={grow'=90, parent anchor=center,child anchor=center, l=0cm,inner ysep=0cm,edge+=thick} [[[][]][]] \end{forest}}         & 21 & $\gl(\rl+\gl)$  & $1_2$ \\\hline
&&&&\\[-5pt]
3 & {\tiny \begin{forest} for tree={grow'=90, parent anchor=center,child anchor=center, l=0cm,inner ysep=0cm,edge+=thick} [[][[][[][]]]] \end{forest}}     & 123 & $(\el+\gl)^2\gl$ & $\emptyset_3$ \\
3 & {\tiny \begin{forest} for tree={grow'=90, parent anchor=center,child anchor=center, l=0cm,inner ysep=0cm,edge+=thick} [[][[[][]][]]] \end{forest}}     & 132 & $(\el+\gl)\gl(\rl+\gl)$  & $2_3$ \\[7pt]
3 & {\tiny \begin{forest} for tree={grow'=90, parent anchor=center,child anchor=center, l=0cm,inner ysep=0cm,edge+=thick} [[[][]][[][]]] \end{forest}}     & 212 & $\gl(\el+\gl+\rl)\gl$  & $1_3$ \\
3 & {\tiny \begin{forest} for tree={grow'=90, parent anchor=center,child anchor=center, l=0cm,inner ysep=0cm,edge+=thick} [[[][[][]]][]] \end{forest}}     & 231 & $(\el+\gl)\gl(\rl+\gl)$  & $2_3$ \\
3 & {\tiny \begin{forest} for tree={grow'=90, parent anchor=center,child anchor=center, l=0cm,inner ysep=0cm,edge+=thick} [[[[][]][]][]] \end{forest}}     & 321 & $\gl(\rl+\gl)^2$ & $12_3$ \\\hline
\end{tabular}
\end{center}
\end{table}

\begin{definition}[Weighted sum of binary trees]\label{def:weightedpolytrees}
  A natural linear combination of planar binary rooted trees of grade $n$, is the following weight sum,
  \begin{equation}\label{eq:weightedpolytrees}
    \pf_n\coloneqq\sum_{\|\tau\|=n}\alpha(\tau)\cdot\tau.
  \end{equation}
\end{definition}
We have the following result. We set $\mathbb N_0\coloneqq\mathbb N\cup\{0\}$.
\begin{lemma}[Derivatives and trees]\label{lemma:derivativesandtrees}
  For any $n\in\mathbb N_0$, we have,
  \begin{equation*}
        \pa_x\pf_n\equiv B_{\vee}(\pf_n)\equiv\pf_{n+1}.
  \end{equation*}
\end{lemma}
\begin{proof}
  We have already seen that $\pa_x\lb V\rb\equiv \lb VP_{\hat 1}V\rb\equiv\lb V\rb^2$.
  In other words if take the $\pa_x$-derivative of any parenthesised monomial involving $\lb V\rb$, using the usual product rule,
  we successively and additively replace each instance of the derivative `$\pa_x\lb V\rb$' by $\lb V\rb^2$.
  This is equivalent to adding a branch `${\tiny \begin{forest} for tree={grow'=90, parent anchor=center,child anchor=center, l=0cm,inner ysep=0cm,edge+=thick} [[][]] \end{forest}}$'
  successively and additively to each free end of the corresponding tree $\tau$. Applying this for each tree in $\pf_n$, establishes the result.\qed
\end{proof}
Recall that the tree `${\tiny \begin{forest} for tree={grow'=90, parent anchor=center,child anchor=center, l=0cm,inner ysep=0cm,edge+=thick} [[][]] \end{forest}}$' corresponds to the monomial $\lb V\rb^2$.
Hence $\pa_x\lb V\rb=\lb V\rb^2$, which corresponds to $\pf_1$, and $\pa_x\bigl(\lb V\rb^2\bigr)=\pa_x^2\lb V\rb=\bigl(\lb V\rb^2\bigr)\lb V\rb+\lb V\rb\bigl(\lb V\rb^2\bigr)$, which corresponds to $\pf_2$. 
Hence as a result of Lemma~\ref{lemma:derivativesandtrees}, we have the following immediate corollary.
\begin{corollary}\label{cor:derivsandtrees}
  For any $n\in\mathbb N\cup\{0\}$, we have,
   \begin{equation*}
        \pa_x^n\lb V\rb=\pf_{n+1}\bigl(\lb V\rb\bigr).
   \end{equation*}
\end{corollary}
\begin{remark}[Augmented pre-P\"oppe algebra]\label{rmk:nonassocpPa}
  The statement of Corollary~\ref{cor:derivsandtrees} is just the natural nonassociative equivalent result to that for $\pa_x^n\lb V\rb$ in~\eqref{eq:expansion}.
  From another perspective, it shows us that $\pa_x^n\lb V\rb$ can be expressed as a degree $n+1$ polynomial in the nonassociative augmented pre-P\"oppe algebra.
\end{remark}

\begin{table}
  \caption{We list all the rooted planar binary trees of grade $4$. For each tree $\tau\in\mathbb T$ in the left column,
           we give its word-coding in terms of levels in the second column, its `$\el\gl\rl$'-code in the third column, and its descent set (see Section~\ref{sec:descent-alg}) in the final column.}
\label{table:treesof4}
\begin{center}
  \setlength{\baselineskip}{2\baselineskip}
\begin{tabular}{c|c|c|c}
\hline
$\phantom{\biggl|}$ tree $\phantom{\biggl|}$ & levels & `$\el\gl\rl$'-code & descent\\ \hline
&&&\\[-5pt]
{\tiny \begin{forest} for tree={grow'=90, parent anchor=center,child anchor=center, l=0cm,inner ysep=0cm,edge+=thick} [[][[][[][[][]]]]] \end{forest}} & 1234 & $(\el+\gl)^3\gl$  & $\emptyset_4$\\
{\tiny \begin{forest} for tree={grow'=90, parent anchor=center,child anchor=center, l=0cm,inner ysep=0cm,edge+=thick} [[][[][[[][]][]]]] \end{forest}} & 1243 & $(\el+\gl)^2\gl(\rl+\gl)$  & $3_4$\\[5pt]
{\tiny \begin{forest} for tree={grow'=90, parent anchor=center,child anchor=center, l=0cm,inner ysep=0cm,edge+=thick} [[][[[][]][[][]]]] \end{forest}} & 1323 & $(\el+\gl)\gl(\el+\gl+\rl)\gl$ & $2_4$\\
{\tiny \begin{forest} for tree={grow'=90, parent anchor=center,child anchor=center, l=0cm,inner ysep=0cm,edge+=thick} [[][[[][[][]]][]]] \end{forest}} & 1342 & $(\el+\gl)^2\gl(\rl+\gl)$ & $3_4$\\
{\tiny \begin{forest} for tree={grow'=90, parent anchor=center,child anchor=center, l=0cm,inner ysep=0cm,edge+=thick} [[][[[[][]][]][]]] \end{forest}} & 1432 & $(\el+\gl)\gl(\rl+\gl)^2$ & $23_4$\\[5pt]
{\tiny \begin{forest} for tree={grow'=90, parent anchor=center,child anchor=center, l=0cm,inner ysep=0cm,edge+=thick} [[[][]][[][[][]]]] \end{forest}} & 2123 & $\gl(\el+\gl+\rl)(\el+\gl)\gl$ & $1_4$\\[5pt]
{\tiny \begin{forest} for tree={grow'=90, parent anchor=center,child anchor=center, l=0cm,inner ysep=0cm,edge+=thick} [[[][]][[[][]][]]] \end{forest}} & 2132 & $\gl(\el+\gl+\rl)\gl(\rl+\gl)$ & $13_4$ \\[5pt]
{\tiny \begin{forest} for tree={grow'=90, parent anchor=center,child anchor=center, l=0cm,inner ysep=0cm,edge+=thick} [[[][[][]]][[][]]] \end{forest}} & 2312 & $(\el+\gl)\gl(\el+\gl+\rl)\gl$ & $2_4$\\[5pt]
{\tiny \begin{forest} for tree={grow'=90, parent anchor=center,child anchor=center, l=0cm,inner ysep=0cm,edge+=thick} [[[[][]][]][[][]]] \end{forest}} & 3212 & $\gl(\rl+\gl)(\el+\gl+\rl)\gl$ & $12_4$\\
{\tiny \begin{forest} for tree={grow'=90, parent anchor=center,child anchor=center, l=0cm,inner ysep=0cm,edge+=thick} [[[][[][[][]]]][]] \end{forest}} & 2341 & $(\el+\gl)^2\gl(\rl+\gl)$ & $3_4$\\
{\tiny \begin{forest} for tree={grow'=90, parent anchor=center,child anchor=center, l=0cm,inner ysep=0cm,edge+=thick} [[[][[[][]][]]][]] \end{forest}} & 2431 & $(\el+\gl)\gl(\rl+\gl)^2$ & $23_4$\\[5pt]
{\tiny \begin{forest} for tree={grow'=90, parent anchor=center,child anchor=center, l=0cm,inner ysep=0cm,edge+=thick} [[[[][]][[][]]][]] \end{forest}} & 3231 & $\gl(\el+\gl+\rl)\gl(\rl+\gl)$ & $13_4$\\
{\tiny \begin{forest} for tree={grow'=90, parent anchor=center,child anchor=center, l=0cm,inner ysep=0cm,edge+=thick} [[[[][[][]]][]][]] \end{forest}} & 3421 & $(\el+\gl)\gl(\rl+\gl)^2$ & $23_4$\\
{\tiny \begin{forest} for tree={grow'=90, parent anchor=center,child anchor=center, l=0cm,inner ysep=0cm,edge+=thick} [[[[[][]][]][]][]] \end{forest}} & 4321 & $\gl(\rl+\gl)^3$ & $123_4$\\
\hline
\end{tabular}
\end{center}
\end{table}

Let $\mathbb T$ denote the set of planar binary rooted trees.
We can define a real algebra of planar binary rooted trees $\R\langle\mathbb T\rangle$, equipped with a natural product, the so-called \emph{grafting product}.
See Malham~\cite{M-coag} for more details. 
The grafting product, which generates new planar binary rooted trees by grafting two trees together at their roots, is defined as follows.
\begin{definition}[Grafting product on trees]\label{def:graftingontrees}
  For two planar binary rooted trees $\tau_1$ and $\tau_2$, such that $\|\tau_1\|+\|\tau_2\|=n$, the \emph{grafting product} of these two trees is the tree constructed as follows, 
\begin{equation*}
     {\scriptsize \begin{forest} for tree={grow'=90, parent anchor=center,child anchor=south, l=0cm,inner ysep=2pt,edge+=thick, s sep=0mm} [[$\tau_1$][$\tau_2$]] \end{forest}}.
\end{equation*}
The tree generated in this way is of grade $n+1$.
\end{definition}
\begin{example}[Nonassociativity]\label{ex:graftingnonassoctrees}
  Suppose $\tau_1={\tiny \begin{forest} for tree={grow'=90, parent anchor=center,child anchor=center, l=0cm,inner ysep=0cm,edge+=thick} [[][]] \end{forest}}$, the single vertex tree,
  and $\tau_0={\tiny \begin{forest} for tree={grow'=90, parent anchor=center,child anchor=center} [] \path[fill=black] (.anchor) circle[radius=1.5pt]; \end{forest}}$, the single `bud'.
  Suppose we use `$\star$' to denote the grafting product between trees---this is the notation we use below. Then we observe that,
  \begin{equation*}
    \tau_0\star\tau_0={\tiny \begin{forest} for tree={grow'=90, parent anchor=center,child anchor=center, l=0cm,inner ysep=0cm,edge+=thick} [[][]] \end{forest}}=\tau_1
    \qquad\text{and}\qquad  \tau_1\star\tau_0={\tiny \begin{forest} for tree={grow'=90, parent anchor=center,child anchor=center, l=0cm,inner ysep=0cm,edge+=thick} [[[][]][]] \end{forest}}.
  \end{equation*}
  Now consider the following triple products. We observe,
  \begin{align*}
    \tau_1\star(\tau_0\star\tau_0)&=\tau_1\star\tau_1={\tiny \begin{forest} for tree={grow'=90, parent anchor=center,child anchor=center, l=0cm,inner ysep=0cm,edge+=thick} [[[][]][[][]]] \end{forest}},
    \intertext{whereas,}
    (\tau_1\star\tau_0)\star\tau_0&={\tiny \begin{forest} for tree={grow'=90, parent anchor=center,child anchor=center, l=0cm,inner ysep=0cm,edge+=thick} [[[][]][]] \end{forest}}\star\tau_0
    ={\tiny \begin{forest} for tree={grow'=90, parent anchor=center,child anchor=center, l=0cm,inner ysep=0cm,edge+=thick} [[[[][]][]][]] \end{forest}}.
  \end{align*}
  This reflects the nonassociativity of the grafting product.
\end{example}
\begin{remark}[Grafting and concatenation]\label{rmk:graftingandconcat}
  We note that the grafting product is very natural, and is akin to a concatenation product of parenthesised strings; see Lundervold and Munthe--Kaas~\cite{LM-Kbackward}.
\end{remark}

There is one further aspect of planar binary rooted trees we require, and that is their \emph{level} or \emph{word} coding, which represents an equivalent representation.
Again, more details on this can be found in Loday and Ronco~\cite{LR} and Malham~\cite[\S4]{M-coag}.
Also see Aguiar and Sottile~\cite{AguiarSottile}, Arcis and M\'arquez~\cite{Arcis}, Chapoton~\cite{Chapoton} and Chatel and Pilaud~\cite{CP}.
We can assign words according to the position of vertices for any planar binary rooted tree as as follows. 
Consider, for example, the trees of grade $4$ listed in Table~\ref{table:treesof4}.
Consider the tree labelled $2132$. This word-coding is generated by recording the positions of the vertices at the different levels from the root up.
The `$1$' indicates the first level and bottom vertex. There is only a single `$1$' in any such code.
The `$2$' either side of the `$1$' indicate that there is a vertex attached to the left branch of the level-one vertex, as well as a vertex attached to the right branch of the that level-one vertex. 
Thus any such word-coding has either none, one or two $2$'s. If there were no `$2$', that would mean that the tree in question is just the grade $1$ tree.
The case of only one `$2$' to the left of the `$1$' would mean that there is only further tree structure attached to the left branch of the level-one vertex,
with the analogous interpretation if there is on a single `$2$' to the right of level-one vertex. 
For our word-coding $2132$, the fact that there is single `$3$' to the right of the `$2$' which is to the left of the `$1$',
means that there is one more vertex attached to the right branch of the vertex at level two, which is to the left of the level-one vertex.
Consider another example, say $3421$. This means that there is only some further tree structure attached to the left branch of the bottom level-one vertex,
as there is a single `$2$' to the left of the `$1$'. That there is a single `$3$' to the left of the `$2$' means there is a vertex at level three attached to
the left branch of the vertex at level two. The single `$4$' to the right of the `$3$' means that there is a level four vertex attached to the right branch of the level three vertex. 
And so forth. The set of rooted planar binary trees can then be ordered by the numerical ordering, within the set of integers, of the word-codes. 
See the ordered lists in Tables~\ref{table:treesupto3} and \ref{table:treesof4}, and Example~\ref{ex:wordordering} below.
\begin{example}\label{ex:wordordering}
  Two further illustrative examples of the word-coding for planar binary rooted trees are as follows. First we have,
  \begin{equation*}
    {\tiny \begin{forest} for tree={grow'=90, parent anchor=center,child anchor=center, l=0cm,inner ysep=0cm,edge+=thick} [[[][]][[[][]][[][]]]] \end{forest}}=21323.
  \end{equation*}
  The corresponding `$\el\gl\rl$' and descent representations for this tree are $\mathrm{g(\el\!+\!r\!+\!g)g(\el\!+\!r\!+\!g)g}$ and $13_5$, respectively.
  We introduce these representations respectively in Sections \ref{sec:lgr-alg} and \ref{sec:descent-alg} just below. Second we have,
  \begin{equation*}
    {\tiny \begin{forest} for tree={grow'=90, parent anchor=center,child anchor=center, l=0cm,inner ysep=0cm,edge+=thick} [[[[][]][[[][]][[][[][]]]]][[][[[][]][[][[][]]]]]] \end{forest}}=324345124345.
  \end{equation*}
  The `$\el\gl\rl$' and descent codes for this tree are,
  \begin{equation*}
    \mathrm{g(\el\!+\!r\!+\!g)g(\el\!+\!r\!+\!g)(\el\!+\!g)g(\el\!+\!r\!+\!g)(\el\!+\!g)g(\el\!+\!r\!+\!g)(\el\!+\!g)g},
  \end{equation*}
  and $1369_{12}$, respectively.
\end{example}

\section{The `left-glue-right' algebra}\label{sec:lgr-alg}
The augmented pre-P\"oppe product is nonassociative. However, there is some redundancy in the planar binary rooted tree representation/encoding.
The source of the nonassociativity present in algebra of operator kernel monomials `$\lb VP_{a_1,b_1}V\cdots VP_{a_n,b_n}V\rb$' generated by the augmented pre-P\"oppe product,
is due to left or right factors involving `$\lb V\rb$' and that pre- or postmultiplication by such factors generates different counts of left and right derivatives in the resulting kernel monomials.
Products involving factors of the form outlined in Lemma~\ref{lemma:Poppeprodmonomials} for the pre-P\"oppe product are associative.
Thus some parenthesised products involving `$\lb V\rb$', though they are parametrised by different trees, generate the same linear combination of kernel monomials of the form `$\lb VP_{a_1,b_1}V\cdots VP_{a_n,b_n}V\rb$'.
In other words, the source of this redundancy is due to the fact that several different parenthesised products can 
generate the same linear combinations of operator kernel monomials with the exact same number of left and right derivatives on the appropriate $P$ operators in each monomial contained therein. 
For example, consider the three planar binary rooted trees in Table~\ref{table:treesof4} parametrised by the word-codings `$1243$', `$1342$' and `$2341$'.
These respectively correspond to the parenthesised products as follows,
\begin{align*}
  &\lb V\rb\cdot\Bigl(\lb V\rb\cdot\bigl(\lb V\rb^2\cdot\lb V\rb\bigr)\Bigr),\\
  &\lb V\rb\cdot\Bigl(\bigl(\lb V\rb\cdot\lb V\rb^2\bigr)\cdot\lb V\rb\Bigr),\\
  &\Bigl(\lb V\rb\cdot\bigl(\lb V\rb\cdot\lb V\rb^2\bigr)\Bigr)\cdot\lb V\rb.
\end{align*}
It is straightforward to check that, if we use the augmented pre-P\"oppe product, each of these parenthesised products generates the same linear combination of kernel monomials.
More precisely, the result is, 
\begin{align}
  &\lb V(\pa_{\el}^2\pa_{\rl}P_{\hat 1})V\rb+\lb V(\pa_{\el}P_{\hat 1})V(\pa_{\rl}P_{\hat 1})V\rb\nonumber\\
  &\;+\lb VP_{\hat 1}V(\pa_{\el}\pa_{\rl}P_{\hat 1})V\rb+\lb VP_{\hat 1} VP_{\hat 1}V(\pa_{\el}P_{\hat 1})V\rb\nonumber\\
  &\;+\lb V(\pa_{\el}^2P_{\hat 1})VP_{\hat 1}V\rb+\lb V(\pa_{\el}P_{\hat 1})VP_{\hat 1} VP_{\hat 1}V\rb\nonumber\\
  &\;+\lb VP_{\hat 1}V(\pa_{\el}P_{\hat 1}) VP_{\hat 1}V\rb+\lb VP_{\hat 1}VP_{\hat 1}VP_{\hat 1}VP_{\hat 1}V\rb.\label{eq:seqexp}
\end{align}
There are several further cases like this in Table~\ref{table:treesof4}, though only one in Table~\ref{table:treesupto3}, namely for the two planar binary trees
with the word-coding `$132$' and `$231$'. That the parenthesised products corresponding to these two trees generate the same linear combination of kernel monomials,
is straightforwardly checked by direct computation. We now introduce an encoding for these `equivalent' trees as follows.
We call this the `left-glue-right' encoding or $\el\gl\rl$-coding for short. We introduce this encoding via some explanatory examples as follows.
\begin{example}\label{ex:firstconstruct}
Consider the tree with the representation,
\begin{equation*}
{\tiny \begin{forest} for tree={grow'=90, parent anchor=center,child anchor=center, l=0cm,inner ysep=0cm,edge+=thick} [[][[][[[][]][]]]] \end{forest}}=1243.
\end{equation*}
Now consider the following construction sequence for the parenthesised product it represents,
\begin{align*}
  \lb V\rb^2 \to\lb V\rb^2\cdot\lb V\rb&\to \lb V\rb\cdot\bigl(\lb V\rb^2\cdot\lb V\rb\bigr)\\
  &\to\lb V\rb\cdot\bigl(\lb V\rb\cdot\bigl(\lb V\rb^2\cdot\lb V\rb\bigr)\bigr).
\end{align*}
We mimic this construction sequence as follows: 
\begin{equation*}
  \gl\to\gl(\rl+\gl)\to(\el+\gl)\gl(\rl+\gl)\to(\el+\gl)^2\gl(\rl+\gl).
\end{equation*}
Superficially, the connection between these two constructions can be interpreted as follows.
In the context of operator kernels we observe that the `$\gl$' represents the form $\lb VP_{\hat 1}V\rb\equiv\lb V\rb^2$.  
Then postmultiplication by $\lb V\rb$ corresponds to postmultiplication by `$(\rl+\gl)$', whilst premultiplication by $\lb V\rb$ corresponds to premultiplication by $(\el+\gl)$.
This generates the `$\el\gl\rl$'-form shown. However, at a deeper level, the pre- and postmultiplication by these factors is recording the actions implicit
in the augmented pre-P\"oppe products in Definition~\ref{def:augmentedpPp}. In postmultiplication by $\lb V\rb$ in Definition~\ref{def:augmentedpPp}, we take the form $\lb V_{a,b}\rb$
and add an additional $\pa_{\rl}$-derivative to $P_{a_k,b_k}$---this corresponds to the `$\rl$' term in the factor `$(\rl+\gl)$'. However there is an additional term on the right
in which we \emph{glue} and extra $P_{\hat 1}V$ term on the right of $V_{a,b}$---this corresponds to the `$\gl$' term in the factor `$(\rl+\gl)$'.
Premultiplication by $\lb V\rb$ in Definition~\ref{def:augmentedpPp} corresponds to the prefactor `$(\el+\gl)$' with the `$\el$' term recording
the term on the right with an additional $\pa_{\el}$-derivative applied to $P_{a_1,b_1}$,
whilst the `$\gl$' term records the second term on the right in which a prefactor $VP_{\hat 1}$ is glued to $V_{a,b}$.
\end{example}
This example explains most of the `$\el\gl\rl$'-encodings listed in column four in Table~\ref{table:treesupto3} and column three in Table~\ref{table:treesof4}, but not all of them.
This is because we have not as yet encoded the standard pre-P\"oppe product given in Lemma~\ref{lemma:Poppeprodmonomials} into this representation.
However, this is straightforward, as the following example demonstrates.
\begin{example}
  Now consider the tree of grade $3$ in Table~\ref{table:treesupto3} with the representation,
  \begin{equation*}
   {\tiny \begin{forest} for tree={grow'=90, parent anchor=center,child anchor=center, l=0cm,inner ysep=0cm,edge+=thick} [[[][]][[][]]] \end{forest}} =212.
  \end{equation*}
  This represents the parenthesised product $\bigl(\lb V\rb^2\bigr)\bigl(\lb V\rb^2\bigr)=\lb VP_{\hat 1}V\rb\cdot\lb VP_{\hat 1}V\rb$.
  This is just the standard pre-P\"oppe product in Lemma~\ref{lemma:Poppeprodmonomials} between $\lb VP_{\hat 1}V\rb$ and itself, which generates three terms.
  Recall that the term $\lb VP_{\hat 1}V\rb\equiv \lb V\rb^2$ corresponds to the term `$\gl$' in the $\el\gl\rl$-encoding. 
  The first term on the right in product in Lemma~\ref{lemma:Poppeprodmonomials} corresponds to $\lb VP_{\hat 1}V\rb$ but with an additional $\pa_{\rl}$-derivative applied to the $P_{\hat 1}$ factor.
  This corresponds to the `$\rl$' term in the central factor in the $\el\gl\rl$-encoding `$\gl(\el+\gl+\rl)\gl$' for this tree.
  The second term corresponds to $\lb VP_{\hat 1}V\rb$ but with an additional $\pa_{\el}$-derivative applied to the $P_{\hat 1}$ factor.
  This corresponds to the `$\el$' term in the central factor in the encoding `$\gl(\el+\gl+\rl)\gl$'.
  Finally, the third term corresponds to the scenario when an extra factor $P_{\hat 1}$ is glued in between the other two factors in the original product.
  This corresponds to the `$\gl$' term in the central factor.
\end{example}
That a central factor `$(\el+\gl+\rl)$' in the $\el\gl\rl$-encoding corresponds to an instance of the standard pre-P\"oppe product in Lemma~\ref{lemma:Poppeprodmonomials}, is a straightforward step.
The $\el\gl\rl$-codes shown in Tables~\ref{table:treesupto3} and \ref{table:treesof4}, as well as in Example~\ref{ex:wordordering}, are now also straightforward.
One issue remains. That is the equivalence of some parenthesised products. The following example, which extends Example~\ref{ex:firstconstruct}, illustrates this.
\begin{example}\label{ex:secondandthirdconstructs}
  Following on from Example~\ref{ex:firstconstruct}, consider the two grade $4$ trees in Table~\ref{table:treesof4} given by,
  \begin{equation*}
    {\tiny \begin{forest} for tree={grow'=90, parent anchor=center,child anchor=center, l=0cm,inner ysep=0cm,edge+=thick} [[][[[][[][]]][]]] \end{forest}}=1342
     \qquad\text{and}\qquad
    {\tiny \begin{forest} for tree={grow'=90, parent anchor=center,child anchor=center, l=0cm,inner ysep=0cm,edge+=thick} [[[][[][[][]]]][]] \end{forest}}=2341.
  \end{equation*}
  Let us consider the construction sequence for the parenthesised product corresponding to the first tree, $1342$. The construction sequence would be,
  \begin{align*}
    \lb V\rb^2 \to\lb V\rb\cdot\lb V\rb^2&\to \bigl(\lb V\rb\cdot\lb V\rb^2\bigr)\cdot\lb V\rb\\
    &\to\lb V\rb\cdot\bigl(\bigl(\lb V\rb\cdot\lb V\rb^2\bigr)\cdot\lb V\rb\bigr).
  \end{align*}
  We can mimic this construction sequence as follows,
  \begin{equation*}
    \gl\to(\el+\gl)\gl\to(\el+\gl)\gl(\rl+\gl)\to(\el+\gl)^2\gl(\rl+\gl).
  \end{equation*}
  We observe that this is the same $\el\gl\rl$-encoding as that for the tree $1243$ in Example~\ref{ex:firstconstruct}.
  Now consider the construction sequence for the parenthesised product corresponding to the second tree, $2341$. The construction sequence is,
  \begin{align*}
    \lb V\rb^2 \to\lb V\rb\cdot\lb V\rb^2&\to \lb V\rb\cdot\bigl(\lb V\rb\cdot\lb V\rb^2\bigr)\\
    &\to\bigl(\lb V\rb\cdot\bigl(\lb V\rb\cdot\lb V\rb^2\bigr)\bigr)\cdot\lb V\rb.
  \end{align*}
  Mimicking this construction by $\el\gl\rl$-encoding, generates,
   \begin{equation*}
    \gl\to(\el+\gl)\gl\to(\el+\gl)^2\gl\to(\el+\gl)^2\gl(\rl+\gl).
  \end{equation*}
  Again, the final $\el\gl\rl$-encoding is the same as that for the other two trees, $1243$ and $1342$.
\end{example}
We observe that in Table~\ref{table:treesupto3}, in which there is a single instance, and in Table~\ref{table:treesof4}, where there are multiple examples,
there are several sets of trees with the same $\el\gl\rl$-encoding. 

Before we formally define the $\el\gl\rl$-algebra, there are some points that we need to elucidate.
The $\el\gl\rl$-codes record the cumulative effect of left and right $\lb V\rb$-factors as well as any standard pre-P\"oppe products.
We illustrate these actions via an example presently.
The natural product on the $\el\gl\rl$-algebra is the product equivalent to \emph{grafting}---see Definition~\ref{def:lgrgrafting} below.
And, lastly, the $\el\gl\rl$-algebra, as the name suggests, is a real polynomial algebra \emph{involving} the alphabet $\{\el,\gl,\rl\}$,
however it is \emph{generated} by the four factors given in Definition~\ref{def:lgr-alg} below.
We illustrate this last point and its consequences via the following example.
\begin{example}\label{ex:illus}
Consider expanding the $\el\gl\rl$-code given by `$(\el+\gl)^2\gl(\rl+\gl)$', corresponding to the trees $1243$, $1342$ and $2341$. See Examples~\ref{ex:firstconstruct} and \ref{ex:secondandthirdconstructs}.
The noncommutative expansion is given by,
\begin{equation}
\el^2\gl\rl+\el\gl^2\rl+\gl\el\gl\rl+\gl^3\rl+\el^2\gl^2+\el\gl^3+\gl\el\gl^2+\gl^4. \label{eq:lgrexp}
\end{equation}
Let us now successively interpret each of the terms in~\eqref{eq:lgrexp}.
The term $\el^2\gl\rl$, recalling that $\gl$ corresponds to $\lb VP_{\hat 1}V\rb$, records that two $\pa_{\el}$-derivatives should be applied to $\gl$, as well as a single $\pa_{\rl}$-derivative.
Note that we interpret the juxtaposition of the letters as follows.
Any powers of `$\el$' apply to the `$\gl$' immediately to its right, whilst any powers of `$\rl$'  apply to the `$\gl$' immediately to its left.
This means that $\el^2\gl\rl$ corresponds to the kernel term $\lb V\pa_{\el}^2\pa_{\rl}P_{\hat 1}V\rb$.
Now consider the term $\el\gl^2\rl$. We might more conveniently express this as $(\el\gl)(\gl\rl)$, with our interpretation that powers of `$\el$' apply to the `$\gl$' immediately to its right, and vice versa for $\rl$, in mind. 
The term $\el\gl^2\rl$ thus corresponds to $\lb V(\pa_{\el}P_{\hat 1})V(\pa_{\rl}P_{\hat 1})V\rb$, whereas $\gl\el\gl\rl$ corresponds to $\lb VP_{\hat 1}V(\pa_{\el}\pa_{\rl}P_{\hat 1})V\rb$ and
$\gl^3\rl$ corresponds to $\lb VP_{\hat 1}VP_{\hat 1}V(\pa_{\rl}P_{\hat 1})V\rb$. And so forth.
We observe the natural one-to-one correspondence between the successive terms in \eqref{eq:lgrexp} and those in \eqref{eq:seqexp}.
\end{example}
Example~\ref{ex:illus} illustrates both, that the expansions of the $\el\gl\rl$-codes, which were developed to record all the actions of the augmented pre-P\"oppe product, 
naturally retain this property when they are expanded as we would expect. 
However, Example~\ref{ex:illus} also demonstrates that the juxtaposition of the letters $\{\el,\gl,\rl\}$ is important.
Hence, as we have observed, each planar binary rooted tree has an $\el\gl\rl$-coding, though several trees of the same grade may share the same $\el\gl\rl$-coding. 
Each $\el\gl\rl$-coding involves the factors $\gl$, $(\el+\gl)$, $(\rl+\gl)$ and $(\el+\gl+\rl)$.
The relative juxtaposition of these factors depends on the tree and is not free, as is illustated in Tables~\ref{table:treesupto3} and \ref{table:treesof4}.
As mentioned above, we define a product of two $\el\gl\rl$-codes based on grafting representative trees; also see Blower and Newsham~\cite[Def.~4.3]{BN}.
\begin{definition}[$\el\gl\rl$-grafting product]\label{def:lgrgrafting}
  For any two $\el\gl\rl$-codes, $\df_1$ and $\df_2$ of non-zero grade, we define their $\el\gl\rl$\emph{-grafting product} `$\star$' by,
  \begin{equation*}
    \df_1\star\df_2\coloneqq\df_1(\el+\gl+\rl)\df_2.
  \end{equation*}
  If $\mathfrak 0$ is the zero grade $\el\gl\rl$-code corresponding to the tree of grade zero given by,
  `${\tiny \begin{forest} for tree={grow'=90, parent anchor=center,child anchor=center} [] \path[fill=black] (.anchor) circle[radius=1.5pt]; \end{forest}}$',
  then for any $\el\gl\rl$-code, $\df$, we define,
  \begin{equation*}
    \df\star\mathfrak{0}\coloneqq\df(\rl+\gl),\quad\mathfrak{0}\star\df\coloneqq(\el+\gl)\df\quad\text{and}\quad\mathfrak{0}\star\mathfrak{0}\coloneqq\gl.
  \end{equation*}
\end{definition}
As highlighted in our development of the $\el\gl\rl$-encoding, the presence of a factor `$(\el+\gl+\rl)$' indicates the grafting of two subtrees together,
or equivalently the product or concantenation of two parenthesised strings together.
As explained in Definition~\ref{def:graftingontrees} and Remark~\ref{rmk:graftingandconcat}, the natural product on the real algebra of planar binary rooted trees is the grafting product,
which corresponds to the concantenation of the corresponding parenthesised strings.
Since the definition of the $\el\gl\rl$-algebra simply involves the algebra of planar binary rooted trees with grafting as product in which we equivalence specific subsets of trees, 
the naturally induced product is that of grafting for the equivalence sets, or equivalently $\el\gl\rl$-grafting of the $\el\gl\rl$-codes.
We thus formally define the `left-glue-right' algebra as follows.
\begin{definition}[`Left-glue-right' algebra]\label{def:lgr-alg}
  We define the `left-glue-right' or $\el\gl\rl$-algebra, to be the real nonassociative noncommutative algebra generated by the real algebra of planar binary rooted trees,
  in which we equivalence trees with the same $\el\gl\rl$-code. It is the real polynomial algebra with monomials involving the factors $\gl$, $(\el+\gl)$, $(\rl+\gl)$ and $(\el+\gl+\rl)$
  corresponding to the $\el\gl\rl$-encoding of the trees, equipped with the $\el\gl\rl$-grafting product.
\end{definition}
\begin{remark}
  Note that in Definition~\ref{def:lgrgrafting} for the $\el\gl\rl$-grafting product, the product $\df_1\star\df_2$ for descents of non-zero grade,
  precisely emulates the pre-P\"oppe product given in Lemma~\ref{lemma:Poppeprodmonomials}.
  The other three products therein, namely $\df\star\mathfrak{0}$, $\mathfrak{0}\star\df$ and $\mathfrak{0}\star\mathfrak{0}$, precisely emulate the augmented pre-P\"oppe products
  given in Definition~\ref{def:augmentedpPp}.
\end{remark}
To help illustrate the $\el\gl\rl$-grafting product, consider the following example.
\begin{example}\label{ex:lgrstarprod}
  Consider the trees of grade $3$ and less given in Table~\ref{table:treesupto3}. Suppose we graft the trees with word-codings $\tau_1=12$ and $\tau_2=212$ together, where graphically, 
  \begin{equation*}
    \tau_1={\tiny \begin{forest} for tree={grow'=90, parent anchor=center,child anchor=center, l=0cm,inner ysep=0cm,edge+=thick} [[][[][]]] \end{forest}} 
    \qquad\text{and}\qquad
    \tau_2= {\tiny \begin{forest} for tree={grow'=90, parent anchor=center,child anchor=center, l=0cm,inner ysep=0cm,edge+=thick} [[[][]][[][]]] \end{forest}}.
  \end{equation*}
  The grafting product of these two trees generates,
  \begin{equation*}
    {\scriptsize \begin{forest} for tree={grow'=90, parent anchor=center,child anchor=south, l=0cm,inner ysep=2pt,edge+=thick, s sep=0mm} [[$\tau_1$][$\tau_2$]] \end{forest}}
    ={\tiny \begin{forest} for tree={grow'=90, parent anchor=center,child anchor=center, l=0cm,inner ysep=0cm,edge+=thick} [[[][[][]]][[[][]][[][]]]] \end{forest}}.
  \end{equation*}
  The $\el\gl\rl$-codes corresponding to the trees $\tau_1=12$ and $\tau_2=212$ are $\df_1=(\el+\gl)\gl$ and $\df_2=\gl(\el+\gl+\rl)\gl$, respectively.
  The $\el\gl\rl$-grafting product of these two $\el\gl\rl$-codes is,
  \begin{equation*}
   (\el+\gl)\gl\star\gl(\el+\gl+\rl)\gl=(\el+\gl)\gl(\el+\gl+\rl)\gl(\el+\gl+\rl)\gl.
  \end{equation*}
  That the resulting $\el\gl\rl$-code on the right just above corresponds to the tree on the right above that results from the grafting of $\tau_1$ and $\tau_2$, is straightforward to check.
\end{example}
\begin{example}[Nonassociativity]\label{ex:graftingnonassoclgr}
  Let us consider the grafting products analogous to those in Example~\ref{ex:graftingnonassoctrees}.
  Suppose that $\df_1=\gl$ and $\df_0=\mathfrak{0}$. Then we have, $\df_0\star\df_0=\gl$ and $\df_1\star\df_0=\gl(\rl+\gl)$.
  However, consider the triple products, $\df_1\star(\df_0\star\df_0)=\gl\star\gl=\gl(\el+\gl+\rl)\gl$, and also, $(\df_1\star\df_0)\star\df_0=\gl(\rl+\gl)\star\mathfrak{0}=\gl(\rl+\gl)^2$.
  This exactly fits with the results of Example~\ref{ex:graftingnonassoctrees}, and reflects the nonassociativity of the grafting product in this context.
\end{example}

Lastly, we observe that we grade $\el\gl\rl$-codes according to the number of factors they contain, which corresponds to the number of vertices in any of its planar binary rooted tree representations. 
A quick glance through Table~\ref{table:treesupto3} reveals a single code of grade $1$, namely $\gl$. There are two codes of grade $2$, namely $(\el+\gl)\gl$ and $\gl(\rl+\gl)$.
There are only $4$ codes of grade $4$, though there are $5$ trees. This is because the trees $132$ and $231$ generate the same $\el\gl\rl$-code, namely $(\el+\gl)\gl(\rl+\gl)$.
Similarly, from Table~\ref{table:treesof4} we see that there are $8$ independent $\el\gl\rl$-codes.
Note that the number of independent $\el\gl\rl$-codes at each grade $n$ appears to be equal to $2^{n-1}$. 
The final columns in Tables~\ref{table:treesupto3} and \ref{table:treesof4} give a hint as to why this is the case, and we explore this in Section~\ref{sec:descent-alg}.
In particular, see Remarks~\ref{rmk:generatingdescents} and \ref{rmk:descentsand compositions}.

\section{Descent algebra}\label{sec:descent-alg}
Herein we develop the descent algebra that underlies the noncommutative KP hierarchy.
This algebra has two obvious representations, the \emph{standard representation}, as well as a \emph{binary representation} which proves more useful for exploring and explaining some of the results herein.  
The final three columns in Tables~\ref{table:treesupto3} and \ref{table:treesof4} respectively record, for each planar binary rooted tree shown,
their word-coding, their corresponding $\el\gl\rl$-code and the descent set of the word-coding.
The descent set records the positions in the wording coding where the next digit is strictly less than the previous digit.
\begin{definition}[Descent set]\label{def:descentset}
Given a word-coding $a_1a_2\cdots a_n$ corresponding to a planar binary rooted tree, the descent set $\df$ lists the indices $k\in\{1,\ldots,n-1\}$ such that $a_{k+1}<a_k$.  
\end{definition}
\begin{example}
  Consider Table~\ref{table:treesof4}. The tree with the word-coding, $1234$, does not contain any descents, so its descent set is $\emptyset_4$, i.e.\/ the empty set.
  We use the subindex $4$ to indicate that the word-coding from which the descent was extracted, is of grade $4$.
  The tree with word-coding, $1243$, has a descent between the third and fourth integers, so its descent set is $3_4$.
  Now consider the word-coding, $1432$, which contains two descents, one between the second and third integers, and another between the third and fourth integers.
  Hence the descent set corresponding to this word-coding is $23_4$. And so forth. See Example~\ref{ex:wordordering} for some further examples.
\end{example}
With the concept of descent sets in hand, we observe that in both Table~\ref{table:treesupto3} and Table~\ref{table:treesof4}, there is a one-to-one correspondence 
between the $\el\gl\rl$-codes corresponding to a given tree/word-code, and the descent set corresponding to that word-code.  
This equivalence is straightforward to establish. All $\el\gl\rl$-codes can be generated from the the zero grade $\el\gl\rl$-code, $\mathfrak 0$. See Definition~\ref{def:lgrgrafting}.
Any factors $(\rl+\gl)$ and $(\el+\gl+\rl)$ in a given $\el\gl\rl$-code denote a descent in the corresponding tree/word-code.
This is because the latter records a grafting between any two non-zero grade trees which naturally corresponds to a ``descent'' into the grafting root, wherever this occurs in the tree,
while the former records a right-grafting with a zero-grade tree, again invoking a descent.
We observe that in Table~\ref{table:treesof4}, the positions of these factors in the $\el\gl\rl$-codes present, occur at precisely the corresponding descent locations indicated in the final column. 
Thus indeed, the aforementioned one-to-one correspondence generalises to all $\el\gl\rl$-codes. Hereafter we abbreviate `descent sets' to `descent'.
\begin{definition}[Grade and length]~\label{def:gradeofdescent}
  For a descent $\df=(d_1\cdots d_k)_n$, we use $\|\df\|$ and $|\df|$ to respectively denote its \emph{grade} and \emph{length}. 
  The former is the length of the word-coding from which the descent was extracted, while the latter is the number of descents recorded therein.
  In other words we respectively set,
  \begin{equation*}
    \|\df\|\coloneqq n\qquad\text{and}\qquad |\df|\coloneqq k.
  \end{equation*}
\end{definition}
\begin{remark}[Generating descents]\label{rmk:generatingdescents}
Generating descents grade by grade is straightforward. There is a single descent of grade $1$, namely $\emptyset_1$.
There are two descents of grade $2$, namely $\{\emptyset_2,1_2\}$.
We can generate the set of grade $3$ descents by keeping this set, but extending the length of the descents by one,
and then also attaching a possible descent `$2$' at position two, to the end of these descents as follows,
\begin{equation*}
  \{\emptyset_2,1_2\}\to \{\emptyset_3,1_3,2_3,12_3\}.
\end{equation*}
Note that when we attach the `$2$' to the empty descent `$\emptyset_2$', we simply replace the empty descent by the `$2$'.
Then to generate all the descents of grade $4$, we take the set of descents of grade $3$ shown on the right just above,
keep that set and extend the length by one, and then also attach a possible descent `$3$' at position three:
\begin{equation*}
  \{\emptyset_3,1_3,2_3,12_3\}\to\{\emptyset_4,1_4,2_4,12_4,3_4,13_4,23_4,123_4\}.
\end{equation*}
The same process can be use to generate all the descents of grade $5$ giving,
\begin{align*}
  \{&\emptyset_5,1_5,2_5,12_5,3_5,13_5,23_5,123_5,\\
    &4_5,14_5,24_5,124_5,34_5,134_5,234_5,1234_5\}.
\end{align*}
And so forth. The number of descents at grade $n$ is thus obviously $2^{n-1}$. 
\end{remark}
\begin{definition}[Grafting product]\label{def:descentgrafting}
  For any two descents, $(a_1\cdots a_k)_n$ and $(b_1\cdots b_\ell)_m$, we define the (descent) \emph{grafting product} `$\star$' by,
  \begin{align*}
    (a_1\cdots& a_k)_n\star(b_1\cdots b_\ell)_m\coloneqq\\
              &\bigl(a_1\cdots a_kn(b_1+n+1)\cdots(b_\ell+n+1)\bigr)_{n+m+1}.
  \end{align*}
\end{definition}
\begin{remark}
  That the right-hand side is the correct resulting descent is explained as follows. The left factor descent retains its form in the grafting procedure.
  Then a root-graft occurs precisely at that point---this would correspond to an insertion of a factor $(\el+\gl+\rl)$ for $\el\gl\rl$-codes.
  This is recorded by the integer `$n$', where $n$ is the grade of the left factor descent.
  This `$n$' thus naturally records the position of the root grafting. We observe that the rest of the descent on the right corresponds to the
  right factor descent, but with all the descent positions shifted by `$n+1$'.
  This is because we must account for the insertion of both the left factor descent and the additional insertion of the single descent `$n$' in the middle due to the root grafting just described.
  This product is much more easily interpreted in the binary representation for descents as we see presently.
  For the moment, for a useful illustrative example, the reader may refer to the product in Example~\ref{ex:lgrstarprod}, though we will present many further illustrative examples presently.
  That the formula above still holds when the left descent factor is the empty descent, $\emptyset_n$, can be assumed if we take the convention that nothing is present on the right-hand side before the `$n$' in this instance. 
  The same is true if the right descent factor is the empty descent, $\emptyset_m$, if again, we take the convention that in this instance, there is nothing present after the `$n$' on the right-hand side.
  As mentioned we thoroughly explore these cases presently. 
  Lastly, the observant reader will have noticed that we use the same symbol `$\star$' for both the $\el\gl\rl$-grafting product and the descent grafting product.
  Presently, below, we confirm that the descent algebra equipped with the grafting product is isomorphic to the $\el\gl\rl$-algebra equipped with $\el\gl\rl$-grafting product.
  Therefore we do not distinguish the two.
\end{remark}
Let us consider some illustrative examples.
\begin{example}\label{ex:descentproducts}
  Some simple descent-grafting products are as follows,
  \begin{align*}
    \emptyset_0\star\emptyset_0&=\emptyset_1, &\!\!\! \emptyset_0\star\emptyset_1&=\emptyset_2,  &\!\!\!  \emptyset_1\star\emptyset_0&=1_2,         &\!\!\!  \emptyset_0\star\emptyset_2&=\emptyset_3,\\
    \emptyset_1\star\emptyset_1&=1_3,         &\!\!\!  \emptyset_2\star\emptyset_0&=2_3,         &\!\!\!  \emptyset_0\star\emptyset_3&=\emptyset_4, &\!\!\!  \emptyset_1\star\emptyset_2&=1_4,\\
    \emptyset_0\star1_2&=2_3,                 &\!\!\!  1_2\star\emptyset_0&=12_3,                &\!\!\!  1_2\star\emptyset_1&=12_4,                &\!\!\!  1_2\star1_2&=124_5.
  \end{align*}
  For another example, see Example~\ref{ex:lgrstarprod}, which translated into descents gives, $\emptyset_2\star1_3=24_6$.
  Some products generate the same result, for example $\emptyset_2\star\emptyset_0=2_3$ and $\emptyset_0\star1_2=2_3$.
\end{example}
\begin{example}[Nonassociativity]\label{ex:graftingnonassocdescents}
  Consider the grafting products analogous to those in Examples~\ref{ex:graftingnonassoctrees} and \ref{ex:graftingnonassoclgr}.
  We observe from Example~\ref{ex:descentproducts} that $\emptyset_0\star\emptyset_0=\emptyset_1$ and $\emptyset_1\star\emptyset_0=1_2$.
  However, consider the triple products. We observe that,
  $\emptyset_1\star(\emptyset_0\star\emptyset_0)=\emptyset_1\star\emptyset_1=1_3$, while on the other hand we have, $(\emptyset_1\star\emptyset_0)\star\emptyset_0=1_2\star\emptyset_0=12_3$.
  This fits with Examples~\ref{ex:graftingnonassoctrees} and \ref{ex:graftingnonassoclgr}, and reflects the nonassociativity of the descent grafting product.
\end{example}

We formally define the `descent algebra' we use herein as follows. Let $\mathbb D$ denote the set of all possible descents.
In other words, $\mathbb D$ includes all monomials of the form $(d_1\cdots d_k)_n$ over all $d_1,\ldots,d_k\in\mathbb N$ with $d_\ell<d_{\ell+1}$ for $\ell\in\{1,\ldots,k-1\}$ and $d_k<n$, for all $k,n\in\mathbb N$.
It also includes all the empty descents of the form $\emptyset_n$, for all $n\in\mathbb N_0$.
\begin{definition}[Descent algebra]\label{def:descent-alg}
  We define the \emph{descent algebra} to be the real nonassociative noncommutative algebra, $\R\langle\mathbb D\rangle_\star$, generated by descents and equipped with the grafting product.
\end{definition}

There is another representation of the descent algebra, that proves very insightful and useful. This is the so-called binary representation.
For this representation, we simply write down all descents as binary codes.
\begin{definition}[Binary representation]\label{def:binaryrepresentation}
For the binary representation of the descent $\df=(d_1\cdots d_k)_n\in\mathbb D$, we take a string of $n$ zeros, and then at each location $d_\ell$ along the string, for $\ell\in\{1,\ldots,k\}$, we place a `$1$' at that position.
In other words the binary-code representation for $(d_1\cdots d_k)_n$ is given by,
\begin{equation*}
  \underbrace{0\cdots0}_{d_1-1}\stackon{\stackon{$1$}{$\downarrow$}}{{\scriptsize $d_1$}}\!\underbrace{0\cdots0}_{d_2-d_1-1}\stackon{\stackon{$1$}{$\downarrow$}}{{\scriptsize $d_2$}}
  \underbrace{0\cdots0}_{d_3-d_2-1}\!\stackon{\stackon{$1$}{$\downarrow$}}{{\scriptsize $d_3$}}\;0\cdots0\stackon{\stackon{$1$}{$\downarrow$}}{{\scriptsize $d_{k-1}$}}\!\!\!\!
  \underbrace{0\cdots0}_{d_k-d_{k-1}-1}\!\!\stackon{\stackon{$1$}{$\downarrow$}}{{\scriptsize $d_k$}}\underbrace{0\cdots0}_{n-d_k}\!.
\end{equation*}
In the above, we have indicated by arrows, the positions of the descents. The braces indicate the lengths of the continuous strings of zeros shown; those lengths could be zero of course.
\end{definition}
\begin{remark}
In the binary representation outlined above, there is always a superfluous `$0$' at the end as the number of descents is at most $n-1$.
\end{remark}
\begin{example}
  Some example binary representations are as follows,
  \begin{align*}
    \emptyset_1&=0,    & \emptyset_2&=00, & 1_2&=10,    && \\
    \emptyset_3&=000,  & 1_3&=100,        & 2_3&=010,   & 12_3&=110,  \\
    \emptyset_4&=0000, & 1_4&=1000,       & 2_4&=0100,  & 12_4&=1100,  \\
            3_4&=0010, & 13_3&=1010,      & 23_3&=0110, & 123_3&=1110,
  \end{align*}
  and so forth. In Remark~\ref{rmk:binaryproducts} just below, we explain how we deal with $\emptyset_0$ in this representation.
\end{example}
\begin{remark}[Descents and compositions]\label{rmk:descentsand compositions}
There is a one-to-one correspondence between descents and compositons. As we have seen, there is a natural one-to-one map between descents of a given grade and length and binary numbers.
In turn there is a natural one-to-one map between binary numbers and compositions, indeed between compositions of $n$ of length $k$ and binary numbers involving $n$ digits and $k$ ones.
This is standard map. For example in the case of binary numbers of length $n$ with $k$ ones, imagine $n+1$ intervals of unit length.
Given such a binary number, suppose we identify at the positions of the ones, the corresponding $k$ out of $n$ internal interval endpoints.
Then suppose we glue together the coincident intervals at those endpoints to generate longer intervals of integer length.
The digit sequence of the count of the lengths of the resulting intervals represents a composition of $n$ of length $k$.
That this gives the one-to-one mapping between any such binary number and compositions of $n$ of length $k$, is straightforward, and by transitivity,
there is a one-to-one map between descents and compositions.
\end{remark}
\begin{remark}[Grafting binary descents]\label{rmk:binaryproducts}
The grafting product of two descents given in binary representation is particularly simple. This explains why this representation is so useful.
Consider the following two examples, $1_3\star2_3=136_7$ and $12_4\star124_5=124679_{10}$.
In the binary representation, these examples are given by,
\begin{equation*}
  100\star010=1010010
  \quad\text{and}\quad
  1100\star11010=1101011010.
\end{equation*}
The effect of the grafting product in this representation is straightforwardly interpreted.
We essentally concatenation the two binary strings in the two factors, and perform a `fix' at the concatenation interface.
For example, in the string represented by the product `$100\star010$' we replace the central `$0\star$' by `$10$'. 
With this in hand, let us consider some simpler examples, in particular those outlined in Example~\ref{ex:descentproducts}.
In binary representation we express these products in respective order as follows, 
\begin{align*}
  \emptyset_0\star\emptyset_0 &=0,  &\!\! \emptyset_0\star  0&=00,   &\!\! 0\star\emptyset_0&=10,  \\
  \emptyset_0\star00&=000,          &\!\!             0\star0&=100,  &\!\! 00\star\emptyset_0&=010, \\
  \emptyset_0\star000&=0000,        &\!\!            0\star00&=1000, &\!\! \emptyset_0\star10&=010,  \\
  10\star\emptyset_0&=110,          &\!\!            10\star0&=1100, &\!\! 10\star10&=11010.
\end{align*}
Note that the first two cases above respectively correspond to $\emptyset_0\star\emptyset_0=\emptyset_1$ and $\emptyset_0\star\emptyset_1=\emptyset_2$.
Observe that we use $\emptyset_0$ in the binary representation as well. 
In the examples above we notice two features.
One is that the concatenation of the binary representations as result of the grafting product, with the `fix' mentioned above, is naturally observed therein.
The other is that left multiplication by $\emptyset_0$, concatenates a zero to the front of the right factor in the grafting product. This is the second component of the aforementioned `fix'.
\end{remark}
\begin{remark}[The concatenation `fix']\label{rmk:twofixes}
  The grafting product for binary representations, highlighted in Remark~\ref{rmk:binaryproducts}, is thus interpreted as the concatenation of the binary strings of the factors involved,
  with the fix mentioned therein. In summary, if $\mathfrak{b}=\mathfrak{b}^\prime0$ and $\mathfrak{c}$ are binary representations of two descents, then this \emph{fix} is given by, 
  \begin{equation}\label{eq:fix}
    \mathfrak{b}^\prime0\star\mathfrak{c}\to\mathfrak{b}^\prime10\,\mathfrak{c}\qquad\text{and}\qquad \emptyset_0\star\,\mathfrak{c}\to0\,\mathfrak{c},  
  \end{equation}
  where straightforward concatenation of the binary strings is assumed on the right.
\end{remark}

We define the following weight character associated with any descent, based on that for trees in Definition~\ref{def:weightchar}.
\begin{definition}[Descent weight]~\label{def:descentweight}
  The weight $\omega\colon\mathbb D\to\R$ of any descent is defined as follows. By convention we set $\omega\circ\emptyset_0\coloneqq1$ and
  start with $\omega\circ\emptyset_1\coloneqq C^0_0$. Then for any descent $\df\in\mathbb D$ of length $k$, we define,
  \begin{equation*}
     \omega\circ\df\coloneqq\omega\circ(\emptyset_0\star\df_1)+\omega\circ(\df_2\star\df_3)+\cdots+\omega\circ(\df_{k}\star\df_{k+1}),
  \end{equation*} 
  where $\emptyset_0\star\df_1$, \ldots, $\df_{k}\star\df_{k+1}$ are all the possible descent pairs, which when grafted together, generate $\df$.
  For any such pair, we can compute the weight of, say $\df_1\star\df_2$, via,
  \begin{equation}\label{eq:pairweight}
    \omega\circ(\df_1\star\df_2)\coloneqq C^{\|\df_1\|+\|\df_2\|}_{\|\df_1\|}(\omega\circ\df_1)\,(\omega\circ\df_2).
  \end{equation}
\end{definition}
\begin{remark}
  We emphasise that: (1) We have naturally extended the weight functional linearly;
  (2) The decomposition of the descent $\df$ of length $k$, into the sum of pairs of descents that when grafted together generate $\df$, is \emph{degrafting}.
  See Section~\ref{sec:degrafting} where we consider this operation in detail. Many basic examples of computing weights of descents can also be found therein; 
  (3) That the weight of the grafting product of two descents, $\df_1\star\df_2$, is given by \eqref{eq:pairweight}, is based on the corresponding formula
  for the grafting of two trees, as outlined in Definition~\ref{def:weightchar};
  (4) We can also think of the weight of a descent, as the sum of all the weight characters associated with the corresponding individual planar binary rooted trees, with that descent.
\end{remark}
Recall the weighted sum of planar binary rooted trees, $\pf_n$, in Definition~\ref{def:weightedpolytrees}. The equivalent sum in terms of descents follows straightforwardly.
\begin{lemma}[Weighted sum of descents]\label{lemma:weightedsumofdescents}
  The sum of planar binary trees, $\pf_n$, in the context of descents, is equivalent to,
   \begin{equation*}
        \pf_n=\sum_{\|\df\|=n}\omega(\df)\cdot\df.
  \end{equation*}
\end{lemma}
Further, the result of Lemma~\ref{lemma:derivativesandtrees} applies for $\pf_n$ in the case here, where the weighted sum is over descents. In other words we have, $\pa_x\pf_n\equiv\pf_{n+1}$.
\begin{example}[Weighted sums]\label{ex:weightedsums}
Some examples of low order weighted sums, $\pf_n$, are as follows. 
We have that $\pf_0=\emptyset_0$ and $\pf_1=\emptyset_1$.
This means that the weight characters of the two descents concerned here, are $\omega(\emptyset_0)=1$ and $\omega(\emptyset_1)=1$.
Further, we have, $\pf_2=\emptyset_2+1_2$ and $\pf_3=\emptyset_3+2\cdot 1_3+3\cdot 2_3+12_3$.
In these examples, we observe that: $\omega(\emptyset_2)=1$; $\omega(1_2)=1$; $\omega(\emptyset_3)=1$; $\omega(1_3)=2$; $\omega(2_3)=3$ and $\omega(12_3)=1$.
\end{example}

We now examine the role of the weighted sum of descents in Lemma~\ref{lemma:weightedsumofdescents}, namely $\pf_n$, in the formulation of the noncommutative KP hierarchy equations.
\begin{definition}[Sato character]\label{def:Satochar}
Given $n,\ell\in\mathbb N$ with $\ell<n$, and a collection of integers $c_1,c_2,\cdots,c_\ell\in\mathbb N$, we define the \emph{Sato character} $\gamma_n$ of the string $c_1c_2\cdots c_\ell$ by, 
\begin{equation*}
     \gamma_n(c_1c_2\cdots c_\ell)\coloneqq C^n_{c_1-1}C^{n-c_1}_{c_2-1}C^{n-c_1-c_2}_{c_3-1}\cdots C^{n-c_1-\ldots-c_{\ell-1}}_{c_\ell-1}.
\end{equation*}
\end{definition}
\begin{prop}[KP hierarchy: descents]\label{pres:KP}
In terms of descents, the equations for the noncommutative KP hierarchy given in Corollary~\ref{cor:hierarchyequations} have the form,
\begin{align*}
  b_{n-2}=&\;-C_1^n\pf_1,\\
  b_{n-3}=&\;-C_2^n\pf_2-C_0^{n-2}b_{n-2}\star\pf_0,\\
  b_{n-4}=&\;-C_3^n\pf_3-C_1^{n-2}b_{n-2}\star\pf_1-C_0^{n-3}b_{n-3}\star\pf_0,\\
  b_{n-5}=&\;-C_4^n\pf_4-C_2^{n-2}b_{n-2}\star\pf_2-C_1^{n-3}b_{n-3}\star\pf_1\\
         &\;-C_0^{n-4}b_{n-4}\star\pf_0,\\
  \vdots & \\
  b_{n-k}=&\;-C^n_{k-1}\pf_{k-1}-\sum_{\ell=2}^{k-1}C^{n-\ell}_{k-1-\ell}b_{n-\ell}\star\pf_{k-1-\ell}.
\end{align*}
  The `$\star$' product represents the augmented pre-P\"oppe product. 
\end{prop}
If, in these equations, we set $\qf_n\coloneqq\pf_{n-1}$, and then successively substitute each row into the next, we find that,
\begin{align*}
  b_{n-2}=&\;-\gamma_n(2)\cdot\qf_2,\\
  b_{n-3}=&\;-\gamma_n(3)\cdot\qf_3+\gamma_n(21)\cdot\qf_2\star\qf_1,\\
  b_{n-4}=&\;-\gamma_n(4)\cdot\qf_4+\gamma_n(31)\cdot\qf_3\star\qf_1+\gamma_n(22)\cdot\qf_2\star\qf_2\\
         &\;-\gamma_n(211)\cdot(\qf_2\star\qf_1)\star\qf_1,\\
  b_{n-5}=&\;-\gamma_n(5)\cdot\qf_5+\gamma_n(41)\cdot\qf_4\star\qf_1+\gamma_n(32)\cdot\qf_3\star\qf_2\\
         &\;+\gamma_n(23)\cdot\qf_2\star\qf_3-\gamma_n(311)\cdot(\qf_3\star\qf_1)\star\qf_1\\
         &\;-\gamma_n(221)\cdot(\qf_2\star\qf_2)\star\qf_1-\gamma_n(212)\cdot(\qf_2\star\qf_1)\star\qf_2\\
         &\;+\gamma_n(2111)\cdot((\qf_2\star\qf_1)\star\qf_1)\star\qf_1.
\end{align*}
For each coefficient, if we consider the pullback of the map, $c_1c_2\cdots c_\ell\mapsto \gamma_n(c_1c_2\cdots c_\ell)\cdot(\qf_{c_1}\star\cdots)\star\qf_{c_\ell}$,
we find, $b_{n-2}\to-2$, $ b_{n-3}\to-3+21$, $b_{n-4}\to-4+31+22-211$, and,
\begin{equation*}
b_{n-5}\to-5+41+32+23-311-221-212+2111. 
\end{equation*}
It thus appears that each $b_{n-k}$ is parametrised by the linear combinations of the compositions of $k$, except for those compositions starting with a `$1$',
and the sign of each term is determined by the length of the composition.
Indeed, let $\mathcal C_k^\ast$ denote the set of compositions of $k$ that do not start with a `$1$'. Then we have the following result.
\begin{theorem}[Sato coefficients]\label{thm:Satocoeffs}
  The $k^{\mathrm{th}}$ coefficient in the Sato formulation of the noncommutative hierarchy equations, $b_{n-k}$, is given by,
\begin{equation*}
   b_{n-k}=\sum_{c=c_1\cdots c_\ell\in\mathcal C_k^\ast}(-1)^{\ell} \gamma_n(c)\cdot(\qf_{c_1}\star\cdots)\star\qf_{c_\ell}.
\end{equation*}
The sum is over all compositions $c=c_1\cdots c_\ell$ in $\mathcal C_k^\ast$, and the product is computed left to right starting with $\qf_{c_1}\star\qf_{c_2}$.
\end{theorem}
\begin{proof}
The result follows by induction. Consider the expression for $b_{n-k}$ given in Proposition~\ref{pres:KP}, now with $\pf_\kappa$ replaced by $\qf_{\kappa+1}$. 
The result is,
\begin{align*}
  b_{n-k}=&\;-C^n_{k-1}\qf_{k}-C^{n-2}_{k-3}b_{n-2}\star\qf_{k-2}-C^{n-3}_{k-4}b_{n-3}\star\qf_{k-3}\\
         &\;-\cdots-C_0^{n-k+1}b_{n-k+1}\star\qf_1.
\end{align*}
Consider the collection of compositions in $\mathcal C^\ast_{k-1}$. Postattach a `$1$' to each such composition therein. 
Now consider the collection of compositions in $\mathcal C^\ast_{k-2}$. Postattach a `$2$' to each such composition. 
Continue this process until the compositions in $\mathcal C^\ast_{2}=\{2\}$ is postattached by `$k-2$'.  
It is straightforward to establish that the union of these sets together with $\{k\}$ itself is $\mathcal C^\ast_{k}$.
That the binomial coefficients mirror this construction and generate the $\gamma_n$-coefficients is also straightforward, and similarly for the signs of the terms.
This establishes the result.
\qed
\end{proof}
The explicit form for the noncommutative KP equation of order $n$ in the hierarchy thus follows immediately. Indeed we have the following. 
\begin{corollary}[KP equation]\label{cor:KPequationformula}
The $n^{\mathrm{th}}$ equation in the noncommutative KP hierarchy is given explicitly by,
\begin{equation*}
   b_{-1}=\sum_{c=c_1\cdots c_\ell\in\mathcal C_{n+1}^\ast}(-1)^{\ell} \gamma_n(c)\cdot(\qf_{c_1}\star\cdots)\star\qf_{c_\ell}.
\end{equation*}
The sum is over all compositions $c=c_1\cdots c_\ell$ in $\mathcal C_{n+1}^\ast$, and, again, the product is computed left to right. 
\end{corollary}
Corollary~\ref{cor:KPequationformula} explicitly outlines all the terms present in the $n^{\mathrm{th}}$ equation in the noncommutative KP hierarchy.

\section{Degrafting}\label{sec:degrafting}
Herein we develop the degrafting operator. 
As we see, the degrafting operator is very natural in the context of the noncommutative KP hierarchy and in our pursuit of the solution to the whole hierarchy.
Recall, our main objective is to solve the whole hierarchy by direct linearisation.

To begin, consider the noncommutative equations given in Proposition~\ref{pres:KP}.
Recall from Example~\ref{ex:weightedsums}, the expansion sums for $\pf_0$, $\pf_1$, $\pf_2$, etc, though here we retain the general forms for the weighted characters, $\omega$. 
If we substitute the explicit expressions for these weighted sums into the hierarchy equations in Proposition~\ref{pres:KP}, we find that,
\begin{align*}
  b_{n-2}=&\;-\gamma_n(\emptyset_1)\omega(\emptyset_1)\botimes\emptyset_1,\\
  b_{n-3}=&\;-\gamma_n(\emptyset_2)\omega(\emptyset_2)\botimes\emptyset_2-\gamma_n(1_2)\omega(1_2)\botimes1_2\\
         &\;+\gamma_n(\emptyset_1,\emptyset_0)\omega(\emptyset_1)\omega(\emptyset_0)\botimes\emptyset_1\star\emptyset_0,\\
  b_{n-4}=&\;-\gamma_n(\emptyset_3)\omega(\emptyset_3)\botimes\emptyset_3-\gamma_n(1_3)\omega(1_3)\botimes1_3\\
         &\;-\gamma_n(2_3)\omega(2_3)\botimes2_3-\gamma_n(12_3)\omega(12_3)\botimes12_3\\
         &\;+\gamma_n(\emptyset_1,\emptyset_1)\omega(\emptyset_1)\omega(\emptyset_1)\botimes\emptyset_1\star\emptyset_1\\
         &\;+\gamma_n(\emptyset_2,\emptyset_0)\omega(\emptyset_2)\omega(\emptyset_0)\botimes\emptyset_2\star\emptyset_0\\
         &\;+\gamma_n(1_2,\emptyset_0)\omega(1_2)\omega(\emptyset_0)\botimes1_2\star\emptyset_0\\
         &\;-\gamma_n(\emptyset_1,\emptyset_0,\emptyset_0)\omega(\emptyset_1)\omega(\emptyset_0)\omega(\emptyset_0)\!\botimes\!(\emptyset_1\star\emptyset_0)\star\emptyset_0,
\end{align*}
and in general,
\begin{equation*}
  b_{n-k}\!=\!\!\sum(-1)^\ell\gamma_n(\df_1,\ldots,\df_\ell)\omega(\df_1)\cdot\!\cdot\!\cdot\omega(\df_\ell)\botimes\df_1\star\cdot\!\cdot\!\cdot\star\df_\ell,
\end{equation*}
where the sum is over all descents $\df_1,\df_2,\ldots,\df_\ell$ such that,
\begin{equation}\label{eq:descentcomps}
\bigl(\|\df_1\|+1\bigr)\bigl(\|\df_2\|+1\bigr)\cdots\bigl(\|\df_\ell\|+1\bigr)\in\mathcal C_k^\ast.
\end{equation}
The factors on the left are to be interpreted as a string of positive integers and thus a composition element.
We assume the product `$\df_1\star\cdots\star\df_\ell$' is computed left to right.
Note that in the expansions above, we separated the real coefficents to the left of the tensor symbol `$\botimes$' from the descents and their grafting products on the right.
Further, we used the notation $\gamma_n=\gamma_n(\df_1,\ldots,\df_k)$ to denote the corresponding Sato character coefficient $\gamma_n$ that can be extracted from $\df_1,\ldots,\df_k$.
This is a slight abuse of notation. More precisely we should write $\gamma_n=\gamma_n\bigl(\|\df_1\|\cdots\|\df_k\|\bigr)$.

Let $\mathbb D^{\otimes}$ denote $\mathbb D^\otimes\coloneqq\cup_{k\geqslant0}\mathbb D^{\otimes k}$, the set of all tensored descents.
It is natural to define $\R\langle\mathbb D\rangle_{\otimes}$ as the real algebra of tensored descents with `$\otimes$' as product.
Hereafter, we also use the notation $\R\langle\mathbb D\rangle_{\cendot}$ to denote the same real algebra of tensored descents, but using the symbol `$\cendot$' as the tensor product instead.
This equivalent representation is convenient in what follows, as we use such tensor algebras in several different contexts.
\begin{definition}[Sato weight map]\label{def:Satoweightmap}
We define the \emph{Sato weight map} $\sigma_n\colon\mathbb D^{\otimes}\to\R$ to be,
\begin{equation}\label{eq:sigmadef}
   \sigma_n\colon\df_1\otimes\cdots\otimes\df_k\mapsto\gamma_n(\df_1,\ldots,\df_k)\omega(\df_1)\cdots\omega(\df_k).
\end{equation}
\end{definition}
Suppose, for each Sato coefficient $b_{n-k}$ above, for the each real coefficient to the left of the tensor symbol `$\botimes$', we consider their pullback to $\R\langle\mathbb D\rangle_{\cendot}$ by $\sigma_n$.
This gives,
\begin{align*}
  b^\ast_{n-2}=&\;-\emptyset_1\botimes\emptyset_1,\\
  b^\ast_{n-3}=&\;-\emptyset_2\botimes\emptyset_2-1_2\botimes1_2+(\emptyset_1\cendot\emptyset_0)\botimes\emptyset_1\star\emptyset_0,\\
  b^\ast_{n-4}=&\;-\emptyset_3\botimes\emptyset_3-1_3\botimes1_3-2_3\botimes2_3-12_3\botimes12_3\\
         &\;+(\emptyset_1\cendot\emptyset_1)\botimes\emptyset_1\star\emptyset_1+(\emptyset_2\cendot\emptyset_0)\botimes\emptyset_2\star\emptyset_0\\
         &\;+(1_2\cendot\emptyset_0)\botimes1_2\star\emptyset_0\\
         &\;-(\emptyset_1\cendot\emptyset_0\cendot\emptyset_0)\botimes(\emptyset_1\star\emptyset_0)\star\emptyset_0,
\end{align*}
and in general,
\begin{equation}\label{eq:bnkformula}
  b_{n-k}^\ast\!=\!\!\sum(-1)^\ell(\df_1\cendot\df_2\cendot\cdots\cendot\df_\ell)\botimes\df_1\star\df_2\star\cdots\star\df_\ell,
\end{equation}
where the sum is over all descents $\df_1,\df_2,\ldots,\df_\ell\in\mathbb D$ satisfying \eqref{eq:descentcomps}.
Again, we assume the product `$\df_1\star\df_2\star\cdots\star\df_\ell$' is computed left to right.
We now recall some of the grafting products outlined in Example~\ref{ex:descentproducts}. 
We observe that all the grafting products shown on the right for $b_{n-2}^\ast$, $b_{n-3}^\ast$ and $b_{n-4}^\ast$ above, are present in Example~\ref{ex:descentproducts}.
For the last term in $b_{n-4}^\ast$, we observe,
\begin{equation*}
(\emptyset_1\star\emptyset_0)\star\emptyset_0=1_2\star\emptyset_0=12_3.
\end{equation*}
As highlighted in Example~\ref{ex:descentproducts}, some products generate the same descent.
Computing all the products in the formulae just above, and collecting like terms together, we have,
\begin{subequations}\label{eq:bnkcoeffs}
\begin{align}
  b^\ast_{n-2}=&\;-\emptyset_1\botimes\emptyset_1,\label{eq:bn2dg}\\
  b^\ast_{n-3}=&\;-\emptyset_2\botimes\emptyset_2-(1_2-\emptyset_1\cendot\emptyset_0)\botimes1_2,\label{eq:bn3dg}\\
  b^\ast_{n-4}=&\;-\emptyset_3\botimes\emptyset_3-(1_3-\emptyset_1\cendot\emptyset_1)\botimes1_3\nonumber\\
         &\;-(2_3-\emptyset_2\cendot\emptyset_0)\botimes2_3\nonumber\\
         &\;-(12_3-1_2\cendot\emptyset_0+\emptyset_1\cendot\emptyset_0\cendot\emptyset_0)\botimes12_3.\label{eq:bn4dg}
\end{align}
\end{subequations}
We make four observations. First, the coefficients $b^\ast_{n-k}$ lie in the tensored algebra,
\begin{equation*}
      \R\langle\mathbb D\rangle_{\cendot}\botimes\R\langle\mathbb D\rangle_{\star}.
\end{equation*}
Here, $\R\langle\mathbb D\rangle_{\star}$ represents the descent algebra equipped with the grafting product in Definition~\ref{def:descent-alg},
but for which, all the `$\star$' products are interpreted, in terms of bracketing, left to right.
Hereafter we write `$\star$' for this product and assume,
\begin{quote}
\emph{all such products are computed left to right,} 
\end{quote}
in this manner.
The left factor, $\R\langle\mathbb D\rangle_{\cendot}$, is the real tensor algebra of descents as outlined above.
Second, since we will always be interpreting the elements on the left in $\R\langle\mathbb D\rangle_{\cendot}$
through the prism of the map $\sigma_n\colon\R\langle\mathbb D\rangle_{\cendot}\to\R$ in \eqref{eq:sigmadef}, there is no need to associate or parenthesise them.
Third, the product of any two terms in $\R\langle\mathbb D\rangle_{\cendot}\botimes\R\langle\mathbb D\rangle_{\star}$ is interpreted in the obvious fashion as follows, 
\begin{equation*}
      \bigl(\df\botimes\df\bigr)\bigl(\df^\prime\botimes\df^\prime\bigr)
      =(\df\cendot\df^\prime)\botimes(\df\star\df^\prime),
\end{equation*}
which is extended linearly.
Fourth, any linear combination in $\R\langle\mathbb D\rangle_{\cendot}\botimes\R\langle\mathbb D\rangle_{\star}$ which involves `$\star$' products on the right,
can be re-written as a linear combination that does not involve any `$\star$' products on the right.
For example, for a given $n\in\mathbb N$, suppose we are given a linear combination in $\R\langle\mathbb D\rangle_{\cendot}\botimes\R\langle\mathbb D\rangle_{\star}$ of the form,
\begin{equation}\label{eq:lincombdescents}
    \sum_{\df_1,\ldots,\df_\ell} \!\! c(\df_1,\ldots,\df_\ell)\bigl(\df_1\cendot\cdots\cendot\df_\ell\botimes\df_1\star\cdots\star\df_\ell\bigr),   
\end{equation}
where the sum is over all possible descents $\df_1,\ldots,\df_\ell$ such that $\|\df_1\|+\cdots+\|\df_\ell\|+\ell-1=n$,
and the $c(\df_1,\ldots,\df_\ell)$ are the corresponding real coefficients.
All the products of the form `$\df_1\star\cdots\star\df_\ell$' generate other descents, some of which might be the same, as highlighted in Example~\ref{ex:descentproducts}.
This means we can always rewrite \eqref{eq:lincombdescents} in the form,
\begin{equation}\label{eq:adjoint}
    \sum_{\df\in\mathbb D} F^{\cendot}(\df)\botimes\df,
\end{equation}
where $F^{\cendot}(\df)$ is a linear combination of all the possible tensored descents $\df_1^\prime\cendot\cdots\cendot\df^\prime_\ell$, for any $\ell\in\mathbb N$, such that,
\begin{equation*}
    \df_1^\prime\star\cdots\star\df_\ell^\prime=\df.
\end{equation*}
This is precisely what we did in the cases of $b^\ast_{n-2}$, $b^\ast_{n-3}$ and $b^\ast_{n-4}$ in \eqref{eq:bnkcoeffs}. We now formalise this observation.
\begin{definition}[Degrafting]\label{def:degraftingop}
We define the \emph{degrafting operator} $\Delta\colon\mathbb D\to\mathbb D\otimes\mathbb D$ by,
\begin{equation*}
  \Delta(\df)\coloneqq \sum_{\df_1,\df_2\in\mathbb D}\langle \df_1\star \df_2, \df\rangle\,\df_1\otimes\df_2,
\end{equation*}
where $\langle\,\cdot\,,\,\cdot\,\rangle$ is the inner product on $\R\langle\mathbb D\rangle_\star$ which is one when both arguments match and zero otherwise.
We set $\Delta(\emptyset_0)\coloneqq0$.
\end{definition}
The effect of $\Delta$ on a descent $\df$ is to enumerate and sum all the possible descents $\df_1$ and $\df_2$ such that $\df_1\star\df_2=\df$.
\begin{remark}
The observant reader will recognise that our definition for $\Delta$ is analogous to a standard construct: a coproduct. See Reutenauer~\cite{Reutenauer} and Section~\ref{sec:discussion}. 
\end{remark}
\begin{example}\label{ex:degrafting}
  With reference to Example~\ref{ex:descentproducts}, we observe,
  \begin{equation*}
    \Delta(\emptyset_1)=\emptyset_0\otimes\emptyset_0, \quad \Delta(\emptyset_2)=\emptyset_0\otimes\emptyset_1, \quad \Delta(1_2)=\emptyset_1\otimes\emptyset_0,
  \end{equation*}
  and so forth, but that, $\Delta(2_3)=\emptyset_2\otimes\emptyset_0+\emptyset_0\otimes1_2$.
\end{example}
\begin{example}[Binary degrafting]\label{ex:binarydegrafting}
  The binary representation for descents is particularly useful in degrafting. In the last example we considered the case $2_3=010$.
  Recall the concantenation `fix' in Remark~\ref{rmk:twofixes}. In particular, consider the inverse of the fixes in \eqref{eq:fix}.
  That $2_3=010$ starts with a `$0$' suggests $010=\emptyset_0\star10=\emptyset_0\star1_2$, where we take $\mathfrak{c}=10$.
  Then that $2_3=010$ contains the subset `$10$' suggests that $010=00\star\emptyset_0$, where we take the `$\mathfrak{c}$' in the lefthand relation in \eqref{eq:fix} to be $\emptyset_0$.
  This shows that $2_3=\emptyset_2\star\emptyset_0$. This exhausts the only two possibilities and generates the final degrafting result in Example~\ref{ex:degrafting}.
\end{example}
Recall the binary representation for any descent $\df=(d_1\cdots d_k)_n$ given in Definition~\ref{def:binaryrepresentation}.
With the results of Example~\ref{ex:binarydegrafting} in mind, as well as the concatenation `fix' in Remark~\ref{rmk:twofixes},
it is straightforward to degraft any descent $\df=(d_1\cdots d_k)_n$ as follows. We work from left to right.
Either $d_1-1>0$, then $\df$ starts with a zero, and we have,
\begin{equation*}
     \df=\emptyset_0\star\df_1,
\end{equation*}
where $\df_1$ is the same as $\df$, but with the first `$0$' missing, or, $d_1-1=0$, i.e.\/ $d_1$ starts with a `$1$', and this possibility does not occur.
Next, if $d_2-d_1-1>0$, then we know,
\begin{equation*}
     \df=\df_2\star\df_3,
\end{equation*}
where $\df_2$ is the same as $\df$ all to the left of the `$1$' in the $d_2$ position, with the `$1$' in that position replaced by a `$0$'.  
The factor $\df_3$ is the same as $\df$ to the right of the `$1$' in the $d_2$ position,  with the first `$0$' missing.
If $d_2-d_1-1=0$, then this does not occur. And so on, until we reach the `$1$' in the $d_k$ position. 
Assuming $d_k-d_{k-1}-1>0$, then we know,
\begin{equation*}
     \df=\df_k\star\df_{k+1},
\end{equation*}
with $\df_k$ and $\df_{k+1}$ constructed in exactly the same manner as just outlined, except that if $n-d_k=1$, then $\df_{k+1}$
is interpreted as $\emptyset_0$. Combining these results, this means that if $\df=(d_1\cdots d_k)_n$, then,
\begin{equation*}
\Delta(\df)=\emptyset_0\otimes\df_1+\df_2\otimes\df_3+\cdots+\df_k\otimes\df_{k+1},
\end{equation*}
under the proviso that for any $\ell\in\{2,\ldots,k\}$, the $\df_\ell\otimes\df_{\ell+1}$ term may not be present if $d_\ell-d_{\ell-1}-1=0$,
that the first term may not be present if $d_1-1=0$, while the last factor is $\emptyset_0$ if $n-d_k=1$.
\begin{remark}[Shorthand degrafting]\label{rmk:shorthand}
  A quick shorthand method to determine all the possible degrafting terms for a given descent, $\df=(d_1\cdots d_k)_n$,
  is to scan through the binary representation given in Definition~\ref{def:binaryrepresentation} left to right as follows.
  If $\df$ starts with a `$0$', replace it by $\emptyset_0\star$. This gives the first factorisation. Then, in turn, look to each `$1$'
  in position $d_\ell$, and assuming there is a `$0$' immediately to its right, replace the pair `$10$' by `$0\star$'.
  In the case of the final `$1$' in the binary representation, if there is a single zero beyond it, replace the `$10$' by `$0\star\emptyset_0$'. 
  This procedure generates all possible factorisations.
\end{remark}

We extend degrafting to $\R\langle\mathbb D\rangle_{\otimes}$ as follows. We set, 
\begin{equation*}
  \Delta_2\coloneqq(\Delta\otimes\id)\circ\Delta,~\text{and}~
  \Delta_k\coloneqq(\Delta\otimes\id^{\otimes(k-1)})\circ\Delta_{k-1},
\end{equation*}
where $\Delta_2\colon\mathbb D\to\mathbb D^{\otimes3}$, and for $k\geqslant3$, $\Delta_k\colon\mathbb D\to\mathbb D^{\otimes(k+1)}$.
See, for example, Klausner~\cite{Klausner}.
\begin{example}\label{ex:multipledegrafting}
  We observe that, 
  \begin{align*}
    \Delta\circ 24_5&=\Delta\circ 01010\\
                    &=\emptyset_0\otimes1010+00\otimes10+0100\otimes\emptyset_0\\
                    &=\emptyset_0\otimes13_4+\emptyset_2\otimes1_2+2_4\otimes\emptyset_0.
  \end{align*}
  Then applying $\Delta\otimes\id$ and using that $\Delta(\emptyset_0)=0$ and that
  $\Delta\circ2_4=\Delta\circ0100=\emptyset_0\otimes1_3+\emptyset_2\otimes\emptyset_1$, we find,
  \begin{align*}
    (\Delta\otimes\id)\circ\Delta\circ 24_5
    =&\;(\Delta\circ\emptyset_2)\otimes1_2+(\Delta\circ2_4)\otimes\emptyset_0\\
    =&\;\emptyset_0\otimes\emptyset_1\otimes1_2+\emptyset_0\otimes1_3\otimes\emptyset_0\\
    &\;+\emptyset_2\otimes\emptyset_1\otimes\emptyset_0. 
  \end{align*}  
\end{example}
\begin{remark}[Convention]\label{rmk:convention}
A descent of grade $n$ cannot be degrafted more than $n$ times. We thus assume that $\Delta_{n-1}\circ\df=0$ for any $\df$ such that $\|\df\|>n$.
\end{remark}

We define the following endomorphisms on $\mathbb D$. We have already seen the \emph{identity}
endomorphism on $\mathbb D$, which extends linearly, and has the following action on any $\df\in\mathbb D$,
\begin{equation*}
    \id\colon\df\mapsto\df.
\end{equation*}
We now define the endomorphism $J\colon\mathbb D\to\mathbb D$ by,
\begin{equation*}
    J\colon\df\mapsto\begin{cases} \df, & \text{if}~ \df\neq\emptyset_0,\\ 0, & \text{if}~ \df=\emptyset_0, \end{cases}
\end{equation*}
which also extends linearly.
In other words $J$ sends $\emptyset_0$ to zero, while it acts like the identity on any descent which is not $\emptyset_0$.
It is analogous to the augmented ideal projector, see Reutenauer~\cite{Reutenauer} or Ebrahimi--Fard \textit{et al.}\/ \cite{E-FLMM-KW}.
Suppose that $H_1,H_2,\ldots,H_k$ are endomorphisms on $\R\langle\mathbb D\rangle_{\cendot}$.
Then we can apply $H_1\otimes H_2\otimes\cdots\otimes H_k$ to any $k$-tensor of descents,
and, for example, construct operations on $\R\langle\mathbb D\rangle_{\cendot}$, for any $\df\in\mathbb D$, of the form,
\begin{equation}\label{eq:endo1}
\cendot\circ(H_1\otimes H_2\otimes\cdots\otimes H_k)\circ\Delta_{k-1}\circ\df.
\end{equation}
In this construction, given a descent $\df$, we degraft $\df$ into the appropriate linear combination of $k$-tensored descents,
apply the endomorphisms $H_1$ through to $H_k$ to the respective components of each $k$-tensored term, and then finally, simply
replace all the tensors, `$\otimes$', by the equivalent tensor product, `$\cendot$'.
As is natural, see for example Reutenauer~\cite{Reutenauer}, and with a slight abuse of notation, we can represent the construction \eqref{eq:endo1} above by,
\begin{equation}\label{eq:endo2}
H_1\cendot H_2\cendot\cdots\cendot H_k\circ\df.
\end{equation}
Let us now re-examine $b^\ast_{n-2}$, $b^\ast_{n-3}$ and $b^\ast_{n-4}$ in \eqref{eq:bn2dg}--\eqref{eq:bn4dg}.
We observe that,
\begin{align*}
  b^\ast_{n-2}=&-(J\circ\emptyset_1)\botimes\emptyset_1,\\
  b^\ast_{n-3}=&-(J\circ\emptyset_2)\botimes\emptyset_2-\bigl((J-J\cendot\id)\circ 1_2\bigr)\botimes1_2.
\end{align*}
In the case of $b^\ast_{n-4}$ we have,
\begin{align*}
  b^\ast_{n-4}=&-(J\circ\emptyset_3)\botimes\emptyset_3-\bigl((J-J\cendot\id)\circ 1_3\bigr)\botimes1_3\\
         &-\bigl((J-J\cendot\id)\circ 2_3\bigr)\botimes2_3\\
         &-\bigl((J-J\cendot\id+J\cendot\id^{\cendot2})\circ 12_3\bigr)\botimes12_3.
\end{align*}
These results suggest our main result, Theorem~\ref{thm:Satocoeffsdescents}, which we present in Section~\ref{sec:Satocoeffs}.
Before we prove our main result though, there is one further intermediate step we need. This is to establish a general formula for the weight character of a general descent.
We tackle this next in Section~\ref{sec:weightchar}.

\section{Weight formula}\label{sec:weightchar}
Herein we establish a general formula for the weight character of a general descent, i.e.\/ for $\omega\circ(d_1\cdots d_\ell)_n$.
This is crucial for the proof of our main result, Theorem~\ref{thm:Satocoeffsdescents}.
Recall Definition~\ref{def:descentweight} for the descent weight $\omega$.
Since any descent $\df$ can be constructed from any of the possible ways of factoring it into two descents, we can compute the weight $\omega\circ\df$
associated with any descent $\df$ as follows. We degraft $\df$ and then compute the weight of each of the factorised terms using Definition~\ref{def:descentweight}.
We thus compute,
\begin{equation}\label{eq:practicalweight}
\omega\circ\df=\omega\circ\bigl(\star\circ(\id\otimes\id)\circ\Delta\bigr)\circ\df.
\end{equation}
The degrafting operation, $\Delta$, splits $\df$ into the sum of all possible pairs of factors generating $\df$.
The operator $\id\otimes\id$, applies the identity to both factors.
The grafting operation `$\star$' grafts all those pairs of factors together---with a slight abuse of notation,
we suppose $\star\circ(\df_1\otimes\df_2)=\df_1\star\df_2$.
Then we can apply the weight operator $\omega$, utilising the property given in Definition~\ref{def:descentweight}.
A convenient shorthand notation for the inner triple of operations is,
\begin{equation*}
\id^{\star2}\coloneqq\star\circ(\id\otimes\id)\circ\Delta. 
\end{equation*}
The following set of binomial coefficient identities are crucial to the proof of the weight formula. 
They are established by straightforward computation.
\begin{lemma}[Binomial identities]\label{lemma:binomialids}
We have the following three binomial identities, the sum identity,
\begin{subequations}\label{eq:identities}
\begin{align}  
C_0^{n-d-1}+C_1^{n-d}+C_2^{n-d+1}&+\cdots+C_d^{n-1}\equiv C_d^n,\label{eq:binomialsum}
\intertext{and the product identities,}
  C^n_{d}+C^{n}_{d-1}&\equiv C^{n+1}_{d}\label{eq:sum-simple}\\
  C^n_{d_1}C^{n-d_1}_{d_2-d_1}&\equiv C^n_{d_2}C^{d_2}_{d_1}.\label{eq:prod-prod}
\end{align}
\end{subequations}
\end{lemma}
We also frequently use the following property.
\begin{definition}[DT-invariance]\label{def:DTinvariance}
  We say that a set of binomial coefficient terms, involving the components $d_1,\ldots,d_\ell$ and $n$ of a generic descent, are \emph{DT-invariant}, if the terms are invariant to the transformation, 
\begin{equation}\label{eq:DT}
  d_k\to d_k-1\qquad\text{and}\qquad n\to n-1,
\end{equation}
for $k=1,\ldots,\ell$. DT-invariance is short for unit \emph{decrease transformation invariance}. 
\end{definition}
Let us consider some simple example cases.
\begin{example}\label{ex:weightlength1}
Let us compute $\omega\circ(d)_n$. To determine this real-valued coefficient, in light of \eqref{eq:practicalweight}, we observe that,
\begin{equation*}
      (d)_n=\underbrace{0\cdots0}_{d-1}1\underbrace{0\cdots0}_{n-d}, 
\end{equation*}
and thus $ \id^{\star 2}\circ(d)_n=\emptyset_0\star(d-1)_{n-1}+\emptyset_d\star\emptyset_{n-d-1}$.
Then using the properties of the descent weight from Definiton~\ref{def:descentweight}, we observe that,  
\begin{equation*}
  \omega\circ(d)_n=\omega\circ\bigl(\emptyset_0\star(d-1)_{n-1}\bigr)+\omega\circ\bigl(\emptyset_d\star\emptyset_{n-d-1}\bigr),
\end{equation*}
and, since $\omega\circ\emptyset_d=\omega\circ\emptyset_{n-d-1}=1$, that,
\begin{equation}
  \omega\circ(d)_n=\omega\circ(d-1)_{n-1}+C_d^{n-1}. \label{eq:weightlength1intermed}
\end{equation}
We can iterate the relation \eqref{eq:weightlength1intermed}, giving,
\begin{align*}
  \omega\circ(d)_n=&\;\omega\circ(1_{n-d-1})+C_2^{n-d+1}+\cdots+C_d^{n-1}\\
  =&\;C_1^{n-d}+C_2^{n-d+1}+\cdots+C_d^{n-1}.
\end{align*}
Here we used that,
\begin{equation*}
  \omega\circ(1_{n-d+1})=\omega\circ(1\underbrace{0\cdots0}_{n-d})=\omega\circ (0\star\underbrace{0\cdots0}_{n-d-1}),
\end{equation*}
and thus $\omega\circ(1_{n-d+1})=\omega\circ(\emptyset_1\star\emptyset_{n-d-1})=C_1^{n-d}$. Using the standard binomial sum \eqref{eq:binomialsum}, we find,
\begin{equation}
  \omega\circ(d)_n=C_d^{n}-1. \label{eq:weightlength1}
\end{equation}
\end{example}
\begin{example}\label{ex:weightlength2} 
  We compute $\omega\circ (d_1d_2)_n$ via a similar approach.
  If we degraft $(d_1d_2)_n$, for example using the shorthand approach in Remark~\ref{rmk:shorthand}, as in Example~\ref{ex:weightlength1}, we find,
  \begin{align*}
    \id^{\star 2}\circ(d_1d_2)_n=&\;\emptyset_0\star(d_1-1)(d_2-1)_{n-1}\\
                               &\;+\emptyset_{d_1}\star(d_2-d_1-1)_{n-d_1-1}\\
                               &\;+(d_1)_{d_2}\star\emptyset_{n-d_2-1}.
  \end{align*}
  Then from Definition~\ref{def:descentweight}, we observe that,    
  \begin{align}
    \omega\circ(d_1d_2)_n=&\;\omega\circ(d_1-1)(d_2-1)_{n-1}\nonumber\\
                          &\;+C^{n-1}_{d_1}\omega\circ(d_2-d_1-1)_{n-d_1-1}\nonumber\\
                          &\;+C_{d_2}^{n-1}\omega\circ(d_1)_{d_2},\label{eq:weightlength2intermed}
  \end{align}
  where we used that $\omega\circ\emptyset_d\equiv1$. As in Example~\ref{ex:weightlength1} we iterate this expression.
  At the next iteration, instead of $\omega\circ(d_1-1)(d_2-1)_{n-1}$ we would have $\omega\circ(d_1-2)(d_2-2)_{n-2}$,
  and in addition to the terms on the right, we would have similar terms but with the replacements indicated in \eqref{eq:DT}. 
  We observe that the term $\omega\circ(d_2-d_1-1)_{n-d_1-1}$ is DT-invariant.
  Further, we know from Example~\ref{ex:weightlength1} that $\omega\circ(d_1)_{d_2}=C^{d_2}_{d_1}-1$.
  Thus using identity \eqref{eq:prod-prod}, we observe that the last term above equals,
  \begin{equation*}
    C_{d_2}^{n-1}(C^{d_2}_{d_1}-1)\equiv C_{d_1}^{n-1}C^{n-d_1-1}_{d_2-d_1}-C_{d_2}^{n-1}.
  \end{equation*}
  Note that the term $C^{n-d_1-1}_{d_2-d_1}$ is also DT-invariant.
  Thus as we iterate the relation \eqref{eq:weightlength2intermed} down to $d_1=1$, we find,
  \begin{align*}
    \omega\circ(d_1d_2)_n=&\;(C^{n}_{d_1}-1)(C^{n-d_1-1}_{d_2-d_1-1}-1)\\
                          &\;+(C_{d_1}^{n}-1)C^{n-d_1-1}_{d_2-d_1}-(C_{d_2}^{n}-C_{d_2-d_1}^{n-d_1}).
  \end{align*}
  Here we used \eqref{eq:binomialsum} directly, and, also used \eqref{eq:binomialsum} to deduce,
  \begin{equation*}
      C^{n-d_1}_{d_2-d_1+1}+\cdots+C^{n-2}_{d_2-1}+C^{n-1}_{d_2}=C_{d_2}^{n}-C_{d_2-d_1}^{n-d_1}.
  \end{equation*}
  Using the identities \eqref{eq:identities}, then reveals,
  \begin{equation*}
     \omega\circ(d_1d_2)_n=C_{d_1}^{n}C^{n-d_1}_{d_2-d_1}-C^n_{d_1}-C^n_{d_2}+1.
  \end{equation*}
\end{example}
\begin{example}\label{ex:weightlength3} 
  Computing the case, $\omega\circ(d_1d_2d_3)_n$, is particularly insightful. It provides the template for the proof in the general case. We use the following steps.
  
  \emph{Step~1:} As in Examples~\ref{ex:weightlength1} and \ref{ex:weightlength2}, we compute $\id^{\star2}\circ(d_1d_2d_3)_n$ and then apply $\omega$. This gives,
   \begin{align}
     \omega\circ(d_1d_2&d_3)_n\nonumber\\
                      =&\;\omega\circ(d_1-1)(d_2-1)(d_3-1)_{n-1}\nonumber\\
                       &\;+C_{d_1}^{n-1}\omega\circ(d_2-d_1-1)(d_3-d_1-1)_{n-d_1-1}\nonumber\\
                       &\;+C_{d_2}^{n-1}\omega\circ(d_1)_{d_2}\cdot\omega\circ(d_3-d_2-1)_{n-d_2-1}\nonumber\\
                       &\;+C_{d_3}^{n-1}\omega\circ(d_1d_2)_{d_3}.\label{eq:weight3part1}
   \end{align}
   
  \emph{Step~2:} We substitute the formulae we already know, here from Examples~\ref{ex:weightlength1} and \ref{ex:weightlength2}, for all the weights shown apart from the very first term.
  The goal, after this substitution, is to express all the terms present in so-called ``Standard Form''. We explain what we mean by this presently.
  To this end, consider the second term on the right. Using Example~\ref{ex:weightlength2}, we observe,
  \begin{align*}
    C_{d_1}^{n-1}&\omega\circ(d_2-d_1-1)(d_3-d_1-1)_{n-d_1-1}\\
                          =&\;C_{d_1}^{n-1}\bigl(C^{n-d_1-1}_{d_2-d_1-1}C^{n-d_2}_{d_3-d_2}-C^{n-d_1-1}_{d_2-d_1-1}-C^{n-d_1-1}_{d_3-d_1-1}+1\bigr)\\
                          =&\;C_{d_1}^{n-1}C^{n-d_1-1}_{d_2-d_1-1}C^{n-d_2}_{d_3-d_2}-C_{d_1}^{n-1}C^{n-d_1-1}_{d_2-d_1-1}\\
                           &\;-C_{d_1}^{n-1}C^{n-d_1-1}_{d_3-d_1-1}+C_{d_1}^{n-1}.
  \end{align*}
  We consider the terms on the right to be in ``Standard Form''. 
  The second term on the right in \eqref{eq:weight3part1} is given by, 
  \begin{align*}
    C_{d_2}^{n-1}&\omega\circ(d_1)_{d_2}\cdot\omega\circ(d_3-d_2-1)_{n-d_2-1}\\
              =&\;C_{d_2}^{n-1}(C_{d_1}^{d_2}-1)(C_{d_3-d_2-1}^{n-d_2-1}-1)\\
              =&\;C_{d_1}^{n-1}C_{d_2-d_1}^{n-d_1-1}C_{d_3-d_2-1}^{n-d_2-1}\\
              &\;-C_{d_1}^{n-1}C_{d_2-d_1}^{n-d_1-1}-C_{d_2}^{n-1}C_{d_3-d_2-1}^{n-d_2-1}+C_{d_2}^{n-1}.
  \end{align*}
  Again, the terms here are in ``Standard Form''. The third term on the right in \eqref{eq:weight3part1} can be rewritten as, 
  \begin{align*}
    C_{d_3}^{n-1}\omega\circ(d_1d_2)_{d_3}
              =&\;C_{d_3}^{n-1}\bigl(C_{d_2}^{d_3}C_{d_2-d_1}^{d_3-d_1}-C_{d_1}^{d_3}-C_{d_2}^{d_3}+1\bigr)\\
              =&\;C_{d_1}^{n-1}C_{d_2-d_1}^{n-d_1-1}C_{d_3-d_2}^{n-d_2-1}\\
               &\;-C_{d_1}^{n-1}C_{d_3-d_1}^{n-d_1-1}-C_{d_2}^{n-1}C_{d_3-d_2}^{n-d_2-1}\\
               &\;+C_{d_3}^{n-1},
  \end{align*}
  where we used the identity \eqref{eq:prod-prod} multiple times. The ``Standard Form'' moniker applies again.
  We observe that the final forms in each of the three expressions are such that in each term, the first factors are not DT-invariant while all the remaining factors are. 
  This is what the label ``Standard Form'' refers to, though we are more precise momentarily.
  We now subsistute each of the expressions above into \eqref{eq:weight3part1}.

  \emph{Step~3:} We simplify the resulting expression using identity \eqref{eq:sum-simple}. This gives,
  \begin{align}
     \omega\circ(d_1d_2d_3)_n=&\;\omega\circ(d_1-1)(d_2-1)(d_3-1)_{n-1}\nonumber\\
                       &\;+C_{d_1}^{n-1}C^{n-d_1}_{d_2-d_1}C^{n-d_2}_{d_3-d_2}-C_{d_1}^{n-1}C_{d_2-d_1}^{n-d_1}\nonumber\\
                       &\;-C_{d_1}^{n-1}C_{d_3-d_1}^{n-d_1}-C_{d_2}^{n-1}C_{d_3-d_2}^{n-d_2}\nonumber\\
                       &\;+C_{d_1}^{n-1}+C_{d_2}^{n-1}+C_{d_3}^{n-1}.\label{eq:weight3part2}
  \end{align}

  \emph{Step~4:} We iterate this expression, exactly as we did in Examples~\ref{ex:weightlength1} and \ref{ex:weightlength2}. This gives,
  \begin{align*}
    \omega\circ(d_1d_2d_3)_n=&\;(C_{d_1}^{n}-1)C^{n-d_1}_{d_2-d_1}C^{n-d_2}_{d_3-d_2}\\
                             &\;-(C_{d_1}^{n}-1)C_{d_2-d_1}^{n-d_1}-(C_{d_1}^{n}-1)C_{d_3-d_1}^{n-d_1}\\
                             &\;-(C_{d_2}^{n}-C_{d_2-d_1}^{n-d_1})C_{d_3-d_2}^{n-d_2}\\
                             &\;+(C_{d_1}^{n-1}-1)+(C_{d_2}^{n-1}-C_{d_2-d_1}^{n-d_1})\\
                             &\;+(C_{d_3}^{n-1}-C_{d_3-d_2}^{n-d_2}).
  \end{align*}
  In this calculation, as in Examples~\ref{ex:weightlength1} and \ref{ex:weightlength2}, we have used that after iteration `$d_1-1$' times,
  the first term in \eqref{eq:weight3part1} becomes $\omega\circ1(d_2-d_1-1)(d_3-d_1-1)_{n-d_1-1}$.
  This is expanded in the same manner as \eqref{eq:weight3part1} above, but does not contain the corresponding first term as the descent starts with a `$1$'.
  It is straightforward to show, using the ``Standard Form'', that this term equals the expected term at the end of the iteration (as if we had completed `$d_1$' iterations).    
  In the expression just above, we observe that each term consists of a first factor which is the difference of two expressions, the second term of which is often `$1$', though not always.
  
  \emph{Step~5:} We observe that the sum of all the expressions involving these second terms in the differences generates `$-1$'. Hence the final result is,
  \begin{align}
    \omega\circ(d_1d_2d_3)_n=&\;C_{d_1}^{n}C^{n-d_1}_{d_2-d_1}C^{n-d_2}_{d_3-d_2}-C_{d_1}^{n}C^{n-d_1}_{d_2-d_1}\nonumber\\
                            &\;-C_{d_1}^{n}C^{n-d_1}_{d_3-d_1}-C_{d_2}^{n}C^{n-d_2}_{d_3-d_2}\nonumber\\
                            &\;+C^n_{d_1}+C^n_{d_2}+C^n_{d_3}-1.\label{eq:weight3}
  \end{align}
\end{example}

From Examples~\ref{ex:weightlength1}--\ref{ex:weightlength3} the pattern is now discernable.
We model our proof of the general case on the steps we emphasised in Example~\ref{ex:weightlength3}. 
Before proceeding to the general case, we need to introduce some further useful notation.
First, we need the notion of ordered subsets.
\begin{definition}[Ordered subsets]\label{def:orderedsubsets}
  Given a set of unique integers, say $\{1,2,\ldots,\ell\}$, with a natural strict ordering, an \emph{ordered} subset $\alpha_1\alpha_2\ldots\alpha_m$ of $\{1,2,\ldots,\ell\}$,
  with $m\leqslant\ell$, is any subset of $\{1,2,\ldots,\ell\}$ for which the order of the elements of $\{1,2,\ldots,\ell\}$ is maintained.
  Thus $\alpha_1\alpha_2\cdots\alpha_m$ must be such that $\alpha_1<\alpha_2<\cdots<\alpha_m$.
  The maximal subset is $12\cdots,\ell$. The smallest subsets are $1$, $2$, \ldots, $\ell$ and the empty subset $\emptyset$.
  We use the notation,
  \begin{equation*}
    [1,\ldots,\ell]_{k},
  \end{equation*}
  to denote all the ordered subsets of $\{1,\ldots,\ell\}$ of length $k$.
  We use the notation $[1,\ldots,\ell]$ to denote the set of all possible ordered subsets of lengths $0\leqslant k\leqslant\ell$.
\end{definition}
\begin{example}[Ordered subsets]\label{ex:orderedsubsets}
The ordered subsets of $\{1,2,3,4\}$ are: $1234$; $123$; $124$; $134$; $234$; $12$; $13$; $14$; $23$; $24$; $34$; $1$; $2$; $3$; $4$ and $\emptyset$. 
\end{example}
\begin{example}[Binary representation]\label{ex:orderedsetsbinary}
  We can represent ordered subsets using binary numbers. For example, the ordered subsets of $\{1,2,3,4\}$ in Example~\ref{ex:orderedsubsets} are respectively represented by:
  $1111$, $1110$, $1101$, $1011$, $0111$, $1100$, $1010$, \ldots, $1000$, $0100$, $0010$, $0001$ and $0000$. The $1$'s in the binary numbers pick out the numbers in $\{1,2,3,4\}$
  that constitute the ordered subset. In general, we deduce that the number of ordered subsets of $\{1,2,\ldots,\ell\}$ is $2^\ell$.
\end{example}
We introduce the following shorthand notation to represent the product between binomial coefficients which are in ``Standard Form''. 
As mentioned in Example~\ref{ex:weightlength3}, this means that first factor may not be DT-invariant, while the remaining factors definitely are.
Indeed, for any set of binomial coefficients, say $C_{d_1}^n$,\ldots, $C_{d_\ell}^n$, we set,
\begin{equation}
  C^n_{d_1}\ast C^n_{d_2}\ast\cdots\ast C^n_{d_\ell}\coloneqq C^n_{d_1}C^{n-d_1}_{d_2-d_1}\cdots C^{n-d_{\ell-1}}_{d_\ell-d_{\ell-1}}. 
\end{equation}
We suppose, as is natural, that $1\ast q=q$ for any $q\in\mathbb Q$. We can also think of `$\ast$' as a multilinear product.
\begin{example}\label{ex:setupordering}
  Using the shorthand notation, we observe,
  \begin{align*}
    \omega\circ(d_1d_2d_3)_n=&\;C_{d_1}^{n}\ast C^{n}_{d_2}\ast C^{n}_{d_3}-C_{d_1}^{n}\ast C^{n}_{d_2}\\
                            &\;-C_{d_1}^{n}\ast C^{n}_{d_3}-C_{d_2}^{n}\ast C^{n}_{d_3}\\
                            &\;+C^n_{d_1}+C^n_{d_2}+C^n_{d_3}-1\\
                           =&\;\sum C_{d_{\alpha_1}}^n\ast C_{d_{\alpha_2}}^n\ast\cdots\ast C_{d_{\alpha_m}}^n,
  \end{align*}
  where the sum is over all ordered subsets $\{\alpha_1,\ldots,\alpha_m\}$ of $\{1,2,3\}$ with $m\leqslant3$.
\end{example}
Further abstraction of this notation is useful for the proof of our main result herein, Theorem~\ref{thm:weightformula}, just below.
\begin{definition}[Shorthand notation]\label{def:shorthandbinomials}
  We use the following shorthand notation for the products of binomial coefficients shown.
  For a given $n\in\mathbb N$, suppose $\{1,2,\ldots,\ell\}$ is a given set of unique integers,
  and $\alpha_1\alpha_2\cdots\alpha_m$ is an ordered subset of $\{1,2,\ldots,\ell\}$.
  Then we set,
  \begin{equation*}
    \alpha_1\alpha_2\cdots\alpha_m\coloneqq C_{d_{\alpha_1}}^n\ast C_{d_{\alpha_2}}^n\ast\cdots\ast C_{d_{\alpha_m}}^n.
  \end{equation*}
  The following variants on this notation are also useful, 
  \begin{equation*}
    \alpha_1^-\coloneqq C_{d_{\alpha_1}}^{n-1}\qquad\text{and}\qquad\alpha_1^=\coloneqq C_{d_{\alpha_1}-1}^{n-1},
  \end{equation*}
  and for all $k\in\{2,\ldots,m\}$, 
  \begin{align*}
    \alpha_k^-&\coloneqq C_{d_{\alpha_k}-d_{\alpha_{k-1}}}^{n-d_{\alpha_{k-1}}-1}, \\
    \alpha_k^=&\coloneqq C_{d_{\alpha_1}-d_{\alpha_{k-1}}-1}^{n-d_{\alpha_{k-1}}-1}.
  \end{align*}
\end{definition}
\begin{example}
 Consider Step~2 in Example~\ref{ex:weightlength3}. We respectively denote the final three terms on the right in \eqref{eq:weight3part1} as $W_1$, $W_2$ and $W_3$,
 written in ``Standard Form''. Using the new abstract notation, these terms are given by,
 \begin{align*}
   W_1&=1^-2^=3\phantom{^-}-1^-2^=-1^-3^=+1^-,\\
   W_2&=1^-2^-3^=-1^-2^--2^-3^=+2^-,\\
   W_3&=1^-2^-3^--1^-3^--2^-3^-+3^-.
 \end{align*}
 Each of the digits above represents a binomial coefficient. Further we observe that identity~\eqref{eq:sum-simple} translates to,
 \begin{equation}\label{eq:sum-simpleabstract}
   \alpha^=+\alpha^-\equiv\alpha.
 \end{equation}
 Using this identity for the sum of these three terms, we observe that the collapse in terms, from \eqref{eq:weight3part1} in Step~2 to \eqref{eq:weight3part2} in Step~3,
 is now much more transparent and,
 \begin{equation*}
 W_1+W_2+W_3=1^-23-1^-2-1^-3-2^-3+1^-\!+2^-\!+3^-.
 \end{equation*}
 Compare this with the final terms on the right in \eqref{eq:weight3part2}. 
 A further observation is useful. We say that $W_1$ `pivots' around `$1$' as all the terms contain the digit `$1^-$'.
 Note that for all the terms in $W_1$, the factor to the immediate right of $1^-$ has the form $\alpha^=$.
 We see that $W_2$ pivots around `$2$' as every term involves the digit `$2^-$'.
 In all the terms in $W_2$, the factors to the left have the form $\alpha^-$, while all those to the immediate right have the form $\alpha^=$.
 Finally we see, $W_3$ pivots around `$3^-$', and, in all the terms therein, the factors to the left of the `$3^-$' have the form $\alpha^-$.
\end{example}
We can now state and prove the main result of this section. 
\begin{theorem}[Weight formula]\label{thm:weightformula}
  The weight character associated with any descent $(d_1d_2\cdots d_\ell)_n$ is given by,
  \begin{equation*}
    \omega\circ (d_1d_2\cdots d_\ell)_n=(C^n_{d_1}-1)\ast(C^n_{d_2}-1)\ast\cdots\ast(C^n_{d_\ell}-1),
  \end{equation*}
  or equivalently, it is given by,
  \begin{equation*}  
    \omega\circ (d_1d_2\cdots d_\ell)_n=\sum (-1)^{m+\ell}C_{d_{\alpha_1}}^n\ast C_{d_{\alpha_2}}^n\ast\cdots\ast C_{d_{\alpha_m}}^n,
  \end{equation*}
  where the sum is over all ordered subsets $\{\alpha_1,\ldots,\alpha_m\}$ of $\{1,\ldots,\ell\}$ with $m\leqslant\ell$.
\end{theorem}
\begin{proof}
  We use induction. We know from Example~\ref{ex:weightlength1} that $\omega\circ(d)_n=C_d^{n}-1$.
  We assume the statement in the theorem holds for all $\ell^\prime\in\{2,\ldots,\ell-1\}$.
  Our goal is to prove that it holds for $\ell^\prime=\ell$. We use the same strategy as the that outlined in Example~\ref{ex:weightlength3}.
  
  \emph{Step~1.} We degraft $(d_1d_2\cdots d_\ell)_n$ and apply $\omega$ as per \eqref{eq:practicalweight}. This procedure shows us that,  
  \begin{align*}
    \omega\circ&(d_1d_2\cdots d_\ell)_n\\
                  =&\;\omega\circ\bigl((d_1\!-\!1)(d_2\!-\!1)\cdots(d_\ell\!-\!1)\bigr)_{n-1}\\                      
                   &\;+\sum_{k=1}^{\ell}C^{n-1}_{d_k}\Bigl(\omega\circ(d_1\cdots d_{k-1})_{d_k}\\
                   &\qquad\qquad\quad\cdot\omega\circ(d_{k+1}\!-\!d_k\!-\!1)\cdots(d_\ell\!-\!d_k\!-\!1)_{n-d_k-1}\Bigr),
  \end{align*}
  where, when $k=1$ the left factor is $\omega\circ\emptyset_{d_1}$, and when $k=\ell$ the right factor is $\omega\circ\emptyset_{n-d_\ell-1}$.

  \emph{Step~2.} We substitute for the weights shown in the formula in Step~1, and convert all the terms to ``Standard Form''.
  Consider a generic term in the sum above with left factor, $C^{n-1}_{d_k}\omega\circ(d_1\cdots d_{k-1})_{d_k}$, and right factor, $\omega\circ(d_{k+1}-d_k-1)\cdots(d_\ell-d_k-1)_{n-d_k-1}$.
  For $k\in\{1,\ldots,\ell\}$, each of these factors involve descents of length $\ell-1$ or less. We can thus invoke our induction assumption for their explicit formulae.
  We observe that the leading term in $C^{n-1}_{d_k}\omega\circ(d_1\cdots d_{k-1})_{d_k}$ can be expressed as,
  \begin{align*}
    C_{d_k}^{n-1}\omega&\circ(d_1\cdots d_{k-1})_{d_k}\\
                     =&\;C_{d_k}^{n-1}\bigl(C_{d_1}^{d_k}C_{d_2-d_1}^{d_k-d_1}\cdots C_{d_{k-1}-d_{k-2}}^{d_k-d_{k-2}}+\cdots\bigr)\\
                     =&\;C_{d_1}^{n-1}C_{d_k-d_1}^{n-d_1-1}C_{d_2-d_1}^{d_k-d_1}\cdots C_{d_{k-1}-d_{k-2}}^{d_k-d_{k-2}}+\cdots\\
                     =&\;C_{d_1}^{n-1}C_{d_2-d_1}^{n-d_1-1}C_{d_k-d_2}^{n-d_2-1}\cdots C_{d_{k-1}-d_{k-2}}^{d_k-d_{k-2}}+\cdots\\
                     \vdots&\\
                     =&\;C_{d_1}^{n-1}C_{d_2-d_1}^{n-d_1-1}\cdots C_{d_{k-1}-d_{k-2}}^{n-d_{k-2}-1}C_{d_{k}-d_{k-1}}^{n-d_{k-1}-1}+\cdots\\
                     =&\;1^-2^-3^-\cdots(k-1)^-k^-+\cdots.
  \end{align*}
  In the computation above, we propagated the identity \eqref{eq:prod-prod} successively through the factors shown. 
  All the other terms in the factor $C^{n-1}_{d_k}\omega\circ(d_1\cdots d_{k-1})_{d_k}$ that we have not shown above,
  consist of exactly the same forms but involving ordered subsets of $\{d_1,\ldots,d_{k-1}\}$.
  The exact same process can be applied to them, and we deduce that,
  \begin{equation*}
    C^{n-1}_{d_k}\omega\circ(d_1\cdots d_{k-1})_{d_k}=\sum(-1)^{m+k-1}\alpha_1^-\cdots\alpha_m^-k^-,
  \end{equation*}
  where the sum is over all ordered subsets $\alpha_1\cdots\alpha_m$ of $\{1,\ldots,k-1\}$, with $m\leqslant k-1$. 
  We now consider $\omega\circ(d_{k+1}-d_k-1)\cdots(d_\ell-d_k-1)_{n-d_k-1}$, which is straightforward.
  The leading order term is,
  \begin{align*}
    \omega\circ(d_{k+1}-d_k&-1)\cdots(d_\ell-d_k-1)_{n-d_k-1}\\
                      =&\;C_{d_{k+1}-d_k-1}^{n-d_k-1}C_{d_{k+2}-d_{k+1}}^{n-d_{k+1}}\cdots C_{d_\ell-d_{\ell-1}}^{n-d_{\ell-1}}+\cdots\\
                      =&\;(k+1)^=(k+2)\cdots (\ell-1)\ell+\cdots.
  \end{align*}
  We note that the first binomial coefficient above also has `$d_k$' subtracted from both the upper and lower sufficies.  
  Again, all the terms not shown involve ordered subsets of $\{k+1,\ldots,\ell\}$ and have the same form apart from the fact they are indexed by their ordered subset.
  We deduce,
  \begin{multline*}
    \omega\circ(d_{k+1}-d_k-1)\cdots(d_\ell-d_k-1)_{n-d_k-1}\\ =\sum(-1)^{m^\prime+\ell-k-1}\beta_1^=\beta_2\cdots\beta_{m^\prime},
  \end{multline*}
  where the sum is over all the ordered subsets $\beta_1\beta_2\cdots\beta_{m^\prime}$ of $\{k+1,\ldots,\ell\}$ with $m^\prime\leqslant\ell-k$.
  Putting these results together we have,
  \begin{align}
    \omega\circ(d_1&d_2\cdots d_\ell)_n\nonumber\\
                  =&\;\omega\circ\bigl((d_1\!-\!1)(d_2\!-\!1)\cdots(d_\ell\!-\!1)\bigr)_{n-1}\nonumber\\  
                   &\;+\sum_{k=1}^\ell\sum(-1)^{m+m^\prime+\ell}\alpha_1^-\cdots\alpha_m^-k^-\beta_1^=\beta_2\cdots\beta_{m^\prime},\label{eq:orderedsuminterim}
  \end{align}
  where the second sum is over both the ordered subsets: $\alpha_1\cdots\alpha_m$ from $\{1,\ldots,k-1\}$ and $\beta_1\beta_2\cdots\beta_{m^\prime}$ from $\{k+1,\ldots,\ell\}$.
  We have thus achieved ``Standard Form'' in \eqref{eq:orderedsuminterim}. 
  
  \emph{Step~3.} We now simplify the form for $\omega\circ(d_1d_2\cdots d_\ell)_n$ in ``Standard Form'' just above using the identity \eqref{eq:sum-simpleabstract}.
  The extremal cases of when the term $\alpha_1^-\cdots\alpha_m^-k^-\beta_1^=\beta_2\cdots\beta_{m^\prime}$ is maximal, and, when both the ordered sets $\alpha_1\cdots\alpha_m$ and $\beta_1\beta_2\cdots\beta_{m^\prime}$
  are empty, are naturally easiest. In the latter case, corresponding to the terms of lowest degree, the sum in \eqref{eq:orderedsuminterim} is simply a sum of $k^-$ over the $k$-values shown, giving, $1^-+2^-+\cdots+\ell^-$.  
  Now consider the former case. When $\alpha_1^-\cdots\alpha_m^-k^-\beta_1^=\beta_2\cdots\beta_{m^\prime}$ is maximal, the sum of the terms from $k=1$ to $k=\ell$ is,
  \begin{align}
   &\phantom{+}1^-2^=3\phantom{^-}4\phantom{^-}5\phantom{^-}6\phantom{^-}\cdots(\ell-2)\phantom{^-}(\ell-1)\phantom{^-}\ell\phantom{^-}\nonumber\\
   &+1^-2^-3^=4\phantom{^-}5\phantom{^-}6\phantom{^-}\cdots(\ell-2)\phantom{^-}(\ell-1)\phantom{^-}\ell\phantom{^-}\nonumber\\
   &+1^-2^-3^-4^=5\phantom{^-}6\phantom{^-}\cdots(\ell-2)\phantom{^-}(\ell-1)\phantom{^-}\ell\phantom{^-}\nonumber\\
   &+1^-2^-3^-4^-5^=6\phantom{^-}\cdots(\ell-2)\phantom{^-}(\ell-1)\phantom{^-}\ell\phantom{^-}\nonumber\\
    &\;\vdots\nonumber\\
   &+1^-2^-3^-4^-5^-6^-\cdots(\ell-2)^=(\ell-1)\phantom{^-}\ell\phantom{^-}\nonumber\\
   &+1^-2^-3^-4^-5^-6^-\cdots(\ell-2)^-(\ell-1)^=\ell\phantom{^-}\nonumber\\
   &+1^-2^-3^-4^-5^-6^-\cdots(\ell-2)^-(\ell-1)^-\ell^=\nonumber\\
   &+1^-2^-3^-4^-5^-6^-\cdots(\ell-2)^-(\ell-1)^-\ell^-\nonumber\\
   =&1^-23456\cdots(\ell-1)\ell.\label{eq:interimresult}
  \end{align}
  Here we used the identity \eqref{eq:sum-simpleabstract} in successive terms, from the final term in the sum back through to the first term. 
  The number of terms in the sum in \eqref{eq:orderedsuminterim} exceeds the number of ordered subsets of $\{1,\ldots,\ell\}$, though the empty subset is not present in \eqref{eq:orderedsuminterim}.
  Let us pick an arbitrary ordered subset $\{p_1,p_2,\ldots,p_j\}$ of $\{1,\ldots,\ell\}$ where we suppose $2\leqslant j\leqslant\ell-1$.
  In this restriction for $j$ we exclude the empty set, and the other two maximal cases we have already considered.
  For a given $j$, as $k$ cycles through $1$ to $\ell$ in the first sum in \eqref{eq:orderedsuminterim}, there are obviously $j$ instances when $k=p_i$ for $i=1,\ldots,j$.
  In each of those cases, there are unique ordered subsets $\alpha_1\cdots\alpha_m$ from $\{1,\ldots,k-1\}$ and $\beta_1\beta_2\cdots\beta_{m^\prime}$ from $\{k+1,\ldots,\ell\}$,
  such that $\alpha_1\cdots\alpha_m=p_1\cdots p_m$ with $m\leqslant i-1$ and $\beta_1\beta_2\cdots\beta_{m^\prime}=p_{i+1}\cdots p_{j}$ with $m^\prime\leqslant j-i-1$.
  Thus in the sum in \eqref{eq:orderedsuminterim} there will be $j$ terms of the form $p_1^-\cdots p_i^-p_{i+1}^=p_{i+2}\cdots p_j$; though note that the final term is just $p_1^-\cdots p_j^-$.
  If we sum all these terms over the cases $k=p_i$, the structure of the sum exactly mirrors that on the left in \eqref{eq:interimresult}, though obviously for the ordered subset $\{p_1,\ldots,p_j\}$.
  Using the identity in a similar fashion to that in \eqref{eq:interimresult}, shows that the sum in this case collapses to $p_1^-p_2p_3\cdots p_j$.
  We have thus succeeded in demonstrating that formula in \eqref{eq:orderedsuminterim} collapses to,
  \begin{align}
    \omega\circ(d_1d_2\cdots d_\ell)_n=&\;\omega\circ\bigl((d_1\!-\!1)(d_2\!-\!1)\cdots(d_\ell\!-\!1)\bigr)_{n-1}\nonumber\\  
                                      &\;+\sum(-1)^{j+\ell}p_1^-p_2\cdots p_j,\label{eq:orderedsuminterim3}
  \end{align}
  where the sum is over all non-empty ordered subsets $p_1\cdots p_j$ of $\{1,\ldots,\ell\}$ with $j\leqslant\ell$.
  We also used that $m+m^\prime+1=j$.
  
  \emph{Step~4.} We now iterate \eqref{eq:orderedsuminterim3}.
  We note that the factors $p_2\cdots p_m$ are DT-invariant, and exactly as outlined in Example~\ref{ex:weightlength3}, we find that,
  \begin{align}
    \omega\circ&(d_1d_2\cdots d_\ell)_n\nonumber\\
                     =&\;\sum(-1)^{j+\ell}\bigl(C_{d_{p_1}}^n-C_{d_{p_1}-d_1}^{n-d_1}\bigr)C_{d_{p_2}-d_{p_1}}^{n-d_{p_1}}\cdots C_{d_{p_j}-d_{p_{j-1}}}^{n-d_{p_{j-1}}}\nonumber\\
                     =&\;\sum(-1)^{j+\ell}C_{d_{p_1}}^nC_{d_{p_2}-d_{p_1}}^{n-d_{p_1}}\cdots C_{d_{p_j}-d_{p_{j-1}}}^{n-d_{p_{j-1}}}\nonumber\\
                      &\;-\sum(-1)^{j+\ell} C_{d_{p_1}-d_1}^{n-d_1}C_{d_{p_2}-d_{p_1}}^{n-d_{p_1}}\cdots C_{d_{p_j}-d_{p_{j-1}}}^{n-d_{p_{j-1}}},\label{eq:step4}
  \end{align}
  where the sums are over all ordered subsets $\{p_1\ldots,p_j\}$ of $\{1,\ldots,\ell\}$ with $j\leqslant\ell$. Note that the first term on the right just above matches the formula given in the theorem.
  
  \emph{Step~5.} In the final step, we show that the second term on the right in \eqref{eq:step4} equals $(-1)^\ell$.
  If we code the second term on the right in \eqref{eq:step4} using our abstract notation, modulo the minus sign in front, it has the form,
  \begin{equation}\label{eq:secondterm}
  \sum (-1)^{j+\ell}p_1^\ast p_2p_3\cdots p_j,
  \end{equation}  
  where the sum is over the ordered subsets  $\{p_1\ldots,p_j\}$ of $\{1,\ldots,\ell\}$.
  We use the notation `$p_1^\ast$' do denote the fact that the first binomial factor is not $C_{d_{p_1}}^{n}$, but is $C_{d_{p_1}-d_1}^{n-d_1}$.
  Note in particular, that if the ordered subset $p_1p_2\cdots p_j$ starts with $1$, i.e.\/ we have $d_{p_1}=d_1$, then $1^\ast$ corresponds to $C_0^{n-d_1}=1$.
  This means that for each ordered subset $p_1p_2\cdots p_j$ of length $j$, we should separate those ordered subsets that start with a `$1$' from the other ordered subsets of that length.
  The ordered subsets $p_1p_2\cdots p_j$ of length $j$ that start with a `$1$', when in the form $p_1^\ast p_2\cdots p_j$ as in \eqref{eq:secondterm}, collapse in length by $1$ to become $p_2\cdots p_j$.
  However there are ordered subsets of length $j-1$ that would match this form, but have the opposite sign. These terms then cancel.
  Our exercise here therefore, is to enumerate these terms and show that from one level to the next there are exact cancellations.
  First, we note that there is of course a unique ordered subset of size $\ell$, namely,
  \begin{equation}\label{eq:set0}
    \{1,2,\ldots,\ell\}.
  \end{equation}
  The subindex keeps track of the length of the ordered subset. The ordered subsets of $\{1,\ldots,\ell\}$ of length $\ell-1$ are,
  \begin{equation}\label{eq:set1}
    [1,\ldots,\ell]_{\ell-1}=1[2,\ldots,\ell]_{\ell-2}\cup 2[3,\ldots,\ell]_{\ell-2}. 
  \end{equation}
  We observe that the ordered subset of size $\ell$ in \eqref{eq:set0} with the first factor $1$ collapsed, equals the second set on the right in the union of two sets in \eqref{eq:set1}, i.e.\/ that starting with $2$.
  The sign of the latter set is the opposite of that in the former and the two terms cancel in the sum \eqref{eq:secondterm}.
  The ordered subsets of $\{1,\ldots,\ell\}$ of length $\ell-2$ are,
  \begin{equation}\label{eq:set2}
    [1,\ldots,\ell]_{\ell-2}=1[2,\ldots,\ell]_{\ell-3}\cup 2[3,\ldots,\ell]_{\ell-3}\cup 3[4,\ldots,\ell]_{\ell-3}.  
  \end{equation}
  We observe that the first ordered subset in \eqref{eq:set1}, once the first factor $1$ is collapsed, exactly matches the pair $2[3,\ldots,\ell]_{\ell-3}\cup 3[4,\ldots,\ell]_{\ell-3}$ in \eqref{eq:set2}.
  The respective signs of the terms in this match are opposite, and so again these terms cancel in the sum \eqref{eq:secondterm}.
  The ordered subsets of $\{1,\ldots,\ell\}$ of length $\ell-3$ are,
  \begin{align}\label{eq:set3}
    [1,\ldots,\ell]_{\ell-3}=&\;1[2,\ldots,\ell]_{\ell-4}\cup 2[3,\ldots,\ell]_{\ell-4}\nonumber\\
                            &\;\cup 3[4,\ldots,\ell]_{\ell-4}\cup 4[5,\ldots,\ell]_{\ell-4}.  
  \end{align}
  Again, we observe that the first ordered subset in \eqref{eq:set2}, once the first factor $1$ is collapsed,
  exactly matches the triple $2[3,\ldots,\ell]_{\ell-4}\cup 3[4,\ldots,\ell]_{\ell-4}\cup 4[5,\ldots,\ell]_{\ell-4}$ in \eqref{eq:set3}.
  Their signs are opposite and these cancel in the sum \eqref{eq:secondterm}. And so forth. The final collection is all the ordered subsets of $\{1,\ldots,\ell\}$ of length $1$ which are,
  \begin{equation}\label{eq:setl}
    [1,\ldots,\ell]_{1}=1[2,\ldots,\ell]_{0}\cup 2[3,\ldots,\ell]_{0}\cup\cdots\cup \ell[\emptyset]_{0}. 
  \end{equation}
  By the same argument the set $2[3,\ldots,\ell]_{0}\cup\cdots\cup \ell[\emptyset]_{0}$ matches the appropriate term at the length level $2$, with the signs opposite and the same cancellation occuring.
  This leaves one remainig term, the first one in \eqref{eq:setl}, namely $1$, or in terms of the sum in \eqref{eq:secondterm} this is the term $1^\ast$, which then collapses to the real value $1$.
  Simple counting, and recalling that the second term on the right in \eqref{eq:step4} is preceeded by a minus sign, completes the proof.
\qed
\end{proof}

\section{Sato coefficients}\label{sec:Satocoeffs}       
Herein we establish our main result, Theorem~\ref{thm:Satocoeffsdescents}.  
We set $S$ to be the endomorphism on $\R\langle\mathbb D\rangle_{\cendot}$ given by,
\begin{equation*}
S\coloneqq \sum_{m\geqslant0}\bigl((-1)^{m}J\cendot\id^{\cendot m}.
\end{equation*}
We begin with the following result---see the examples at the end of Section~\ref{sec:degrafting}. 
\begin{lemma}[Sato coefficient endomorphism]\label{lemma:Satocoeffsviadescents}
  The Sato coefficients $b^\ast_{n-k}$ are given by,
  \begin{equation*}
    b^\ast_{n-k}=-\sum_{\|\df\|=k-1} S\circ\df\botimes\df.
  \end{equation*}
\end{lemma}
\begin{proof}
  This follows from the formula for $b_{n-k}^\ast$ in \eqref{eq:bnkformula}, and using that linear combinations of the form \eqref{eq:lincombdescents} can be rewritten in the form \eqref{eq:adjoint}.
  The sum in the formula for $b_{n-k}^\ast$ is over all descents $\df_1,\df_2,\ldots,\df_\ell\in\mathbb D$ satisfying \eqref{eq:descentcomps}. 
  The products `$\df_1\star\cdots\star\df_\ell$' thus generate all possible descents $\df$ such that $\|\df\|=k-1$.
  We can generate all the possible `$\star$' products of descents that generate such descents $\df$ by degrafting.
  The terms in $S$ achieve this precise resummation, and the first factor $J$ in each term therein, ensures the first factor in any degrafting is not $\emptyset_0$,
  corresponding to the compositions that start with `$1$'.
\qed
\end{proof}
Before proceeding to our main result we consider some example cases.
Recall the Sato character $\gamma_n$ given in Definition~\ref{def:Satochar}.
For convenience, given a string of integers $c_1c_2\cdots c_\ell$, we define the map,
\begin{equation*}
\hat{\gamma}_n(c_1c_2\cdots c_\ell)\coloneqq\gamma_n\bigl((c_1+1)(c_2+1)\cdots(c_\ell+1)\bigr).
\end{equation*}
More explicitly, using Definition~\ref{def:Satochar}, we observe that,
\begin{equation*}
     \hat{\gamma}_n(c_1c_2\cdots)\coloneqq C^n_{c_1}\,C^{n-c_1-1}_{c_2}\,C^{n-c_1-c_2-2}_{c_3}\cdots.
\end{equation*}
Again, as for $\gamma_n$, with a slight abuse of notation, we may use $\hat{\gamma}_n=\hat{\gamma}_n(\df_1,\ldots,\df_k)$ to denote $\hat{\gamma}_n=\hat{\gamma}_n\bigl(\|\df_1\|\cdots\|\df_k\|\bigr)$.
\begin{example}\label{ex:main2}
  We compute $\sigma_n\circ S\circ (d_1d_2)_n$. In this case, since $(d_1d_2)_n$ can only maximally be degrafted into three components,
  we have, $S=-J+J\cendot\id-J\cendot\id^{\cendot2}$. Recall the definition of the Sato weight map from \eqref{eq:sigmadef}. 
  First, we need to compute $\sigma_n\circ J\circ (d_1d_2)_n=\sigma_n\circ(d_1d_2)_n=\hat{\gamma}_n\circ(d_1d_2)_n\cdot\omega\circ(d_1d_2)_n$.
  The first factor is given by,
  \begin{equation*}
  \hat{\gamma}_n\circ(d_1d_2)_n=\hat{\gamma}_n\circ\|(d_1d_2)_n\|=\hat{\gamma}_n\circ n=C_n^n=1.
  \end{equation*}
  The second factor is given by the weight formula in Theorem~\ref{thm:weightformula}, or explicitly in Example~\ref{ex:weightlength2}.
  Thus we have,
  \begin{align*}
    \sigma_n\circ J\circ (d_1d_2)_n=&\;\omega\circ (d_1d_2)_n\\
                                  =&\;C_{d_1}^{n}C^{n-d_1}_{d_2-d_1}\!-C^n_{d_1}\!-C^n_{d_2}+1\\
                                  =&\;12-1-2+\nu,
  \end{align*}
  where in the last line we are using the shorthand notation outlined in Definition~\ref{def:shorthandbinomials}, and we use `$\nu$' in the shorthand notation to correspond to the real binomial coefficient `$1$'.
  Second, we need to compute $\sigma_n\circ J\cendot\id\circ (d_1d_2)_n$. Either using Example~\ref{ex:weightlength2}, or the binary representation for descents,
  we observe that,
  \begin{align*}
    \Delta\circ(d_1d_2)_n=&\;\emptyset_0\otimes(d_1-1)(d_2-1)_{n-1}\\
                               &\;+\emptyset_{d_1}\otimes(d_2-d_1-1)_{n-d_1-1}\\
                               &\;+(d_1)_{d_2}\otimes\emptyset_{n-d_2-1}.
  \end{align*}
  Since $J\circ\emptyset_0=0$, we observe, $\sigma_n\circ J\cendot\id\circ(d_1d_2)_n$ equals,
  \begin{align*}
    \sigma_n\circ&\;\bigl(\emptyset_{d_1}\cendot(d_2-d_1-1)_{n-d_1-1}\bigr)+\sigma_n\circ\bigl((d_1)_{d_2}\cendot\emptyset_{n-d_2-1}\bigr)\\
    =&\;\hat{\gamma}_n\circ\bigl(d_1(n-d_1-1)\bigr)\,\omega\circ(d_2-d_1-1)_{n-d_1-1}\\
    &\;+\hat{\gamma}_n\circ\bigl(d_2(n-d_2-1)\bigr)\,\omega\circ(d_1)_{d_2}\\
    =&\;C_{d_1}^n\,\omega\circ(d_2-d_1-1)_{n-d_1-1}+C_{d_2}^n\,\omega\circ(d_1)_{d_2}\\
    =&\;C_{d_1}^n(C_{d_2-d_1-1}^{n-d_1-1}-1)+C_{d_2}^n(C_{d_1}^{d_2}-1)\\
    =&\;C_{d_1}^nC_{d_2-d_1-1}^{n-d_1-1}-C_{d_1}^n+C_{d_1}^nC_{d_2-d_1}^{n-d_1}-C_{d_2}^n\\
    =&\;12^=-1-12-2,
  \end{align*}
  where we used identity~\eqref{eq:prod-prod}. Third, we need to compute $\sigma_n\circ J\cendot\id^{\cendot2}\circ (d_1d_2)_n$.
  Since $\Delta\circ\emptyset_0=0$, we observe,
  \begin{align*}
    \Delta_2\circ (d_1d_2)_n=&\;(\Delta\otimes\id)\circ\Delta\circ (d_1d_2)_n\\
                           =&\;\Delta\circ\emptyset_{d_1}\otimes(d_2-d_1-1)_{n-d_1-1}\\
                            &\;+\Delta\circ(d_1)_{d_2}\otimes\emptyset_{n-d_2-1}\\
                           =&\;\emptyset_0\otimes\emptyset_{d_1-1}\otimes(d_2-d_1-1)_{n-d_1-1}\\
                            &\;+\emptyset_0\otimes(d_1-1)_{d_2-1}\otimes\emptyset_{n-d_2-1}\\
                            &\;+\emptyset_{d_1}\otimes\emptyset_{d_2-d_1-1}\otimes\emptyset_{n-d_2-1}.
  \end{align*}
  Again, since $J\circ\emptyset_0=0$, we observe,
  \begin{align*}
    \sigma_n\circ J\cendot\id^{\cendot2}\circ (d_1d_2)_n=&\;\sigma_n\circ\emptyset_{d_1}\otimes\emptyset_{d_2-d_1-1}\otimes\emptyset_{n-d_2-1}\\
                                                      =&\;\hat{\gamma}_n\circ\bigl(d_1(d_2\!-\!d_1\!-\!1)(n\!-\!d_2\!-\!1)\bigr)\\
                                                      =&\;C_{d_1}^nC_{d_2-d_1-1}^{n-d_1-1}\\
                                                      =&\;12^=,
  \end{align*}
  where in the second line we used that $\omega\circ\emptyset_{d_1}=\omega\circ\emptyset_{d_2-d_1-1}=\omega\circ\emptyset_{n-d_2-1}=1$.
  If we now put all these results together we observe that $\sigma_n\circ S\circ (d_1d_2)_n$ equals,
  \begin{align*}
    -(12-1-2+\nu)+(12^=-1-12-2)-12^==-\nu,
  \end{align*}
  which corresponds to `$-1$' once we convert back the shorthand, i.e. \/ $\sigma_n\circ S\circ (d_1d_2)_n=-1$.
\end{example}
\begin{example}\label{ex:main3}
  We now compute $\sigma_n\circ S\circ (d_1d_2d_3)_n$. This case is particularly insightful and the strategy we employ for the general case in Theorem~\ref{thm:Satocoeffsdescents}
  becomes apparent. Here, we have $S=-J+J\cendot\id-J\cendot\id^{\cendot2}+J\cendot\id^{\cendot3}$. As in Example~\ref{ex:main2}, the first term is given by,
  \begin{align*}
    \sigma_n\circ J\circ (d_1d_2d_3)_n=&\;\omega\circ (d_1d_2d_3)_n\\
                                     =&\;123-12-13-23+1+2+3-\nu,
  \end{align*}
  using Example~\ref{ex:weightlength3}. For the second term, we note that,
  \begin{align*}
    \Delta\circ(d_1d_2d_3)_n=&\;\emptyset_0\otimes(d_1\!-\!1)(d_2\!-\!1)(d_3\!-\!1)_{n-1}\\
                               &\;+\emptyset_{d_1}\otimes(d_2\!-\!d_1\!-\!1)(d_3\!-\!d_1\!-\!1)_{n-d_1-1}\\
                               &\;+(d_1)_{d_2}\otimes(d_3\!-\!d_2\!-\!1)_{n-d_2-1}\\
                               &\;+(d_1d_2)_{d_3}\otimes\emptyset_{n-d_3-1}.
  \end{align*}
  Thus, using identity~\eqref{eq:prod-prod}, we have, 
  \begin{align*}
    \sigma_n\circ J&\cendot\id\circ (d_1d_2d_3)_n\\
                 =&\;C_{d_1}^n\,\omega\circ(d_2-d_1-1)(d_3-d_1-1)_{n-d_1-1}\\
                  &\;+C_{d_2}^n\,\omega\circ(d_1)_{d_2}\,\omega\circ(d_3-d_2-1)_{n-d_2-1}\\
                  &\;+C_{d_3}^n\,\omega\circ(d_1d_2)_{d_3}\\
                 =&\;C_{d_1}^n\,\bigl(C_{d_2-d_1-1}^{n-d_1-1} C_{d_3-d_2}^{n-d_2}-C_{d_2-d_1-1}^{n-d_1-1}-C_{d_3-d_1-1}^{n-d_1-1}+1\bigr)\\ 
                  &\;+C_{d_2}^n\,\bigl(C_{d_1}^{d_2}-1\bigr)\bigl(C_{d_3-d_2-1}^{n-d_2-1}-1\bigr)\\ 
                  &\;+C_{d_3}^n\,\bigl(C_{d_1}^{d_3} C_{d_2-d_1}^{d_3-d_1}-C_{d_1}^{d_3}-C_{d_2}^{d_3}+1\bigr)\\             
                 =&\;\bigl(C_{d_1}^nC_{d_2-d_1-1}^{n-d_1-1}C_{d_3-d_2}^{n-d_2}-C_{d_1}^nC_{d_2-d_1-1}^{n-d_1-1}\\
                  &\;-C_{d_1}^nC_{d_3-d_1-1}^{n-d_1-1}+C_{d_1}^n\bigr)\\
                  &\;+\bigl(C_{d_1}^nC_{d_2-d_1}^{n-d_1}C_{d_3-d_2-1}^{n-d_2-1}-C_{d_1}^nC_{d_2-d_1}^{n-d_1}\\
                  &\;-C_{d_2}^nC_{d_3-d_2-1}^{n-d_2-1}+C_{d_2}^n\bigr)\\
                  &\;+\bigl(C_{d_1}^nC_{d_2-d_1}^{n-d_1}C_{d_3-d_2}^{n-d_2}-C_{d_1}^nC_{d_3-d_1}^{n-d_1}\\
                  &\;-C_{d_2}^nC_{d_3-d_2}^{n-d_2}+C_{d_3}^n\bigr)\\
                 =&\;(12^=3-12^=-13^=+1)+(123^=-12-23^=+2)\\
                  &\;+(123-13-23+3).               
  \end{align*}
  For the third term, we have,
  \begin{align*}
    \Delta_2\circ(d_1&d_2d_3)_n\\
          =&\;\Delta\circ\emptyset_{d_1}\otimes(d_2-d_1-1)(d_3-d_1-1)_{n-d_1-1}\\
                              &\;+\Delta\circ(d_1)_{d_2}\otimes(d_3-d_2-1)_{n-d_2-1}\\
                              &\;+\Delta\circ(d_1d_2)_{d_3}\otimes\emptyset_{n-d_3-1}\\
                             =&\;\emptyset_0\otimes\emptyset_{d_1-1}\otimes(d_2-d_1-1)(d_3-d_1-1)_{n-d_1-1}\\
                              &\;+\emptyset_{0}\otimes(d_1-1)_{d_2-1}\otimes(d_3-d_2-1)_{n-d_2-1}\\
                              &\;+\emptyset_{d_1}\otimes\emptyset_{d_2-d_1-1}\otimes(d_3-d_2-1)_{n-d_2-1}\\
                              &\;+\emptyset_0\otimes(d_1-1)(d_2-1)_{d_3-1}\otimes\emptyset_{n-d_3-1}\\
                              &\;+\emptyset_{d_1}\otimes(d_2-d_1-1)_{d_3-d_2-1}\otimes\emptyset_{n-d_3-1}\\
                              &\;+(d_1)_{d_2}\otimes\emptyset_{d_3-d_2-1}\otimes\emptyset_{n-d_3-1}.
  \end{align*}
  Thus, using that $\omega\circ\emptyset_d=0$, we observe that,
  \begin{align*}
    \sigma_n\circ J\cendot&\id^{\cendot2}\circ (d_1d_2d_3)_n\\
           =&\;C_{d_1}^{n}C_{d_2-d_1-1}^{n-d_1-1}(C_{d_3-d_2-1}^{n-d_2-1}-1)\\
           &\;+C_{d_1}^{n}C_{d_3-d_1-1}^{n-d_1-1}(C_{d_2-d_1-1}^{d_3-d_1-1}-1)\\
           &\;+C_{d_2}^{n}C_{d_3-d_2-1}^{n-d_2-1}(C_{d_1}^{d_2}-1)\\
                  =&\;(12^=3^=-12^=)+(12^=3-13^=)+(123^=-23^=),
  \end{align*}
  where the identity~\eqref{eq:prod-prod} is required in a couple of places in the last step to obtain ``Standard Form''.
  The fourth term is straightforwardly computed. Observing the expression for $\Delta_2\circ(d_1d_2d_3)_n$, we see that if we degraft the left factors by applying $\Delta$ to them,
  then the terms with left factors $\emptyset_0$ are trivialised when we apply $\Delta$.
  The terms shown with the left factors $\emptyset_{d_1}$ generate a four-tensor term with left factors $\emptyset_{0}$ after we apply $\Delta$, and they are trivialised when we apply $J\otimes\id^{\otimes 3}$.
  The only relevant term, here, is thus, $\emptyset_{d_1}\otimes\emptyset_{d_2-d_1-1}\otimes\emptyset_{d_3-d_2-1}\otimes\emptyset_{n-d_3-1}$. 
  Thus, we have,
  \begin{equation*}
    \sigma_n\!\circ\! J\!\cendot\!\id^{\cendot3}\!\circ (d_1d_2d_3)_n=C_{d_1}^{n}C_{d_2-d_1-1}^{n-d_1-1}C_{d_3-d_2-1}^{n-d_2-1}=12^=3^=.
  \end{equation*}
  Putting all the results together we observe that,
  \begin{align}
    \sigma_n&\circ S\circ (d_1d_2d_3)_n\nonumber\\
           =&\;-(123-12-13-23+1+2+3-\nu)\nonumber\\
            &\;+(12^=3-12^=-13^=+1)\nonumber\\
            &\;+(123^=-12-23^=+2)\nonumber\\
            &\;+(123-13-23+3)\nonumber\\
            &\;-(12^=3^=-12^=)-(12^=3-13^=)-(123^=-23^=)\nonumber\\
            &\;+12^=3^=.\label{eq:siglength3}
  \end{align}
  We observe that all the terms on the right above cancel except for $\nu$, and hence we deduce, $\sigma_n\circ S\circ (d_1d_2d_3)_n=1$.
\end{example}
\begin{remark}\label{rmk:mainthmstrategy}
 We observe that on the right in \eqref{eq:siglength3}, cancellations occur locally between neighbouring terms corresponding to individual terms in $S=-J+J\cendot\id-J\cendot\id^{\cendot2}+J\cendot\id^{\cendot3}$.
 In other words, the terms generated by $-J$, corresponding to the first bracketed term on the right, cancel completely with some terms in $J\cendot\id$, corresponding to the next three bracketed terms on the right.
 And then the remaining terms from $J\cendot\id$ cancel completely with some terms from $-J\cendot\id^{\cendot2}$, corresponding to the next three pairs of bracketed terms on the right.
 Finally, we observe that the  one remaining term in $-J\cendot\id^{\cendot2}$, namely $12^=3^=$, cancels with the final term on the right generated by $J\cendot\id^{\cendot3}$. 
 We exploit this observation in our proof of Theorem~\ref{thm:Satocoeffsdescents}, next.  
\end{remark}
\begin{theorem}[Sato coefficient via descents]\label{thm:Satocoeffsdescents}
  In the case $k=n+1$ in Lemma~\ref{lemma:Satocoeffsviadescents}, the push forward of the Sato coefficient $b^\ast_{-1}$ by $\sigma_n\botimes\id$ is given by,
  \begin{equation}\label{eq:finalSatocoeff}
        b_{-1}=-\sum_{\|\df\|=n}(-1)^{|\df|}\botimes\df,
  \end{equation}
  or equivalently, $\sigma_n\circ S\circ\df=(-1)^{|\df|}$.
\end{theorem}
\begin{proof}
We proceed systematically through the terms generated by each endomorphism in $S$, namely $-J$, $J\cendot\id$, $-J\cendot\id^{\cendot2}$, etc. 
As indicated in Remark~\ref{rmk:mainthmstrategy}, we compare the terms generated by each endomorphism with the next, and keep track of the cancellation of terms.
For convenience here, we use the following notation for ordered subsets.
For a generic ordered subset, say, $\alpha_1\cdots\alpha_m$, we suppose,
\begin{equation*}
\balpha\coloneqq\alpha_1\cdots\alpha_m\qquad\text{and}\qquad\balpha^=\coloneqq\alpha_1^=\alpha_2\cdots\alpha_m.
\end{equation*}
We use $|\balpha|$ to denote the length of the ordered subset, so for example, $|\alpha_1\cdots\alpha_m|=m$.
Recall that $J\circ\emptyset_0=0$ and $\Delta\circ\emptyset_0=0$.
Hereafter, $\df=(d_1d_2\cdots d_\ell)_n$ denotes a generic descent in $\mathbb D$.
To begin, using Theorem~\ref{thm:weightformula}, we know that,
\begin{equation}\label{eq:mainstep1}
  \sigma_n\circ(-J)\circ\df=-\omega\circ\df=-\sum(-1)^{\ell+|\balpha|}\balpha,
\end{equation}
where the sum is over all ordered subsets $\balpha\in[1,\ldots,\ell]$.

Next, we consider the terms generated by $\sigma_n\circ J\cendot\id\circ\df$. We observe that,
\begin{align*}
  \Delta\circ\df=\sum_{k=1}^\ell&\Bigl((d_1\cdots d_{k-1})_{d_k}\\
                               &\!\otimes\!(d_{k+1}-d_k-1)\cdots(d_\ell-d_k-1)_{n-d_k-1}\Bigr),
\end{align*}
where we have ignored the first term starting with a $\emptyset_0$ left factor as this is eliminated in the next step.
Indeed, when computing $\sigma_n\circ J\cendot\id\circ\df$, the first term just mentioned is eliminated, and within the sum above,
we generate the three factors $\hat{\gamma}_n\bigl(d_k(n-d_k-1)\bigr)=C_{d_k}^n$, $\omega\circ(d_1\cdots d_{k-1})_{d_k}$ and $\omega\circ(d_{k+1}-d_k-1)\cdots(d_\ell-d_k-1)_{n-d_k-1}$.
Proceeding as in Step~2 in the proof of Theorem~\ref{thm:weightformula}, we observe that $\sigma_n\circ J\cendot\id\circ\df$ is given by,
\begin{equation}
  \sum_{k=1}^\ell\sum(-1)^{\ell+|\balpha|+|\bbeta|-1}\balpha k\bbeta^=,\label{eq:mainstep2a}
\end{equation}
where the second sum is over all $\balpha\in[1,\ldots,k-1]$ and $\bbeta\in[k+1,\ldots,\ell]$.
Note that the binomial coeeficient corresponding to the first term $\beta_1^=$ in $\bbeta^=$, also has `$d_k$' subtracted from both its arguments.
Consider the case when $\bbeta=\emptyset$, the empty ordered subset. The terms in \eqref{eq:mainstep2a} corresponding to this case are,
\begin{equation}\label{mainstep2b}
  \sum_{k=1}^\ell\sum(-1)^{\ell+|\balpha|-1}\balpha k.
\end{equation}
We note that we have the following natural decomposition,
\begin{equation*}
[1,\ldots,\ell]\backslash\{\emptyset\}=1\cup[1]2\cup[1,2]3\cup\cdots\cup[1,\ldots,\ell-1]\ell.
\end{equation*}
Hence, except for the single term in \eqref{eq:mainstep1} corresponding to $\balpha=\emptyset$, the terms in \eqref{mainstep2b} exactly match those in \eqref{eq:mainstep1}.
With care taken with signs, these terms exactly cancel. We are thus left with the single term from \eqref{eq:mainstep1} mentioned, and the terms from \eqref{eq:mainstep2a}, for which,
\begin{equation}\label{eq:mainconds2}
  \bbeta\in[k+1,\ldots,\ell]\backslash\{\emptyset\}\quad\text{and}\quad k\in1,\ldots,\ell-1,
\end{equation}
where the former necessitates the latter.

We now consider the terms generated by $\sigma_n\circ(-J\cendot\id^{\cendot2})\circ\df$.
To generate $\Delta_2\circ\df$, we simply need to apply $\Delta$ to the left factors in our expression for $\Delta\circ\df$ above.
This eliminates the first term in $\Delta\circ\df$, and means we need to compute $\Delta\circ(d_1\cdots d_{k-1})_{d_k}$ for which can use the formula for $\Delta\circ\df$ itself, replacing $n$ by $d_k$ and $\ell$ by $d_{k-1}$.
Proceeding in this manner, we find that,
\begin{align*}
\Delta_2\circ\df&=\sum_{k_1=1}^\ell\sum_{k_2=1}^{k_1-1}\Bigl((d_1\cdots d_{k_2-1})_{d_{k_2}}\\
&\otimes(d_{k_2+1}-d_{k_2}-1)\cdots(d_{k_1-1}-d_{k_2}-1)_{d_{k_1}-d_{k_2}-1}\\
&\otimes(d_{k_1+1}-d_{k_1}-1)\cdots\cdot(d_\ell-d_{k_1}-1)_{n-d_{k_1}-1}\Bigr).
\end{align*}
where, again, we have ignored all terms starting with a $\emptyset_0$ left factor.
Thus, using $\Delta_2\circ\df$, computing $\sigma_n\circ(-J\cendot\id^{\cendot2})\circ\df$, generates the following four factors within the double sum involving $k_1$ and $k_2$. 
The first factor is,
\begin{equation*}
\hat{\gamma}_n\bigl(d_{k_2}(d_{k_1}-d_{k_2}-1)(n-d_{k_1}-1)\bigr)=C_{d_{k_2}}^nC_{d_{k_1}-d_{k_2}-1}^{n-d_{k_2}-1},
\end{equation*}
while the other three factors are the weights of each of the three tensored terms shown in the double sum for $\Delta_2\circ\df$ shown above.
We distribute the two factors from $\gamma_n$ above across the first two weights.
We observe that, by an exactly analogous calculation to that in Step~2 in the proof of Theorem~\ref{thm:weightformula}, we have,
\begin{equation}\label{eq:mainstep3a}
    C_{d_{k_2}}^n\omega\circ(d_1\cdots d_{k_2-1})_{d_{k_2}}=\sum(-1)^{k_2-1+|\balpha|}\balpha\,k_2.
\end{equation}
The sum is over all ordered subsets $\balpha\in[1,\ldots,k_2-1]$. Next, again by an analogous calculation, we find,
\begin{align}
  C_{d_{k_1}-d_{k_2}-1}^{n-d_{k_2}-1}\omega\!\circ\!(d_{k_2+1}&\!-\!d_{k_2}\!-\!1)\!\cdot\!\cdot\!\cdot\!(d_{k_1-1}\!-\!d_{k_2}\!-\!1)_{d_{k_1}-d_{k_2}-1}\nonumber\\
  =&\sum(-1)^{k_1-k_2-1+|\bbeta|}\bbeta^=\,k_1.\label{eq:mainstep3b}
\end{align}
The sum is over $\bbeta\in[k_2+1,\ldots,k_1-1]$.
Some care is required here, and we remark that the first term $\beta_1^=$ in $\bbeta^=$ is actually, $C_{d_{\beta_1}-d_{k_2}-1}^{n-d_{k_2}-1}$, as we should anticipate.
Lastly, the final weight term is given exactly by the same calculation as that for the corresponding term in Step~2 of the proof of Theorem~\ref{thm:weightformula}, generating,
\begin{equation}\label{eq:mainstep3c}
    \sum(-1)^{\ell-k_1+|\bgamma|}\bgamma^=.
\end{equation}
The sum is over all $\bgamma\in[k_1+1,\ldots,\ell]$, and similar care is required in interpreting the first term in $\bgamma^=$, as we outlined for $\bbeta^=$ just above. 
Putting the factors \eqref{eq:mainstep3a}, \eqref{eq:mainstep3b} and \eqref{eq:mainstep3c} together, we observe that $\sigma_n\circ(-J\cendot\id^{\cendot2})\circ\df$ equals,
\begin{equation}\label{eq:mainstep3}
   -\sum_{k_1=2}^\ell\sum_{k_2=1}^{k_1-1}\sum(-1)^{\ell+|\balpha|+|\bbeta|+|\bgamma|}\balpha\,k_2\bbeta^=\,k_1\bgamma^=.
\end{equation}
The third sum shown is over all $\balpha\in[1,\ldots,k_2-1]$, $\bbeta\in[k_2+1,\ldots,k_1-1]$ and $\bgamma\in[k_1+1,\ldots,\ell]$.
Now consider the case when $\bgamma=\emptyset$, in which case \eqref{eq:mainstep3} collapses to,
\begin{equation*}
   -\sum_{k_1=2}^\ell\sum_{k_2=1}^{k_1-1}\sum(-1)^{\ell+|\balpha|+|\bbeta|}\balpha\,k_2\bbeta^=\,k_1.
\end{equation*}
If we swap the order of the $k_1$ and $k_2$ sums, and then swap the $k_1$ and $k_2$ labels, we obtain,
\begin{equation}\label{eq:mainstep3d}
   -\sum_{k_1=1}^{\ell-1}\sum_{k_2=k_1+1}^{\ell}\sum(-1)^{\ell+|\balpha|+|\bbeta|}\balpha\,k_1\bbeta^=\,k_2,
\end{equation}
where now $\balpha\in[1,\ldots,k_1-1]$ and $\bbeta\in[k_1+1,\ldots,k_2-1]$.
We note that in \eqref{eq:mainstep3d}, the subsum over $k_2$ and $\bbeta$ of $\bbeta^=\,k_2$ equals the sum of $\hat{\bbeta}$
over all $\hat{\bbeta}\in[k_1+1,\ldots,\ell]\backslash\{\emptyset\}$. This is due to the decomposition of $[k+1,\ldots,\ell]\backslash\{\emptyset\}$ into,
\begin{equation}\label{eq:decomp}
(k+1)\cup[k+1](k+2)\cup\cdots\cup[k+1,\ldots,\ell-1]\ell.
\end{equation}
Hence \eqref{eq:mainstep3d} becomes,
\begin{equation}\label{eq:mainstep3e}
   -\sum_{k_1=1}^{\ell-1}\sum(-1)^{\ell+|\balpha|+|\bbeta|-1}\balpha\,k_1\hat{\bbeta}^=,
\end{equation}
where the second sum is over all $\balpha\in[1,\ldots,k_1-1]$ and $\hat{\bbeta}\in[k_1+1,\ldots,\ell]\backslash\{\emptyset\}$.
This exactly matches and cancels with the terms remaining from \eqref{eq:mainstep2a} with the conditions~\eqref{eq:mainconds2}.
We are thus left with the terms in \eqref{eq:mainstep3}, restricted so that,
\begin{equation}\label{eq:mainconds3}
\bgamma\in[k_1+1,\ldots,\ell]\backslash\{\emptyset\}\quad\text{and}\quad k_1\in2,\ldots,\ell-1. 
\end{equation}

As a last case, we consider $\sigma_n\circ(-J\cendot\id^{\cendot3})\circ\df$. The pattern thereafter becomes straightforward. 
To compute $\Delta_3\circ\df$, we must apply $\Delta$ to the left factors in the expression $\Delta_2\circ\df$ above.
Ignoring the term with the left factor $\emptyset_0$, we find,
\begin{align*}
  \Delta_3\circ\df&=\sum_{k_1=1}^\ell\sum_{k_2=1}^{k_1-1}\sum_{k_3=1}^{k_2-1}\Bigl((d_1\cdots d_{k_3-1})_{d_{k_3}}\\
                  &\otimes(d_{k_3+1}-d_{k_3}-1)\cdots(d_{k_2-1}-d_{k_3}-1)_{d_{k_2}-d_{k_3}-1}\\
                  &\otimes(d_{k_2+1}-d_{k_2}-1)\cdots(d_{k_1-1}-d_{k_2}-1)_{d_{k_1}-d_{k_2}-1}\\
                  &\otimes(d_{k_1+1}-d_{k_1}-1)\cdots\cdot(d_\ell-d_{k_1}-1)_{n-d_{k_1}-1}\Bigr).
\end{align*}
Using this form to compute $\sigma_n\circ(-J\cendot\id^{\cendot3})\circ\df$ generates five factors within the triple sum,
the first of which is the $\hat{\gamma}_n$ factor, which in this case is,
\begin{equation}
  C_{d_{k_3}}^nC_{d_{k_2}-d_{k_3}-1}^{n-d_{k_3}-1}C_{d_{k_1}-d_{k_2}-1}^{n-d_{k_2}-1}. \label{eq:threefactors}
\end{equation}
The other four factors are the weights of the four tensored terms shown in the triple sum.
We distribute the three binomial coefficients respectively across the first three weights.
Again, exactly analogous to the calculation in Step~2 in the proof of Theorem~\ref{thm:weightformula}, we have,
\begin{equation}\label{eq:mainstep4a}
    C_{d_{k_3}}^n\omega\circ(d_1\cdots d_{k_3-1})_{d_{k_3}}=\sum(-1)^{k_3-1+|\balpha|}\balpha\,k_3.
\end{equation}
The sum is over all $\balpha\in[1,\ldots,k_3-1]$.
The next two factors are computed in the same manner as that above. Indeed, respectively,
the second/third factor in \eqref{eq:threefactors} times the weight of the second/third factor in $\Delta_3\circ\df$ give,
\begin{align}
  &\sum(-1)^{k_2-k_3-1+|\bbeta|}\bbeta^=\,k_2, \label{eq:mainstep4b}
  \intertext{and}
  &\sum(-1)^{k_1-k_2-1+|\bgamma|}\bgamma^=k_1. \label{eq:mainstep4c}
\end{align}
The sums are respectively over $\bbeta\in[k_3+1,\ldots,k_2-1]$ and $\bgamma\in[k_2+1,\ldots,k_1-1]$.
The final weight term generates,
\begin{equation}\label{eq:mainstep4d}
    \sum(-1)^{\ell-k_1+|\bdelta|}\bdelta^=.
\end{equation}
The sum is over $\bdelta\in[k_1+1,\ldots,\ell]$.
Putting the factors \eqref{eq:mainstep4a}--\eqref{eq:mainstep4d} together, we observe, $\sigma_n\circ(-J\cendot\id^{\cendot3})\circ\df$ equals,
\begin{equation}\label{eq:mainstep4}
   \sum_{k_1=3}^\ell\sum_{k_2=2}^{k_1-1}\sum_{k_3=1}^{k_2-1} \sum\rho\,\balpha\,k_3\bbeta^=\,k_2\bgamma^=k_1\bdelta^=,
\end{equation}
where $\rho=(-1)^{\ell+|\balpha|+|\bbeta|+|\bgamma|+|\bdelta|-1}$.
The fourth sum is over all ordered subsets associated with the factors \eqref{eq:mainstep4a}--\eqref{eq:mainstep4d}.
Now consider the case when $\bdelta=\emptyset$, in which case \eqref{eq:mainstep4} collapses to, 
\begin{equation}\label{eq:mainstep4e}
   \sum_{k_1=3}^\ell\sum_{k_2=2}^{k_1-1}\sum_{k_3=1}^{k_2-1} \sum(-1)^{\ell+|\balpha|+|\bbeta|+|\bgamma|-1}\balpha\,k_3\bbeta^=\,k_2\bgamma^=k_1,
\end{equation}
where the fourth sum is now only over the ordered subsets for $\balpha$, $\bbeta$ and $\bdelta$.
If we reverse the order of the $k_1$, $k_2$ and $k_3$, which involves three swaps, and then swap the $k_1$ and $k_3$ labels, we obtain,
\begin{equation}\label{eq:mainstep4f}
   \sum_{k_1=1}^{\ell-2}\sum_{k_2=k_1+1}^{\ell-1}\sum_{k_3=k_2+1}^{\ell} \sum\rho^\prime\balpha\,k_1\bbeta^=\,k_2\bgamma^=k_3,
\end{equation}
where $\rho^\prime=(-1)^{\ell+|\balpha|+|\bbeta|+|\bgamma|-1}$.
The fourth sum is over $\balpha\in[1,\ldots,k_1-1]$, $\bbeta\in[k_1+1,\ldots,k_2-1]$ and $\bgamma\in[k_2+1,\ldots,k_3-1]$.
Again using the decomposition \eqref{eq:decomp}, the form \eqref{eq:mainstep4f} collapses to,
\begin{equation}\label{eq:mainstep4g}
   \sum_{k_1=1}^{\ell-2}\sum_{k_2=k_1+1}^{\ell-1}\sum(-1)^{\ell+|\balpha|+|\bbeta|+|\bgamma|}\balpha\,k_1\bbeta^=\,k_2\hat{\bgamma}^=,
\end{equation}
where the third sum is over the same ranges for $\balpha$ and $\bbeta$, but now $\hat{\bgamma}\in[k_2+1,\ldots,\ell]\backslash\{\emptyset\}$.
If we swap the order of the $k_1$ and $k_2$ sums, and then swap the $k_1$ and $k_2$ labels, in \eqref{eq:mainstep4g},
then it precisely matches \eqref{eq:mainstep3} subject to the conditions \eqref{eq:mainconds3}, with opposite sign.
These two terms thus cancel.
The procedure is now straightforward. We continue in this manner until we reach the final single term $12^=3^=\cdots\ell^=$, which cancels with the corresponding term at the previous level.
There is thus only a single term remaining, namely the `$-(-1)^{\ell}$' term from \eqref{eq:mainstep1} corresponding to the case $\balpha=\emptyset$.
This completes the proof.
\qed
\end{proof}
\begin{remark}
We note that the formula for $b_{-1}$ in Theorem~\ref{thm:Satocoeffsdescents}, after collapsing the `$\botimes$' tensor, can also be succinctly expressed as follows,
\begin{equation*}
   b_{-1}=-\Biggl(\prod_{k=1}^{n-1}(\emptyset-k)\Biggr)_n. 
\end{equation*}
Here, we assume the terms in the product are concatenated to generate the descents of grade $n$.
Thus, for example, in the case $n=4$, we have, $ b_{-1}=-(\emptyset-1)(\emptyset-2)(\emptyset-3)_4=-\emptyset_4+1_4+2_4-12_4+3_4-13_4-23_4+123_4$.
\end{remark}

\section{Solutions}\label{sec:solutions}     
Our main result, Theorem~\ref{thm:Satocoeffsdescents} in Section~\ref{sec:Satocoeffs}, establishes that, given any $n\in\mathbb N$,
the Sato coefficient $b_{-1}$ collapses to the sum over all descents $\df$ of grade $n$, in which the real coefficient of each descent $\df$ therein is $(-1)^{|\df|+1}$.
Our goal herein is to finally close the loop.
We show that \eqref{eq:finalSatocoeff} is equivalent to the statement that the solution to the noncommutative KP equation is generated
by the solution to the GLM equation~\eqref{eq:GLMeqintro}, or equivalently \eqref{eq:GLMPoppe}, in which the semi-additive scattering data is generated by the linearised KP flow.
This is equivalent to showing that \eqref{eq:finalSatocoeff} in Theorem~\ref{thm:Satocoeffsdescents} matches the statement,
\begin{equation}\label{eq:PoppeVkernel}
   \pa_{t_n}\lb V\rb=\lb V\bigl(\pa_{\el}^nP-(-1)^n\pa_{\rl}^nP\bigr)V\rb.
\end{equation}
To establish this connection, we use the $\el\gl\rl$-representation of the descent algebra from Section~\ref{sec:lgr-alg}.

To begin, recall our discussion on generating descents from Remark~\ref{rmk:generatingdescents} in Section~\ref{sec:descent-alg}.
For example, we can generate all the descents at grade $3$ from those at grade $2$, namely $\{\emptyset_2,1_2\}$, by ``stretching'' and ``appending'', as follows.
We simply ``stretch'' the descents $\{\emptyset_2,1_2\}$ by extending the grade to $3$ so that we get $\{\emptyset_3,1_3\}$.
We then ``append'' the descents $\{\emptyset_2,1_2\}$ by tagging onto the end of each descent therein the new possible descent position at grade $3$, to get $\{2_3,12_3\}$.
The union of the results of these two processes is $\{\emptyset_3,1_3,2_3,12_3\}$, which represents all the descents of grade $3$. ``Stretching'' and then ``appending''
the descent set $\{\emptyset_3,1_3,2_3,12_3\}$, then generates the descent set $\{\emptyset_4,1_4,2_4,12_4,3_4,13_4,23_4,123_4\}$ which represents all the descents of grade $4$, and so forth.
Let us now establish this generation procedure in the $\el\gl\rl$-algebra context. Recall that $\{\emptyset_2,1_2\}$ have the $\el\gl\rl$-codes,
\begin{equation*}
\{(\el+\gl)\gl,\gl(\rl+\gl)\}.
\end{equation*}
See, for example, Table~\ref{table:treesupto3}.
Recall from Section~\ref{sec:lgr-alg} that all $\el\gl\rl$-codes consist of factors of $\gl$, $(\el+\gl)$, $(\el+\gl+\rl)$ and $(\rl+\gl)$,
and there are some natural juxtaposition constraints on such factors. In particular, all $\el\gl\rl$-codes either end in the factor $\gl$ or the factor $(\rl+\gl)$.
This is apparent from the $\el\gl\rl$-grafting product in Definition~\ref{def:lgrgrafting}; and the knowledge that, as representatives, all planar binary rooted trees are generated by grafting.
\begin{definition}[Stretching and appending]\label{def:S&A}
  Given a generic $\el\gl\rl$-code $\df$ which we can always decompose into the form $\df=\df^\prime\mathfrak{c}$,
  where $\mathfrak{c}$ is the single rightmost factor such that, either $\mathfrak{c}=\gl$ or $\mathfrak{c}=(\rl+\gl)$.
  Then the \emph{stretch} procedure of $\df$ is given by,
  \begin{equation*}
    \df\mapsto\begin{cases}\df^\prime(\el+\gl)\gl, &~\text{if}\quad\mathfrak{c}=\gl,\\ \df^\prime(\el+\gl+\rl)\gl, &~\text{if}\quad\mathfrak{c}=(\rl+\gl),\end{cases}
  \end{equation*}
  while the \emph{append} procedure is given by, $\df\mapsto \df(\rl+\gl)$.
\end{definition}
With formal ``stretching'' and ``appending'' procedures in hand, we have the following result.
\begin{theorem}[Sato coefficient via $\el\gl\rl$-codes]\label{thm:Satocoeffsvialgr}
  For any $n\in\mathbb N$, the Sato coefficient $b_{-1}$ given in Theorem~\ref{thm:Satocoeffsdescents}
  has the following form in terms of $\el\gl\rl$-codes,
  \begin{equation}\label{eq:Satocoeffsvialgr}
    b_{-1}=-\sum_{k=0}^{n-1}(-1)^k\el^{n-1-k}\gl\rl^k.
  \end{equation}
\end{theorem}
\begin{proof}
  For convenience, we set,
  \begin{equation*}
    \hat{b}_{-1}^{(n)}\coloneqq \sum_{k=0}^{n-1}(-1)^k\el^{n-1-k}\gl\rl^k.
  \end{equation*}
  We know that $\hat{b}_{-1}^{(2)}=\emptyset_2-1_2=\el\gl-\gl\rl$.
  We observe that, we can express $\hat{b}_{-1}^{(n)}$ in the form,
  \begin{equation*}
   \hat{b}_{-1}^{(n)}=\el^{n-1}\gl+\hat{b}_{-1}^{(n-1)}\gl-\hat{b}_{-1}^{(n-1)}(\rl+\gl).
  \end{equation*}
  We now ``stretch'' and ``append'' $\hat{b}_{-1}^{(n)}$. When we ``stretch'' we add to the length of the corresponding descent and we effectively postmultiply by `$-(\rl+\gl)$'.
  This gives,
  \begin{align*}
     \el^{n-1}&(\el+\gl)\gl+\hat{b}_{-1}^{(n-1)}(\el+\gl)\gl-\hat{b}_{-1}^{(n-1)}(\el+\gl+\rl)\gl\\
    -\el^{n-1}&\gl(\rl+\gl)-\hat{b}_{-1}^{(n-1)}\gl(\rl+\gl)+\hat{b}_{-1}^{(n-1)}(\rl+\gl)^2\\
    =&\;\el^{n-1}\el\gl+\el^{n-1}\gl^2+\hat{b}_{-1}^{(n-1)}\rl^2-\el^{n-1}\gl\rl-\el^{n-1}\gl^2\\
    =&\;\hat{b}_{-1}^{(n+1)}.
  \end{align*}
  The statement of the theorem follows by induction.
\qed
\end{proof}
Recall from Section~\ref{sec:lgr-alg} that $\gl$ represents the form $\lb VP_{\hat{1}}V\rb$. 
Further, recall our discussion in Example~\ref{ex:illus} regarding the actions of $\el$ and $\rl$ on $\gl$.
We observe that, in terms of counting the number of left and right derivatives on $\lb V\rb$, we can replace $\gl$ by $(\el+\rl)$.
Indeed we observe the following.
\begin{corollary}[Sato coefficient]\label{cor:Satocoeff}
 For any $n\in\mathbb N$, the Sato coefficient $b_{-1}$ in Theorem~\ref{thm:Satocoeffsvialgr} is given by,
  \begin{equation*}
    b_{-1}=-\bigl(\el^n-(-1)^n\rl^n\bigr).
  \end{equation*} 
\end{corollary}
\begin{proof}
  If we replace $\gl$ in \eqref{eq:Satocoeffsvialgr} by $(\el+\rl)$, we observe that,
  \begin{align*}
    b_{-1}=&\;-\biggl(\el^{n-1}(\el+\rl)+\sum_{k=1}^{n-2}(-1)^k\el^{n-k}\rl^k\\
           &\;\quad+\sum_{\kappa=1}^{n-2}(-1)^{\kappa}\el^{n-1-\kappa}\rl^{\kappa+1}+(-1)^{n-1}(\el+\rl)\rl^{n-1}\biggr)\\
          =&\;-\biggl(\el^{n}+\sum_{k=1}^{n-1}(-1)^k\el^{n-k}\rl^k\\
           &\;\quad+\sum_{\kappa=0}^{n-2}(-1)^{\kappa}\el^{n-1-\kappa}\rl^{\kappa+1}+(-1)^{n-1}\rl^{n}\biggr),    
  \end{align*}
  which gives the result once we cancel the two middle sums on the right.
\qed
\end{proof}
The statement in \eqref{eq:PoppeVkernel} follows directly from Corollary~\ref{cor:Satocoeff}.

\section{Discussion}\label{sec:discussion}
There are many further aspects and future directions of interest to explore.
Some examples are as follows.

(1) \emph{Noncommutative modified KP hierarchy:} For the noncommutative modified KP hierarchy, the GLM equation has the form $P=G(\id-Q)$, with $Q\coloneqq P^2$.
Many more details can be found in Blower and Malham~\cite{BM-KP}.
Therein we establish the Miura transformation between the noncommutative KP equation and a lifted noncommutative modified KP system.
It would be of interest to adapt the methods herein to establish the integrability of the complete noncommutative modified KP hierarchy (lifted version) by direct linearisation,
either via the Miura transformation, or more directly via the modified hierarchy equations.

(2) \emph{Reductive hierarchies:} The noncommutative KP hierarchy reduces to many well-known hierarchies under various symmetry assumptions.
The noncommutative KP hierarchy reduces to the noncommutative Boussinesq hierarchy, and, via a separate symmetry assumption, to the noncommutative KdV hierarchy.
The noncommutative Boussinesq hierarchy is related to the noncommutative Kaup--Kupershmidt and Sawada--Kotera hierarchies; see for example,
Foursov and Moreno Maza~\cite{FM}, Hone and Wang~\cite{HW} and Kang, Liu, Olver and Qu~\cite{KLOQ}.
These hierarchies are in turn related to the Camassa--Holm and Degasperis--Procesi hierarchies via Liouville transformations and so forth.
It would be of interest to understand the consequences, of the direct linearisation solution ansatz and the classes of solutions we have considered herein, for solutions to these reduced hierarchies.

(3) \emph{Numerical simulations:} As mentioned in the introduction, in Blower and Malham~\cite{BM-KP}, we generated solutions to the noncommutative KP equation by numerically solving the direct linearisation problem.
The solution to the linearised equations~\eqref{eq:linearform} can be analytically advanced to at anytime $t>0$ in Fourier space. 
Then the discrete inverse Fourier transform (via FFT) can be substituted as the scattering data into the GLM equation~\eqref{eq:GLMeqintro}.
This in turn can be solved as a numerical linear algebra problem to generate the solution to the noncommutative KP equation at that time $t>0$. 
See Blower and Malham~\cite{BM-KP} for the precise details, which include simulations involving three-wave interactions to the noncommutative KP equations. 
A natural next step is to extend this procedure to higher hierarchy members.

(4) \emph{Quasi-determinant solutions:} It is well-known that solutions to the noncommutative KP hierarchy can be expressed in terms of quasi-determinants,
see Etingof, Gelfand, Retakh~\cite{EGR97,EGR98} and Gilson and Nimmo~\cite{GN}.
Of interest here, would be to establish the explicit relation between the solution ansatz we have used herein and these quasi-determinant solution formulations.
We have established a quasi-trace interpretation for the solutions herein, see Blower and Malham~\cite{BM-KP}.
Investigating this quasi-trace solution form and its connection to quasi-determinant formulations is another interesting future investigation.

(5) \emph{Fredholm Grassmannian flows:} It is well-known that the solution to the KP hierarchy can be interpreted as a flow on a Fredholm Grassmann manifold.
See Sato~\cite{SatoI,SatoII}, Segal and Wilson~\cite{SW}, and Mulase~\cite{Mu84}.
Further, all possible KP multi-soliton interactions can be parametrised within the context of finite dimensional Grassmannians. This classification was established by Kodama~\cite{Kodama}.
From our perspective herein, for a given semi-additive scattering operator $P$ satisfying the linearised KP equations~\eqref{eq:linearform},
the solution $G$ to the GLM equation~\eqref{eq:GLMeqintro} naturally generates a solution flow on a given coordinate patch of a Fredholm Grassmannian.
The coordinate patch is represented by the pair $(\id,G)$,
which parametrises the time-evolving \emph{envelope} or \emph{subspace} of dispersive field solutions represented by the pair $(\id-P,P)$; see Blower and Malham~\cite{BM}.
Also see Doikou, Malham and Stylianidis~\cite{DMS}, Doikou \textit{et al.\/} \cite{DMSWa,DMSWc}, Malham~\cite{MKdV} and Blower and Malham~\cite{BM-KP}, as well as Beck \textit{et al.\/} \cite{BDMSa,BDMSb}.
There are many directions of interest to pursue for solutions to the noncommutative KP hierarchy within the Fredholm Grassmannian context, for example:
(a) Flows with singularities could be interpreted as flows for which, locally, the coordinate patch $(\id,G)$ is a poor representative patch.
We could easily continue such flows as non-singular flows in another coordinate patch, for example, the pair $(G^\prime,\id)$;
(b) Could we extend, improve or view known solution regularity results in this Fredholm Grassmannian context,
for example along the lines of the results of Grudsky and Rybkin~\cite{GR}---see item (7) below;
(c) Can we construct a quasi-determinantal bundle associated with the Fredholm Grassmannian and connect that to the quasi-determinant solutions
established for the noncommutative KP equation by Gilson and Nimmo~\cite{GN}, and/or a suitably defined $\tau$-function?

(6) \emph{Fay identities:} Connecting our results herein to the algebraic and/or differential Fay identities, and theta functions,
is very much of interest, though definitively a longer term project.
See, for example, Dimakis and M\"uller--Hoissen~\cite{DMH} and Takasaki~\cite{Takasaki}.

(7) \emph{Analytical generalisations:} The fact that we restrict ourselves herein to Hilbert--Schmidt valued operators might be too restrictive.
It would be of interest to see if we could relax the class of operators we use throughout to bounded linear operators with sufficiently differentiable kernels, for example.
Weakening the classes of solutions for which solutions to integrable systems can be established, has seen a recent resurgence of interest, see for example Grudsky and Rybkin~\cite{GR}.

(8) \emph{Algebra generalisations:} The following generalisations are very much of interest:
(a) \emph{General fields:} The results we have established herein naturally generalise to the complex field $\mathbb C$.
However they may also generalise to algebra structures in which we replace these fields by quaternions or more general rings.
(b) \emph{Lifted associative algebra:} For any nonassociative algebra, such as the descent algebra $\R\langle\mathbb D\rangle_{\star}$ we considered herein, we can generate, via left and right linear maps,
the \emph{associative enveloping algebra} or \emph{multiplication algebra} which is the subassocciative algebra of the algebra of endomorphisms on $\R\langle\mathbb D\rangle_{\star}$.   
See, for example, Schafer~\cite{Schafer}. Can we lift the descent algebra in this way to simplify our proof? 
(c) \emph{Degrafting coalgebra:} The degrafting operation $\Delta$ introduced in Section~\ref{sec:degrafting} and utilised extensively thereafter, is a natural candidate for a coproduct,
and thus establishing a coalgebra equipped with a degrafting co-product.
Coassociativity of $\Delta$ appears to be the main obstruction, but perhaps a more general perspective might help remove this obstacle and enable us to establish a coalgebra based on $\Delta$;
(d) \emph{Hopf algebra:} If a degrafting coalgebra can be established, then it is only natural to look for a Hopf algebra structure.

\section{Declarations}

\subsection{Funding and conflicts or competing interests}
There are no conflicts of interests or competing interests. 

\subsection{Data availability statement}
No data was used in this work.

\end{document}